\documentclass[11pt]{article}
\usepackage{graphicx}
\usepackage[curve,matrix,arrow,color]{xy}
\usepackage{graphicx}
\usepackage[curve,matrix,arrow,color]{xy}

\usepackage[all,dvips,color]{xypic}
\usepackage[usenames,dvipsnames]{color}
\xyoption{line}

\newcommand{\dd}{\partial}

\usepackage{epsf}
\usepackage{epsfig}
\usepackage{amssymb}
\usepackage{amsthm}
\usepackage{amsmath}
\usepackage{epstopdf}
\usepackage{tikz}
\usepackage[mathscr]{eucal}

\makeatletter
\@addtoreset{equation}{section}
\makeatother

\def\begeqar{\begin{eqnarray}}
\def\endeqar{\end{eqnarray}}
\def\begeq{\begin{equation}}
\def\endeq{\end{equation}}

\def\wgta#1#2#3#4{\hbox{\rlap{\lower.35cm\hbox{$#1$}}
\hskip.2cm\rlap{\raise.25cm\hbox{$#2$}}
\rlap{\vrule width1.3cm height.4pt}
\hskip.55cm\rlap{\lower.6cm\hbox{\vrule width.4pt height1.2cm}}
\hskip.15cm
\rlap{\raise.25cm\hbox{$#3$}}\hskip.25cm\lower.35cm\hbox{$#4$}\hskip.6cm}}

\def\wgtb#1#2#3#4{\hbox{\rlap{\raise.25cm\hbox{$#2$}}
\hskip.2cm\rlap{\lower.35cm\hbox{$#1$}}
\rlap{\vrule width1.3cm height.4pt}
\hskip.55cm\rlap{\lower.6cm\hbox{\vrule width.4pt height1.2cm}}
\hskip.15cm
\rlap{\lower.35cm\hbox{$#4$}}\hskip.25cm\raise.25cm\hbox{$#3$}\hskip.6cm}}

\def\begeqar{\begin{eqnarray}}
\def\endeqar{\end{eqnarray}}

\newcommand{\tensor}{\otimes}

\newcommand{\q}{\mathfrak{q}}
\newcommand{\ffrac}[2]{\mbox{\footnotesize$\displaystyle\frac{#1}{#2}$}}

\newcommand{\TL}[1]{\mathsf{TL}_{#1}}
\newcommand{\JTL}[1]{\mathsf{JTL}^{au}_{#1}}
\newcommand{\rJTL}[1]{\mathsf{JTL}_{#1}}

\newcommand{\ATL}[1]{\mathsf{T}^a_{#1}}

\newcommand{\poless}{\preceq}
\newcommand{\pomore}{\succeq}

\newcommand{\LQG}{U_{\q} s\ell(2)}

\newcommand{\LQGodd}{U^{\text{odd}}_{\q} s\ell(2)}

\newcommand{\half}{%
  \mathchoice{\ffrac{1}{2}}{\frac{1}{2}}{\frac{1}{2}}{\frac{1}{2}}}

\newcommand{\Hilb}{\mathcal{H}}

\newcommand{\veven}[1]{|v^{\text{even}}\rangle}
\newcommand{\vodd}[1]{|v^{\text{odd}}\rangle}

\newcommand{\oN}{\mathbb{N}}
\newcommand{\oC}{\mathbb{C}}
\newcommand{\oZ}{\mathbb{Z}}

\newcommand{\Endo}{\mathrm{End}}

\newcommand{\Hom}{\mathrm{Hom}}

\newcommand{\Vir}{\mathcal{V}}

\newcommand{\VirN}{\mathfrak{vir}\oplus\overline{\mathfrak{vir}}}

\newcommand{\gl}{g\ell}
\newcommand{\sll}{s\ell}

\newcommand{\ExtJTL}{\mathrm{Ext}_{\rule{0pt}{7.5pt}%
{\rJTL{N}}}^1}

\newcommand{\rep}{\mathcal}
\newcommand{\repF}{\rep{F}} 
\newcommand{\repJ}{\rep{J}}

\newcommand{\scrBox}{{\scriptstyle\Box}}

\newcommand{\chVv}{\Hilb_{N}}
\newcommand{\chV}{V}
\newcommand{\repgl}{\pi_{m,n}}

\newcommand{\PrTL}[1]{\mathcal{P}_{#1}}
\newcommand{\StTL}[1]{\mathcal{W}_{#1}}
\newcommand{\IrrTL}[1]{\mathcal{X}_{#1}}
\newcommand{\IrJTL}[2]{[#1,#2]}
\newcommand{\AIrrTL}[2]{\IrJTL{#1}{#2}}
\newcommand{\IrrJTL}[2]{\mathcal{X}_{#1,#2}}
\newcommand{\PrJTL}[2]{\mathcal{P}_{#1,#2}}
\newcommand{\StJTL}[2]{\mathcal{W}_{#1,#2}}
\newcommand{\bStJTL}[2]{\overline{\mathcal{W}}_{#1,#2}}
\newcommand{\TilJTL}[2]{\mathcal{T}_{#1,#2}}

\newcommand{\AStTL}[2]{\mathcal{W}_{#1,#2}}
\newcommand{\bAStTL}[2]{\overline{\mathcal{W}}_{#1,#2}}

\newcommand{\IrCA}[1]{\mathcal{X}_{#1}}
\newcommand{\PrCA}[1]{\mathcal{P}_{#1}}
\newcommand{\StCA}[1]{\mathcal{W}_{#1}}

\newcommand{\atyp}[1]{\{#1\}}
\newcommand{\slPr}{P}

\newcommand{\HXXZ}{H_{\mathrm{XXZ}}(K)}
\newcommand{\PXXZ}{P_{\mathrm{XXZ}}(K)}
\newcommand{\Ket}[1]{\left|#1  \right>}

\newtheorem{Thm}[subsection]{Theorem}

\newtheorem{conj}[subsubsection]{Conjecture}

\theoremstyle{definition}

\newtheorem{rem}[subsubsection]{Remark}

\begin{document}
\begin{center}

\Large{
The periodic $s\ell(2|1)$ alternating spin chain and its continuum limit  as a bulk Logarithmic Conformal Field Theory at $c=0$}

\vskip 1cm

{\large A.M. Gainutdinov $^{a,b}$, N. Read $^{c}$,
H. Saleur $^{d,e}$ and R. Vasseur $^{f,g}$}

\vspace{1.0cm}

{\sl\small $^a$  Fachbereich Mathematik, Universit\"at Hamburg, Bundesstra\ss e 55,
20146 Hamburg, Germany\\}
{\sl\small $^b$
DESY, Theory Group, Notkestrasse 85, Bldg. 2a, 22603 Hamburg, Germany\\}
{\sl\small $^c$\ Department of Physics, Yale University, P.O. Box 208120, New Haven, CT 06520-8120, USA\\}
{\sl\small $^d$  Institut de Physique Th\'eorique, CEA Saclay,
Gif Sur Yvette, 91191, France\\}
{\sl\small $^e$ Department of Physics and Astronomy,
University of Southern California,
Los Angeles, CA 90089, USA\\}
{\sl\small $^f$ Department of Physics, University of California, Berkeley,
 Berkeley, CA 94720, USA\\}
{\sl\small $^g$ Materials Science Division, Lawrence Berkeley National Laboratory,
 Berkeley, CA 94720, USA\\}

\end{center}

\begin{abstract}
The periodic $s\ell(2|1)$ alternating spin chain encodes  (some of) the properties of hulls of percolation clusters, and  is described in the continuum limit by a logarithmic conformal field theory (LCFT) at central charge $c=0$. This theory corresponds to the strong coupling regime of a sigma model on the complex projective superspace $\mathbb{CP}^{1|1} = \mathrm{U}(2|1) / (\mathrm{U}(1) \times \mathrm{U}(1|1))$, and the spectrum of critical exponents can be obtained exactly. In this  paper we push  the analysis further, and  determine the main representation theoretic (logarithmic) features of   this continuum limit by extending to the periodic case the approach of~\cite{ReadSaleur07-2} [N. Read and H. Saleur, Nucl. Phys. B {\bf 777} 316 (2007)]. We first focus on determining the representation theory of the finite size spin chain with respect to the algebra of local energy densities provided by a representation of the affine Temperley-Lieb algebra at fugacity one. We then analyze how these algebraic properties carry over to the continuum limit to deduce the structure of the space of states as a representation over the product of left and right Virasoro algebras. Our main result is the full structure of the vacuum module of the theory, which exhibits Jordan cells of arbitrary rank for the Hamiltonian.
\end{abstract}

\section{Introduction}

Many  applications of conformal field theory to statistical mechanics and condensed matter physics are related to the case of  central charge $c=0$. These applications include
the  statistical properties of critical geometrical models such as self-avoiding walks (polymers) or percolation, and the critical properties of non interacting $2+1$ dimensional disordered electronic systems -- for instance, at the  transition between plateaux in the integer quantum Hall effect.

Unfortunately, apart from some supergroup WZW models, the only well understood $c=0$ conformal field theory is the trivial, minimal and unitary,  one, which contains  no field but the identity. This corresponds to the existence of the trivial modular invariant partition function $Z=1$ at $c=0$, which is the result obtained by calculating the  $n\to 0$ limit of the partition function of $O(n)$ models or replicated disordered systems, or the $Q\to 1$ limit of the Q-state Potts model (see the recent review~\cite{Cardy-rev}). Of course, these geometrical or disordered models have non trivial observables and critical exponents. But to observe them, one needs to understand what is happening  ``outside'' the minimal trivial theory. While this issue was identified rather early~\cite{Saleurold}, it has proven surprisingly hard to control entirely. For instance, despite the huge progress realized in determining exponents for polymers and percolation, including rigorous work and connections with the SLE formalism, there is to this day no agreement -- let alone a consistent proposal -- on what ``the'' proper conformal field theoretic description of say percolation clusters could be. This means in particular that very little is known about four point functions of geometrical observables
in the bulk, despite the well established existence of measurable, universal quantities~\cite{SaleurDerrida}.

An aspect that got rather quickly understood is that the introduction of non trivial observables in say percolation forces one to consider a non unitary conformal field theory. While this is, on general grounds, a rather small price to pay for the description of observables which are most of the time non local in geometrical terms, the non-unitarity of the $c=0$ theory one is after is bound to be rather violent (in contrast say to what happens in the Lee--Yang singularity), and leads in particular to the emergence of \textit{logarithmic features}~\cite{Gurarie}. The adjective logarithmic refers to logarithmic dependence of the four point functions on the harmonic ratio, and to logarithmic terms in the OPEs. It does not  imply that the  field theory is in any way non local,
but rather that the representations of the Virasoro algebra which are involved in the description of the Hilbert space of the theory are not fully reducible. In other words, the action of the dilatation operator $L_0$ is not always diagonalizable~\cite{RozanskySaleur}.

It turns out to be pretty hard to deal with logarithmic conformal field theories in general, and for many years the field has seen but little progress. The difficulty can be traced back to the complexity of Virasoro algebra representation theory once the criterium of unitarity -- or semi-simplicity -- has been relaxed. In fact, the study of LCFTs appears at first sight at least as difficult as the study of non semi-simple Lie algebras, which is proverbially intricate indeed - and plagued most of the time by wilderness issues~\cite{wilderness}.

Nevertheless, the study of WZW models on super groups~\cite{SaleurSchomerus,QSch-rev}, the construction of restricted classes of indecomposable modules~\cite{Rohsiepe,GK1,MatRid}, and the discovery
of deep relations with the theory of associative algebras~\cite{Pearceetal, ReadSaleur07-1,ReadSaleur07-2} have suggested that the problem, however hard, might not be impossible. In the last few years, based in part on the analysis of lattice regularizations and of the deeper role played by quantum groups at roots of unity~\cite{[FGST],[FGST4]}, a lot has been understood about \textit{boundary} LCFTs. There are now reasonable conjectures about the general structure of Virasoro indecomposable modules~\cite{VGJS} and the fusion rules~\cite{[FHST],ReadSaleur07-1,ReadSaleur07-2,MatRid,RP1,GRW,[BGT],GVfusion}, methods to determine the logarithmic couplings (indecomposability parameters)~\cite{DubJacSal, VJSbnumbers}, \textit{etc.}, see the recent reviews~\cite{LCFTlatticereview,Rid-rev}. An important role in these recent developments has been played by algebra, and concepts such as projective and tilting modules, which we review below.

The case of \textit{bulk} LCFT remains however less understood. Indecomposable features now occur not only within the chiral and antichiral sectors, but also in the way they are mixed, and there is evidence that the relationship between bulk and boundary LCFTs is considerably  more intricate than for unitary CFTs \cite{GabRunW,VGJSletter,Ridout}.  This paper is the fourth in a series~\cite{GRS1,GRS2,GRS3} aimed at extending in  the bulk case  a lattice approach that was quite successful in the boundary case. Rather than try to build a $c=0$ LCFT abstractly, we focus on a well defined, local, lattice model, the $s\ell(2|1)$ Heisenberg spin chain with alternating fundamental and conjugate fundamental representations. This chain is closely related -- this is discussed in details below -- with the properties of the \textit{hulls} of percolation clusters. It is gapless, and can be argued to have a conformal invariant continuum limit indeed,
which must have $c=0$ and be logarithmic. In a nutshell, our strategy is to infer as many properties of this LCFT as possible from our analysis of the lattice model.
Despite the fact that we focus on what is, after all, a very simple model, the endeavor remains highly difficult, and our results will be presented in a series of papers.

One of the key ideas at the root of our strategy is that most indecomposable aspects of the LCFT are already present on the lattice, in finite size, the algebra of local energy densities playing the role of a lattice version of the Virasoro (or product of left and right Virasoro) algebra. For the model we consider here (as well as in our previous papers), this algebra is the Temperley Lieb algebra in the boundary case, and the Jones Temperley Lieb algebra in the bulk case. A representation theoretic analysis of these algebras acting on our models is thus a perquisite to our understanding of the logarithmic properties of their continuum limit, and will occupy us in a great deal of the present paper.

We start in section 2 by discussing the percolation problem and  how it is related with conformal field theory. We then focus on properties of cluster boundaries - so called hulls. We recall how they are formulated in terms of an alternating $s\ell(2|1)$ super spin chain, which is expected to flow, in the continuum limit, to a super projective sigma model~\cite{SQHE,ReadSaleur01}. We also remind the reader of basic considerations about modular invariance, observables, and logarithmic features. Section 3 provides reminders on the algebraic description of the $s\ell(2|1)$ spin chain, and some basic representation theoretic properties of the corresponding algebras $\JTL{2L}(m)$ and  $\rJTL{2L}(m)$, both in the generic case, and in the special case $m=1$ of interest here.

In section 4 - which is the most important of this paper - we discuss  the  decomposition of the $s\ell(2|1)$ spin chain over indecomposable representations of  $\rJTL{2L}(1)$. This involves several technical aspects. First, we make the crucial observation --  discussed in more detail in  appendix  \ref{sec:TL-faith} -- that the $s\ell(2|1)$ spin chain provides a \textit{faithful} representation of the algebra $\rJTL{2L}(1)$, in sharp contrast with the $\gl(1|1)$ spin chain studied in~\cite{GRS1,GRS2}, where the corresponding representation was not faithful. The continuum limit of this $\gl(1|1)$ spin chain is the ubiquitous $c=-2$ symplectic fermions theory, a rather simple example of LCFT, which is quite well understood now: accordingly, the indecomposable modules appearing in this spin chain are of very manageable form. We believe this is due to the non faithfulness, and the fact that the $\gl(1|1)$ spin chain sees only a small part of the complexity of the full $\rJTL{2L}(0)$ algebra: this complexity would be revealed in the $\gl(2|2)$ spin chain, whose continuum limit, although also having $c=-2$, is considerable more involved than symplectic fermions. Note that, in contrast with the periodic version, open $\gl(1|1)$ and $s\ell(2|1)$ spin chains both provide faithful representations of the  the ordinary Temperley-Lieb algebra, and the modules appearing in their continuum limit are of similar complexity.

Faithfulness gives us access to powerful tools in the analysis of the spin chain, especially when combined with the fact that $\rJTL{N}(m)$ is a cellular algebra.  This requires manipulating several key concepts of the theory of associative algebras, which are explained in section \ref{sec:indecgen}. In section \ref{sec:proj}, we discuss the structure of projective modules over $\rJTL{N}(1)$. In section \ref{sec:self dual}, we discuss why, the spin chain admitting a non degenerate bilinear form, the representation of $\rJTL{N}(1)$ should be self-dual. This allows us to argue that the building blocks of our spin chains are a special kind of modules -- called \textit{tilting} -- which we introduce in section \ref{sec:tilt}. In section  \ref{secsl21resuts}, we put all these ingredients together to obtain the decomposition of the spin chain. Details on the structure of the tilting modules are provided in section \ref{sec:Jcells}, together with a discussion of the rank of Jordan cells. A remarkable result is that arbitrarily large ranks for the Hamiltonian are encountered as the length of the chain increases.

We then turn to conformal field theory. This is considerably more involved than in the $\gl(1|1)$ case since the $s\ell(2|1)$ spin chain is not free. To start, we focus  in  section \ref{sec:takecontlim} on the generating functions of energy levels, which contain information about the left-right Virasoro $\VirN$ content in the scaling limit. While the chain is not solvable in the usual sense, these generating functions can be obtained by using known results for twisted XXZ spin chains where the same standard modules appear: results for the $s\ell(2|1)$ chain itself are then obtained by using the algebraic analysis. In section \ref{sec:opconsimmod} we then turn to the discussion of the operator content of simple JTL modules at $c=0$. While in the case of the $\gl(1|1)$ spin chain, simples of JTL led, in the scaling limit, to direct sums of simples of $\VirN$, we find  that it is not the case here. We provide the essential features of the left-right Virasoro structure of simple JTL modules in the scaling limit, and reach  in particular the conclusion that the size of Jordan cells in the continuum theory can be even larger than those observed on the lattice. This is a new feature, compared with the boundary and $\gl(1|1)$ cases. We also discuss the field content of our theory up to level $(2,2)$.
In section~\ref{sec:indextiltmod}, we finally come to the discussion of the indecomposable content of the scaling limit of the tilting modules.
We discuss in particular the identity or vacuum module, and the appearance of Jordan cells for $L_0+\bar{L}_0$ of arbitrarily large rank. The last section contains conclusions and directions for future work.

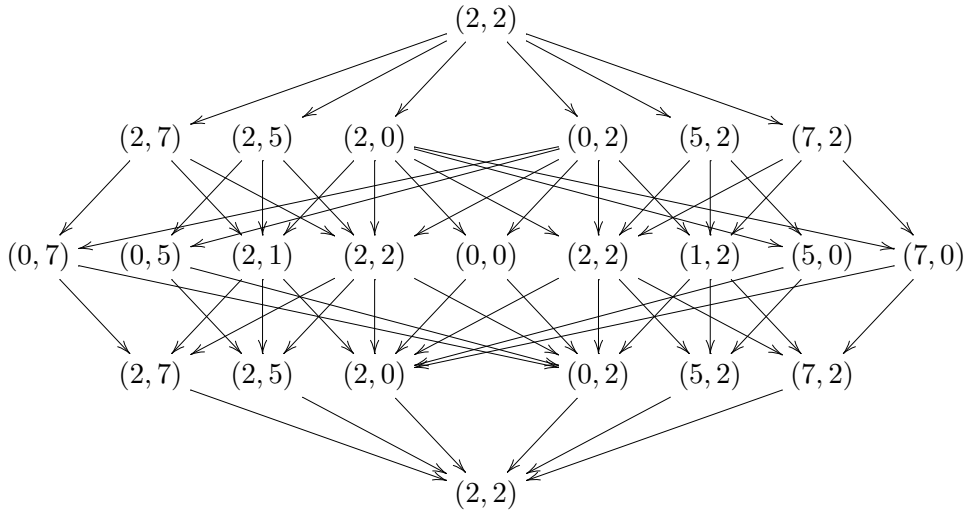
\begin{figure}
\begin{equation*}
  \xymatrix@R=28pt@C=12pt
   { &&&&(2,2)\ar[dl]\ar[dll]\ar[dlll]\ar[dr]\ar[drr]\ar[drrr]&&&&\\
   &(2,7)\ar[dl]\ar[drr]\ar[dr]&(2,5)\ar[d]\ar[dl]\ar[dr]&(2,0)\ar[d]\ar[dl]\ar[dr]\ar[drr]\ar[drrrr]\ar[drrrrr]&&(0,2)\ar[d]\ar[dl]\ar[dll]\ar[dllll]\ar[dlllll]\ar[dr]&(5,2)\ar[d]\ar[dl]\ar[dr]&(7,2)\ar[dl]\ar[dr]\ar[dll]&\\
   (0,7)\ar[dr]\ar[drrrrr]&(0,5)\ar[dr]\ar[drrrr]&(2,1)\ar[d]\ar[dr]\ar[dl]&(2,2)\ar[d]\ar[dl]\ar[dll]\ar[drr]&(0,0)\ar[dl]\ar[dr]&(2,2)\ar[d]\ar[dr]\ar[dll]\ar[drr]&(1,2)\ar[d]\ar[dr]\ar[dl]&(5,0)\ar[dl]\ar[dllll]&(7,0)\ar[dl]\ar[dlllll]\\
   &(2,7)\ar[drrr]&(2,5)\ar[drr]&(2,0)\ar[dr]&&(0,2)\ar[dl]&(5,2)\ar[dll]&(7,2)\ar[dlll]&\\
   &&&&(2,2)&&&&
   }
\end{equation*}
\caption{The ``kernel'' part of the vacuum module of our $c=0$ bulk theory considered as a $\VirN$-subquotient (see section~\ref{sec:indextiltmod}). The nodes $(h,\bar{h})$ denote irreducible $\VirN$ subquotients with conformal weight  $(h,\bar{h})$. The arrows represent the ``irreversible'' action of the algebra $\VirN$: if two nodes $A$ and $B$ are connected by an arrow $A \longrightarrow B$, this means that one can go from $A$ to $B$ by acting with $\VirN$, but not the other way around. We show only the positive and negative modes action, {\it i.e.}, the action of $L_0 + \bar{L}_0$ is not shown explicitly.}\label{Vir-tilt}
\end{figure}

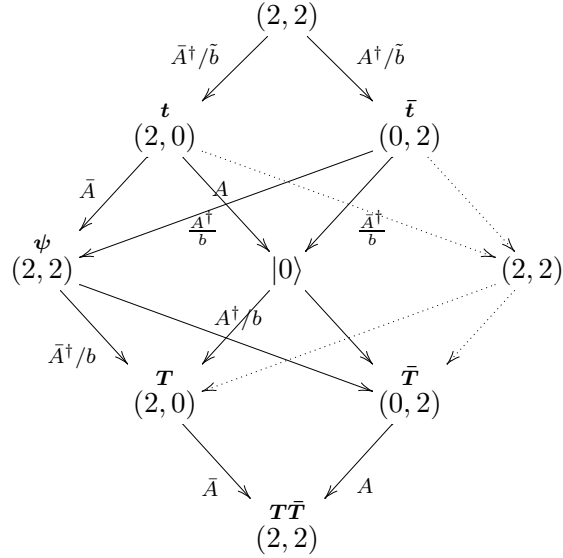
\begin{figure}
\begin{equation}
  \xymatrix@R=26pt@C=18pt@W=2pt@M=2pt
   { &&(2,2)\ar_{\bar{A}^\dagger/\tilde{b}}[dl]\ar^{A^\dagger/\tilde{b}}[dr]&&\\
   &\stackrel{\boldsymbol{t}}{(2,0)}\ar_{\frac{A^\dagger}{b}}[dr]\ar@{.>}[drrr]\ar_{\bar{A}}[dl]&&\stackrel{\boldsymbol{\bar{t}}}{(0,2)}\ar@{.>}[dr]\ar^{\frac{\bar{A}^\dagger}{b}}[dl]\ar_{A}[dlll]&\\
   {\stackrel{\boldsymbol{\psi}}{(2,2)}}\ar_{\bar{A}^\dagger/b}[dr]\ar^{A^\dagger/b}[drrr]&&|0\rangle\ar[dl]\ar[dr]&&(2,2)\ar@{.>}[dl]\ar@{.>}[dlll]\\
   &\stackrel{\boldsymbol{T}}{(2,0)}\ar_{\bar{A}}[dr]&&\stackrel{\boldsymbol{\bar{T}}}{(0,2)}\ar^{A}[dl]&\\
   &&\stackrel{\boldsymbol{T}\boldsymbol{\bar{T}}}{(2,2)}&&
   }
\end{equation}
\caption{The structure of the vacuum module up to level $(h,\bar{h})=(2,2)$. The operators $A$ and $\bar{A}$ are defined as $A=L_{-2}-\frac{3}{2} L_{-1}^{2}$ and $\bar{A}=\bar{L}_{-2}-\frac{3}{2} \bar{L}_{-1}^{2}$, with the corresponding indecomposability parameters $b=\bar{b}=-5$~\cite{VGJSletter}. The irreducible subquotients are simply represented by their conformal weight $(h,\bar{h})$, except for the vacuum $\Ket{0}$ state which has $(h,\bar{h})=(0,0)$. We also show some of the corresponding quantum fields, including the stress energy tensor $\boldsymbol{T}$ and its logarithmic partner  $\boldsymbol{t}$. Note that the vacuum is unique, and satisfies in particular $L_{-1} \Ket{0} = \bar{L}_{-1} \Ket{0}  = 0$.
}\label{fig:vac-tilt-22}
\end{figure}

To help the reader we now provide a summary of our main results, with the indication of where they can be found in the text.

\bigskip

\noindent First, for the lattice:

\begin{itemize}

\item{} The Hilbert space of the spin chain decomposes onto tilting modules $\TilJTL{j}{P}$ that are glueings of standard modules as in Figs.~\ref{fig:tilt-thirdroot} and~\ref{fig:tilt-JTL-vac}. The multiplicities can be obtained, see Sec.~\ref{secsl21resuts} and an example in Sec.~\ref{sec:N6-ex}, but are complicated and not particularly illuminating at this stage. They correspond to combinations of many representations of
$s\ell(2|1)$. For a given tilting, these multiplicities quickly stabilize as $N$ increases, and are the same in the scaling limit.

\item{} The structure of the tilting modules in terms of JTL simples can be obtained. It is also complicated, and depends on $N$. But patterns as $N$ increases can easily be understood. The most detailed analysis is provided  in Fig.~\ref{fig:tilt-N18} for the vacuum ``tilting'' module.

\item{} A consequence of the structure of the modules is the likely appearance of Jordan cells of arbitrarily large size for the
Hamiltonian as $N$ increases. The analysis is described in Sec.~\ref{sec:Jcells}.

\end{itemize}

\noindent Second, in the scaling limit:

\begin{itemize}

\item{}In contrast with the $\gl(1|1)$ case of the open case, simple representations of JTL do not become direct sum of simple representations of $\VirN$ in the scaling limit. In particular, the scaling limit should exhibit Jordan cells which are not present on the lattice. Examples are given in equations~\eqref{VirN-21} and~\eqref{VirN-3q}.

\item{}The indecomposable structure of the tilting modules under $\VirN$ is thus particularly cumbersome. It is worked out in full detail for the vacuum module with the final structure given in Fig.~\ref{Vir-tilt} and with explicit left-right Virasoro action up to conformal weights $(2,2)$ in Fig.~\ref{fig:vac-tilt-22}, where for notations see also Sec.~\ref{sec:indextiltmod}. The complete description of the operator content (including the multi hulls operators) of our theory is in Sec.~\ref{sec:fields}. Several important conclusions follow, among which:

\item{} The identity field occurs with multiplicity one, and satisfies all the properties expected from a well-defined vacuum of a bulk field theory -- in particular, it is invariant under translations.

\item{}There is a Jordan cell of rank two for the fields in $(0,2)$ and $(2,0)$ -- the stress energy tensor has a single logarithmic partner, with indecomposability parameter $b=-5$.

\item{}There is a Jordan cell of rank three for the field $(2,2)$.

\item{}Jordan cells of arbitrarily high rank occur in the scaling limit for large enough conformal weights. These ranks can be calculated, and examples are given in section~\ref{sec:conjrank}.

\end{itemize}

\subsection{Notations}

To help the reader navigate throughout the paper, we provide a partial list of notations (consistent with all other papers in the series):

\begin{itemize}

\item[\mbox{}]$\TL{N}$ --- the (ordinary) Temperley--Lieb algebra,

\item[\mbox{}]$\ATL{N}$ --- the periodic Temperley--Lieb algebra with the translation $u$, or the algebra of affine \\\mbox{}\qquad\quad diagrams,

\item[\mbox{}]$\rJTL{N}$ --- the Jones--Temperley--Lieb algebra,

\item[\mbox{}]$\JTL{2L}(m)$ --- the augmented Jones--Temperley--Lieb algebra,

\item[\mbox{}]$\AStTL{j}{e^{2iK}}$ --- the standard modules over $\rJTL{N}$,

\item[\mbox{}]$\StJTL{j}{P}$ --- the same, with $P=e^{2iK}$,

\item[\mbox{}]$\bAStTL{0}{\q^2}$ --- the standard module over $\rJTL{N}$ for $j=0$,

\item[\mbox{}] $\PrJTL{j}{P}$ ---  projective modules,

\item[\mbox{}] $\TilJTL{j}{P}$  --- tilting modules,

\item[\mbox{}] $\left[ j, e^{2iK}\right]$ or $\IrrJTL{j}{P}$  --- simple modules over $\rJTL{N}$,

 \item[\mbox{}] $F^{(0)}_{j,P}$  --- characters of JTL simples,

\item[\mbox{}] $(h,\bar{h})$ --- simple Virasoro modules,

\item[\mbox{}] $\chi_{h,\bar{h}}$  --- characters of Virasoro simples,

\item[\mbox{}] $\VirN$  --- the direct sum of the left and right Virasoro algebras of central charge $c=0$.

\end{itemize}

Finally, we stress that we adopt the convention that when a module is called `indecomposable module', it  is not irreducible.

\section{The model and some general observations.}

\subsection{Remarks about percolation}

It is important to start by some generalities about geometrical models and conformal field theory. In the case of a critical statistical mechanic model which is defined locally in terms of spins and their interaction, the definition of the conformal field theoretic limit is a priori obvious: what one needs to do is  simply  consider the scaling limit of local observables in order to obtain scaling fields \cite{DiFrancescoMathieuSenechal}. Considering for instance the Ising model, the most natural such observables are the spin and the energy, which, together with the identity, constitute in fact the full primary operator content of what is usually considered as `the' CFT for the Ising model, which is simply the minimal, unitary $c=\frac{1}{2}$ CFT. If one tries to apply the same approach to the problem of percolation, one is forced to recognize that the only local observable is the presence or absence of a bond on a given edge. In percolation however, edges are occupied independently of each other. There is thus no trivial scaling limit to this local observable, and the CFT description of the problem is bound to be the trivial one, with only the identity field, and $c=0$. The obvious point is that, in order to obtain interesting quantities in percolation, one needs to consider observables which are in fact non local in terms of the basic edges occupancy. Part of the history of the field has been to recognize that these observables could also be described by a local field theory, via maneuvers which trade non-locality for non-unitarity - this will be discussed more below. The important point here is that, once non local observables are introduced, it is not at all clear where to stop, nor is it clear which set of such observables one can hope to describe within a consistent CFT.

In percolation, one can first consider the connectivity of clusters, that is, define correlation functions via the probability that a given set of points belongs to the same cluster. The associated conformal weight is known to be $h=\bar{h}={5\over 96}$, and can be formally reproduced using the Kac formula with half integer labels: $h=h_{3/2,3/2}$. Three point functions have been studied recently, and found in numerical agreement with a continuation  of the formula for Liouville theory to the imaginary (time like) domain \cite{DelV,Santachiara}. Nothing is known so far for the four point functions, and there  is no evidence that the continuation of the Liouville theory itself makes sense as a CFT.

Meanwhile, one can consider refined connectivities, for instance via the probability that a set of points not only belongs to the same cluster, but are connected via two non intersecting  paths on this cluster. The corresponding exponent (related to the so called backbone fractal dimension) is not known in closed form. It has been determined numerically to be $h=\bar{h}=.1784\pm 0.0003$ by transfer matrix calculation. It can also be obtained, in principle, within the SLE formalism  and the numerical solution of a differential equation \cite{Lawler}.

To add to the confusion, there are many more geometrical observables one can consider. Of particular interest to us are the properties of percolation hulls. One can indeed consider the probability that two points belong to the {\sl boundary} of the same cluster: the associated conformal weight is known to be $h=\bar{h}={1\over 8}$. It is then easy to generalize this to observables ${\cal O}_k$ where $k$ cluster boundaries come together, with the exponent $h_k=\bar{h}_k={4k^2-1\over 24}$ \cite{DuplSal}.

When discussing the possible LCFT description of a problem such as percolation, it is crucial to consider which set of observables one would like to describe, with the understanding that there is probably no single LCFT encompassing them all, since different observables may well not be mutually local. Attempts to define abstractly {\sl the} fusion algebra of percolation without specifying which observables one is interested in appear as pure nonsense to us.

Although the cluster connectivity is the most natural observable to consider, we have unfortunately not been able to make much progress in its bulk conformal description, for technical reasons that we will discuss later. Instead, this paper will focus mostly on attempts at constructing a LCFT describing the properties of hulls, and in particular the observables ${\cal O}_k$.  Before getting into details, we would like to stress one more confusing fact: while in the bulk, the properties of the insides and hulls of clusters are profoundly different, they coincide near  a boundary. This is simply because, near a boundary, a point which belongs to a cluster necessarily also belongs to its hull. Hence, the set of possible geometrical observables in percolation is smaller in the boundary case than in the bulk case, a deep indication of the fact that, for LCFTs, bulk and boundary are not as tightly related as for unitary, rational CFTs \cite{Fjelstad}.

\subsection{The SUSY spin chain}\label{sec:susy}

As anticipated in the foregoing discussion, we will discuss in this paper the LCFT description of the hulls of percolation cluster. The main reason for this is that the hulls, whose definition is initially non local, can be described by a fully local lattice model involving spins with nearest neighbor interaction. The drawback of this model is that it is non unitary -- the Boltzmann weights are not positive definite. The spins take values in representations of the superalgebra $\gl(2|1)$, and the model enjoys the corresponding (target space) supersymmetry.

\begin{figure}[ht]
\begin{center}
    \includegraphics[width=6.6cm]{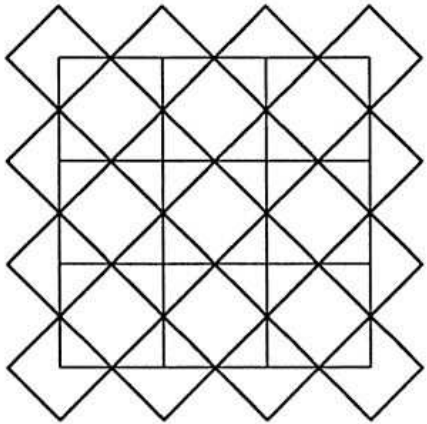}
     \caption{The medial lattice of the square lattice. }\label{medial-lat}
\end{center}
\end{figure}

The first step to obtain our model is to trade the description of percolation in terms of clusters for a description in terms of loops. Geometrically, these loops are obtained via the so called polygon decomposition of the medial  lattice in Fig.~\ref{medial-lat}. For bond percolation on the square lattice - to which we restrict here - the loops live on another square lattice, which they cover entirely.  There is a one to one correspondence between loops and clusters, see Fig.~\ref{loop-cl}.
\begin{figure}
\begin{center}
    \includegraphics[width=7cm]{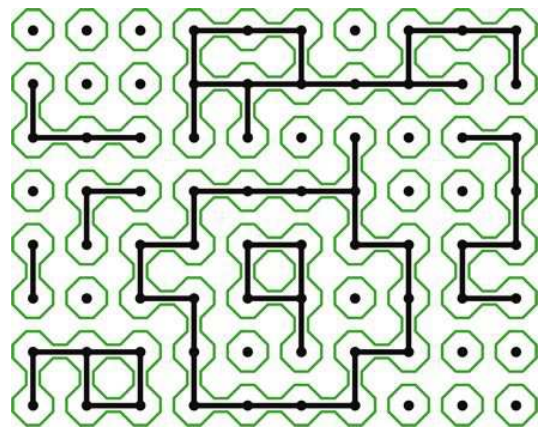}
     \caption{The one to one correspondence between clusters of occupied edges and loop configurations.}\label{loop-cl}
\end{center}
\end{figure}
At criticality, occupied and empty edges occur with probability $p={1\over 2}$. This translates into the fact that all loop configurations are equiprobable, and that the loops must all be counted with a fugacity equal to one. The partition function of the model with doubly periodic boundary conditions is trivial, and can be taken to be  $Z=1$ with the proper normalizations. Observe in particular that the operators ${\cal O}_k$ have a natural interpretation in terms of $k$
 loops joining two points, see Fig.~\ref{Ok}.
\begin{figure}
\begin{center}
    \includegraphics[width=5.8cm]{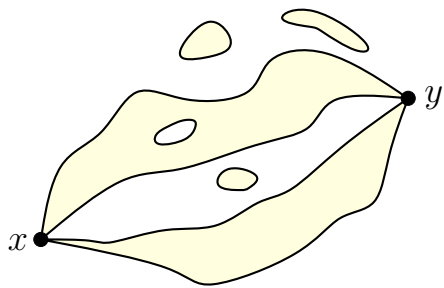}
     \caption{A sketch of the multi hulls operators.}\label{Ok}
\end{center}
\end{figure}

The next step is to consider the loops as  Feynman diagrams expressing contraction of what will turn out to be supergroup variables. The idea is that vertices of the medial lattice represent interactions, while the fugacity of the loops is obtained by a simple counting argument, the resulting weight being the number of bosonic minus the number of fermionic colors (or degrees of freedom).

To give details, it is more convenient to imagine calculating the partition function or some correlation function using a transfer matrix formalism~\cite{SQHE, ReadSaleur01}. We  consider for instance vertical propagation on our figures, and imagine that every edge carries a color $c_i$ which can be bosonic or fermionic. When two edges meet, either there is no interaction, or a contraction takes place. This is described by the two possible diagrams in Fig.~\ref{F-diag}
\begin{figure}
\begin{center}
    \includegraphics[width=5.8cm]{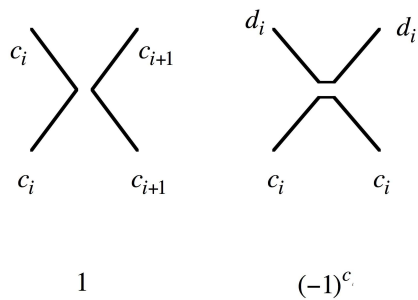}
     \caption{The elementary Boltzmann weights in the interpretation of loops as Feynman diagrams.}\label{F-diag}
\end{center}
\end{figure}
and in the second case, the `Boltzmann weight' is a sign factor $(-1)^{|c_i|}$, where $|c_i|$ is zero if the color is bosonic, and unity if the color is fermionic. It is easy to see that this reproduces the required statistics of the  loops on the medial graph. The elementary interaction encoded by the second diagram in Fig.~\ref{F-diag} corresponds to the action of a Temperley--Lieb algebra generator denoted by $e_i$. This will be discussed in more detail below. First, we need to reformulate the problem slightly in order to make the underlying supersymmetry manifest.
To do so,  we associate with each edge of the medial lattice  a $\mathbb{Z}_2$ graded vector space of dimension $2|1$,
that is, a bosonic ({\it resp.} fermionic) space of dimension $2$ ({\it resp.} $1$).
We choose these vector spaces to {\sl alternate}: we choose   the fundamental $V$ of the Lie superalgebra
$\gl(2|1)$ for even
 edges,
and the dual $V^*$ on odd
 ones, see App.~\ref{appSl21} for definitions. The transfer matrix
then acts on the graded tensor product $\mathcal{H} = (V \otimes V^*)^{\otimes L}$. The distinction between odd and even edges can be interpreted as a choice of {\sl orientation} for the loops, see Fig.~\ref{fig_ConfigPerco}.
\begin{figure}[ht]
\begin{center}
    \includegraphics[width=8cm]{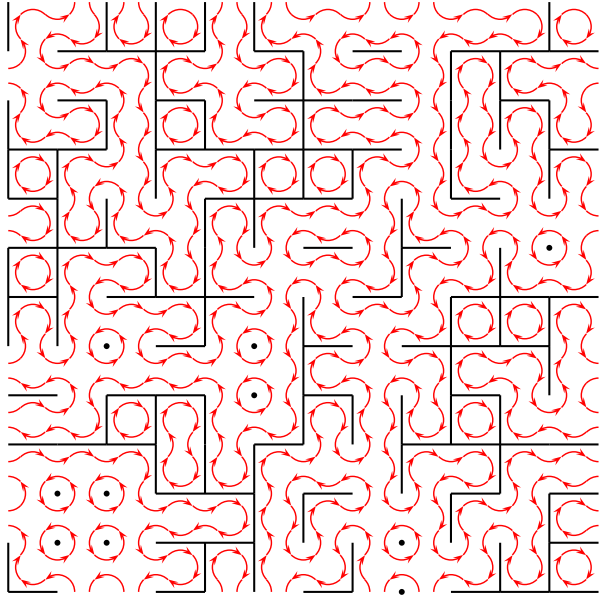}
     \caption{\label{fig_ConfigPerco}
     Example of dense loop configuration obtained as an expansion of the partition
     function of supersymmetric vertex models. We also show the equivalent percolating clusters.
     The lattice consists of alternating arrows going up for $i$ odd and down for $i$ even,
  where $i=1,\dots,N=2L$ corresponds to the horizontal (space) coordinate.
   The alternating $\square,\bar{\square}$ representations
   correspond to a lattice orientation, conserved along each loop.
   The system has periodic boundary conditions in both spacial and imaginary time directions. Each closed loop carries a weight
   $\mathrm{str} \ \mathbb{I} = 1$.
  }
\end{center}
\end{figure}

To proceed, it is convenient to describe the representations $V$ and $V^*$ using oscillators.
For $i$ even we introduce  boson operators $b_i^a$, $b_{ia}^\dagger$, satisfying
$[b_i^a,b_{jb}^\dagger]=\delta_{ij}\delta_b^a$ ($a$, $b=1$, $2$), and fermion operators $f_i^\alpha$,
$f_{i\alpha}^\dagger$,
$\{f_i^\alpha,f_{j\beta}^\dagger\}=\delta_{ij}\delta_\beta^\alpha$, with $\alpha=1$, $\beta=1$, and $i,j = 1,\dots,N$.


For $i$ odd, we have similarly boson
operators $\overline{b}_{ia}$, $\overline{b}_i^{a\dagger}$,
$[\overline{b}_{ia},\overline{b}_j^{b\dagger}]=\delta_{ij}\delta^b_a$
($a$, $b=1$, $2$), and fermion operators
$\overline{f}_{i\alpha}$, $\overline{f}_i^{\alpha\dagger}$,
$\{\overline{f}_{i\alpha},\overline{f}_j^{\beta\dagger}\}=
-\delta_{ij}\delta^\beta_\alpha$ ($\alpha$, $\beta=1$). Notice the minus sign in the last anticommutator;  our
convention is that the $\dagger$ stands for the adjoint, this
minus sign implies that the norm-square of any two states that are
mapped onto each other by the action of a single
$\overline{f}_{i\alpha}$ or $\overline{f}_i^{\alpha\dagger}$ have
opposite signs, and the ``Hilbert'' space has an indefinite inner
product (with the respect to the adjoint operation). The space $V$ is now defined as the subspace of states
that obey the ``one-particle per cite'' constraints
\begin{eqnarray}
\sum_ab_{ia}^\dagger b_i^a+\sum_\alpha f_{i\alpha}^\dagger
f_i^\alpha &=& 1 \quad (i\hbox{ even}),\\
\sum_a\overline{b}_i^{a\dagger}
\overline{b}_{ia}-\sum_\alpha\overline{f}_i^{\alpha\dagger}
\overline{f}_{i\alpha} &=& 1 \quad (i\hbox{ odd}).
\end{eqnarray}
The sums here and below are over $a=1$, $2$, and
$\alpha=1$.

The generators of the Lie superalgebra $\gl(2|1)$ acting on each
site of the chain are the bilinear forms $b_{ia}^\dagger
b_i^b$, $f_{i\alpha}^\dagger f_i^\beta$, $b_{ia}^\dagger
f_i^\beta$, $f_{i\alpha}^\dagger b_i^b$ for $i$ even, and correspondingly
$-\overline{b}_i^{b\dagger} \overline{b}_{ia}$,
$\overline{f}_i^{\beta\dagger} \overline{f}_{i\alpha}$,
$-\overline{f}_i^{\beta\dagger} \overline{b}_{ia}$,
$-\overline{b}_i^{b\dagger} \overline{f}_{i\alpha}$ for $i$ odd,
which for each $i$ have the same (anti-)commutators as those for
$i$ even.

The next step is to build, in this language, the proper interaction to reproduce the statistical properties of percolation hulls. First, we note that
for any two sites $i$ (even), $j$ (odd), the combinations
\begin{equation}
\sum_a  \overline{b}_{ja} b_i^a + \sum_\alpha
\overline{f}_{j\alpha}f_i^\alpha,\quad
\sum_ab_{ia}^\dagger\overline{b}_j^{a\dagger} + \sum_\alpha
f_{i\alpha}^\dagger\overline{f}_j^{\alpha\dagger}
\end{equation}
are invariant under $\gl(2|1)$.  Introduce now
\begin{eqnarray}
e_i=\left(\sum_ab_{ia}^\dagger\overline{b}_{i+1}^{a\dagger} +
\sum_\alpha
f_{i\alpha}^\dagger\overline{f}_{i+1}^{\alpha\dagger}\right)\left(\sum_a \overline{b}_{i+1,a}b_i^a  + \sum_\alpha
\overline{f}_{i+1,\alpha}f_i^\alpha\right)\quad (i\hbox{ even}),\label{ej-1}\\
e_i=\left(\sum_ab_{i+1,a}^\dagger\overline{b}_i^{a\dagger} +
\sum_\alpha
f_{i+1,\alpha}^\dagger\overline{f}_i^{\alpha\dagger}\right)\left(\sum_a \overline{b}_{ia}b_{i+1}^a  + \sum_\alpha
\overline{f}_{i\alpha}f_{i+1}^\alpha\right)\quad (i\hbox{ odd})\label{ej-2}
\end{eqnarray}
The interaction at a vertex where edges $i$ and $i+1$ meet is then simply given,  by the elementary transfer matrix
\begin{equation}
T_i=1+e_i
\end{equation}
The complete transfer matrix is
\begin{equation}
T\equiv T_1T_3\ldots T_0T_2\ldots
\end{equation}

Finally, we note that, although we have formulated everything so far in terms of an isotropic system, it is well known that the same universality class is obtained by choosing different probabilities of occupancy for horizontal and vertical edges. This corresponds to the more general choice of elementary transfer matrices
\begin{eqnarray}
T_i=(1-p_A)+p_A e_i,\quad (i\hbox{ even})\nonumber\\
T_i=p_B+(1-p_B)e_i, \quad (i\hbox{ odd})
\end{eqnarray}
with $p_A+p_B=1$. It is in particular possible to chose a very anisotropic limit  $p_A\rightarrow 0$ where the transfer matrix description is replaced by a hamiltonian description according to $T\approx e^{-p_AH}$ with
\begin{equation}
H=-\sum_i e_i\label{hamil}.
\end{equation}

\subsection{Algebra}

The interpretation of the elementary vertices in terms of contractions suggests the simple algebraic nature of the interaction in our model. Indeed, the elementary generators $e_i$ satisfy special relations, and provide in fact a representation of a well known algebra, the Jones Temperley Lieb algebra. We assume the reader is familiar with the basics: more details will be given in the next section.

An important point is that, while we have restricted so far to $\gl(2|1)$, a more general $\gl(n+1|n)$ model can be introduced, simply by allowing everywhere the labels $a,b$ now to run from $1$ to $n+1$ and $\alpha,\beta$ from $1$ to $n$. Because each loop comes weighted with the super trace of the fundamental (or dual fundamental), it still will come with a factor 1. There is thus, in fact, a multiplicity of spin chains related with percolation. As discussed in~\cite{ReadSaleur07-1} these chains all describe the same geometrical objects, but the associated field theoretic observables come with different multiplicities. This means that  there is in fact an infinite family of LCFTs at $c=0$, with larger and larger degeneracies as $n$ increases. We will, in this paper, restrict to the case $n=1$.

\subsection{The sigma model description.}

This spin chain formulation can be used to obtain a sigma model description of the low energy excitations: see references \cite{Affleck,ReadSachdev} for non supersymmetric examples, and \cite{Zirnbauer,Gruzberg} for supersymmetric examples in the context of disordered systems.   The target space  - which appears in the construction of the coherent state path integral - should be the symmetry supergroup (here $U(2|1)$) modulo the isotropy supergroup of the highest weight state: in our case, this gives  the complex projective superspace
$\mathbb{CP}^{1|1} = \mathrm{U}(2|1) / (\mathrm{U}(1) \times \mathrm{U}(1|1)) $. The mapping could be fully controlled by using a `higher spin' generalization of the spin chain, obtained by taking larger irreducible representations, represented by the value of the highest weight ${\cal S}$, with $2{\cal S}$ integer (with ${\cal S}={1\over 2}$ for the fundamental)  and their duals on alternate sites. For an appropriate family of such representations, the sigma model can be obtained with a bare coupling constant $g_\sigma^2\approx 1/{\cal S}$.  On top of this, for the hamiltonian (\ref{hamil}), there is a topological angle with bare value $\theta=2\pi {\cal S} \hbox{ mod }2\pi$.
 The Lagrangian of the sigma model
involves a multiplet of fields with complex bosonic components $z^a$ ($a=1,2$)
and fermionic components  $\xi^\alpha$ ($\alpha=1$). These fields satisfy
the constraint equation $z_a^{\dag} z^a + \xi_\alpha^{\dag} \xi^\alpha = 1$,
modulo $U(1)$ gauge transformations, so they provide a parametrization of  $\mathbb{CP}^{1|1}$.
In terms of these fields, the euclidian Lagrangian density reads
\begin{equation}
\displaystyle  \mathcal{L} = \frac{1}{2 g_{\sigma}^2} \left[ (D_\mu z_a)^\dag D_\mu z^a +(D_\mu \xi_\alpha)^\dag D_\mu \xi^\alpha \right] + \frac{i \theta}{2 \pi} \varepsilon^{\mu \nu} \partial_{\mu} a_{\nu},
\end{equation}
where $a_\mu = \frac{i}{2}\left(z_a^{\dag} \partial_{\mu} z^a + \xi_\alpha^{\dag} \partial_{\mu} \xi^\alpha -  \partial_{\mu} z_a^{\dag} z^a -  \partial_{\mu} \xi_\alpha^{\dag}  \xi^\alpha \right)$ is a gauge potential and $D_\mu = \partial_\mu + i a_\mu$ is the covariant derivative.
Finally, the beta function for the model is
\begin{equation}
{dg_\sigma^2\over dl}\equiv \beta(g_\sigma^2)=g_\sigma^4+O(g_\sigma^6)
\end{equation}
where $l=\ln L$ is the logarithm of the length scale at which the coupling is defined. The beta function is independent of $\theta$, and the beta function for $\theta$ is zero, in perturbation theory. The beta function being positive at weak coupling, it is expected that the same fixed point CFT will be reached in the universal, large length scale, low energy limit for all spin chains with $2{\cal S}$ an odd integer, and in particular for ${\cal S}={1\over 2}$. This fixed point theory is the LCFT we are trying to build in this paper. While this theory should  have $U(2|1)$ symmetry, it is not expected to be a WZW theory, as confirmed by the early analysis of the spectrum in~\cite{ReadSaleur01}. This is because the general arguments promoting conformal invariance plus continuous group symmetry into a current algebra symmetry fail in non unitary cases, where, typically, logarithms can appear in the OPE of the currents~\cite{Troost}.

\subsection{A note on modular invariance}

The partition function of our model on a torus is $Z=1$. This is the only modular invariant we will associate with our model. This does not mean that the operator content is trivial of course. The point is, that doubly periodic boundary conditions for the geometrical model corresponds to periodic boundary conditions for the bosonic and fermionic degrees of freedom. Turning to a transfer matrix description, the partition function is thus the {\sl supertrace} of the appropriate power of the  transfer matrix, which itself acts on a system of periodic bosons and fermions. While the Hilbert space in which the transfer matrix acts is non trivial, since it is of dimension $\left[(2n+1)^2\right]^L$, this same space has a simple superdimension
\begin{equation}
\hbox{Sdim} {\cal H}=\left[\hbox{Sdim}(V)\hbox{Sdim}(V^*)\right]^L=1.
\end{equation}
This means, as we shall see below, that all operators but the identity appear in non trivial representations of the supersymmetry~\cite{EsslerFS}.

Now, in order to identify the operators present in the model together with their associated representation content, it is convenient to consider a {\sl modified} partition function~\cite{ReadSaleur01}. This partition function is defined simply by taking, instead of the supertrace, the ordinary trace. Of course one has
\begin{equation}
\hbox{dim} {\cal H}=\left[\hbox{dim}(V)\hbox{dim}(V^*)\right]^L=3^{2L}
\end{equation}
and now {\sl each level} of the transfer matrix/hamiltonian will be counted with multiplicity one. The modified partition function is {\sl not modular invariant}, and there is no reason why it should be. There is also no reason why it should be part of a bigger theory encompassing both periodic and antiperiodic boundary conditions in the (imaginary)  time and space directions. The model as we define it on the lattice is perfectly local with periodic boundary conditions, and should lead, as it is,  to a local field theory in the scaling limit.

\subsection{Observables}

\begin{figure}
\begin{center}
    \includegraphics[scale=0.67]{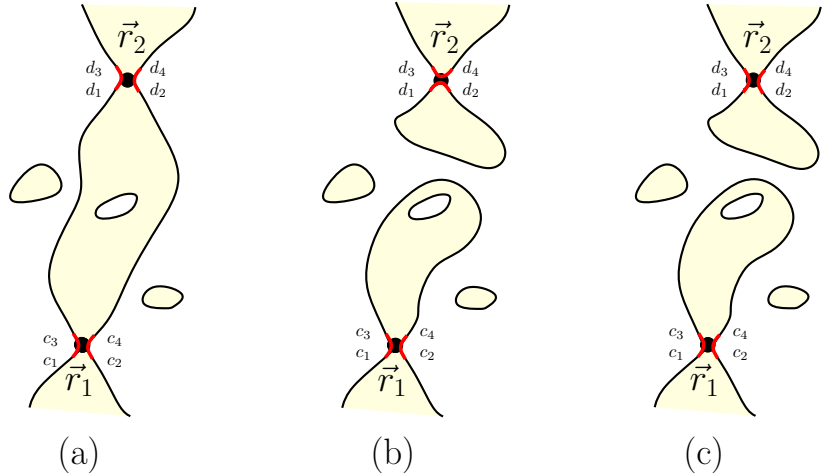}
     \caption{Two-point correlation functions in percolation (see text).}\label{twoptfct}
\end{center}
\end{figure}

Having introduced the supersymmetric formulation of the loop model, it is fair to ask what kind of observables we can now describe. The answer, ultimately, will be provided by the detailed analysis of the spectrum and the operator content given below. It is easy however to see that one should expect first all the multihull operators ${\cal O}_k$. The two point function of the ${\cal O}_2$ operators for instance is expressed geometrically as the probability that there exists (at least) two lines connecting a pair of neighboring edges in $\vec{r}_1$ and a pair of neighboring edges in $\vec{r}_2$. In order to select the appropriate diagrams in the sum over all configurations, all one has to do is to insert the proper terms to prevent contractions of the lines. In other words, the two-point function can be expressed as
\begin{multline}
\Bigl\langle {\cal O}_2(\vec{r}_1){\cal O}_2(\vec{r}_2)\Bigr\rangle=\Bigl\langle \Bigl(\sum_{c_i} \delta(c_1,c_3)\delta(c_2,c_4)-(-1)^{|c_1|}\delta(c_1,c_2)\delta(c_3,c_4)\Bigr)\\
\times\Bigl(\sum_{d_i}\delta(d_1,d_3)\delta(d_2,d_4)-(-1)^{|d_1|}\delta(d_1,d_2)\delta(d_3,d_4)\Bigr)\Bigr\rangle
\end{multline}
where the labels are shown on Fig.~\ref{twoptfct}.
 Consider for instance a diagram such as the one shown on figure \ref{twoptfct}(b). The contraction of the lines forces $d_1=d_2$ (while of course $c_3=c_1$ and $c_4=c_2$). For each  such diagram, the insertion of $1-\delta\delta$ in the two point function subtracts the diagram where the lines are contracted in $\vec{r}_2$. The sum over the labels with the $(-1)^{|d_1|}$ inserted gives the loop thus formed a weight equal to one, so, summing over the rest of the system, the contributions arising from Figs. \ref{twoptfct}(b) and \ref{twoptfct}(c) exactly cancel out. Similar reasonings show that the only diagrams that survive the sum are those where the lines in $\vec{r}_1,\vec{r}_2$ are never contracted as in~\ref{twoptfct}(a)  but simply go through the system. This means in turn that the two points belong to the hull of the infinite percolation cluster.

\begin{figure}[t]
\begin{center}
    \includegraphics[scale=0.25]{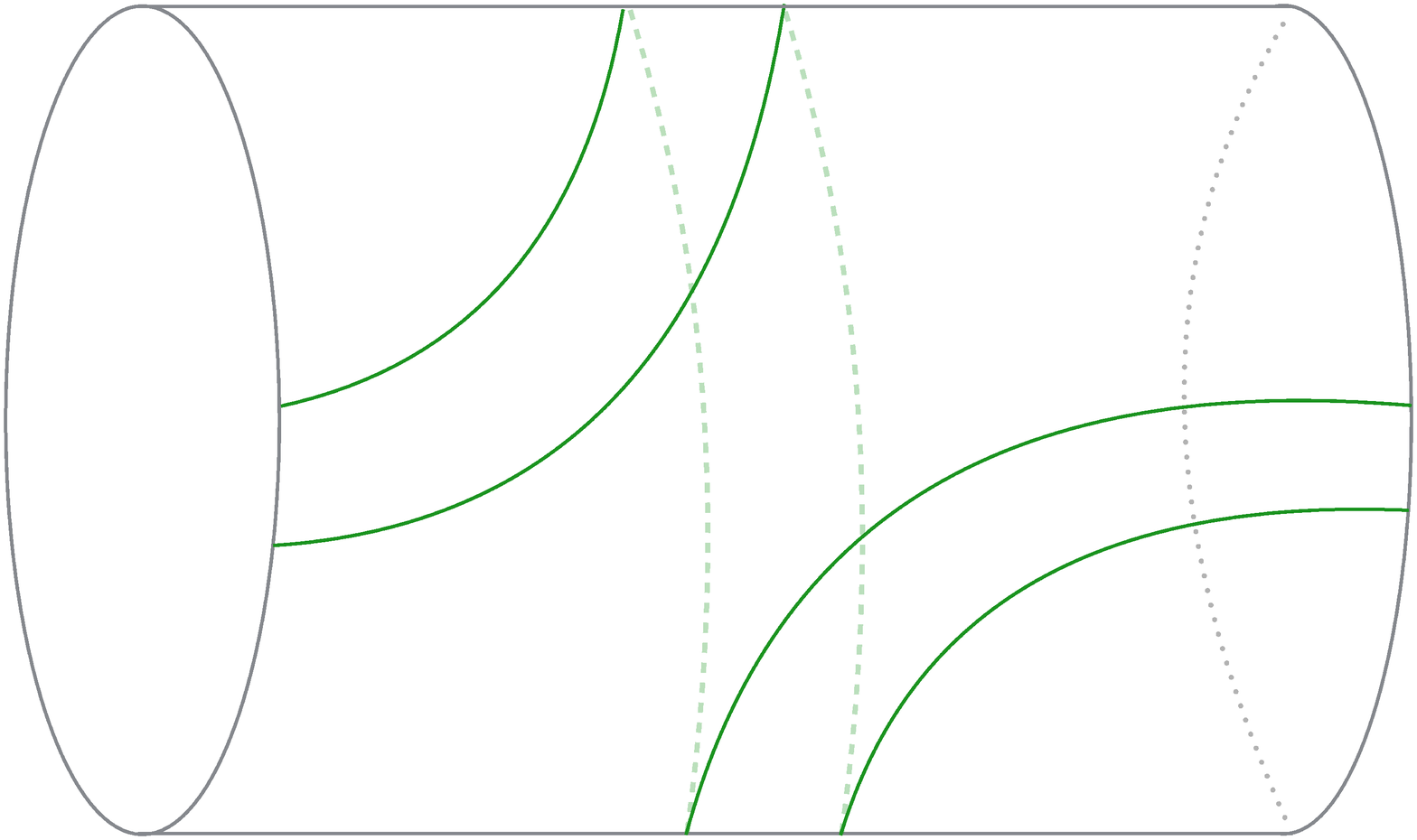}
     \caption{A pair of non contractible loops (a percolation `hull') winding around the axis of the cylinder.}\label{winding}
\end{center}
\end{figure}

It is convenient to think of this  after a conformal map onto the cylinder. Two point functions are then evaluated  in the transfer matrix language. By translation invariance, we see that the eigenvalue in the sector with $2k$ through lines should then give the exponent $(h_k,h_k)$, i.e., zero spin and $h_k=\frac{4k^2-1}{24}$. It is also clear that lines connecting $\vec{r}_1,\vec{r}_2$ can obviously wind around the axis. On the other hand, it is well known, within the Coulomb gas formalism, that primary fields can be obtained, for a given number of through lines, by inserting at either extremity of the cylinder additional charges, whose effect is to give an additional phase factor $z^{\pm 1}$ for every line going clockwise/counterclockwise like in Fig.~\ref{winding}.

 Setting
\begin{equation}
z = e^{i\pi \alpha}
\end{equation}
the critical exponents for configurations of $2k$ through lines are then
\begin{equation}
h_k(\alpha)={(3\alpha-2k)^2-1\over 24},\qquad\bar{h}_k(\alpha)={(3\alpha+2k)^2-1\over 24}.
\end{equation}
A given value of $z$ leads to many choices for $\alpha\hbox{ mod }1$ (the sign of $z$ is not relevant since in the model only pairs of lines propagate), and all the corresponding operators are present in the spectrum.

Another important fact, which occurs only for the percolation model, is that it is impossible to define higher ${\cal O}_k$ operators without some amount of `mixing' with lower ones. This is because of the fact that the object preventing contractions among three lines:
\begin{equation}
X_3=e_i+e_{i+1}-e_ie_{i+1}-e_{i+1}e_i
\end{equation}
while it obeys indeed
\begin{eqnarray}
X_3 e_i=e_i X_3=0,\nonumber\\
X_3 e_{i+1}=e_{i+1}X_3=0
\end{eqnarray}
cannot be normalized to become a projector. In other words, one has
\begin{equation}
X_3^2=0.
\end{equation}
It follows that mixtures between the various sectors with different numbers of non contractible lines must occur producing in the end complicated representations of the periodic TL algebra. This part will be discussed more below.

\section{The algebraic description: generalities}

The following two subsections contain material  discussed already in our earlier work on the subject, which we prefer to reproduce here for clarity, completeness and in order to establish notations. The reader familiar with one of our previous papers~\cite{GRS1,GRS2,GRS3} can go directly to Sec.~\ref{sec:4} where the  crucial aspect of faithfulness is discussed.

\subsection{The Temperley--Lieb algebras in the periodic case}
\label{sec:TL-alg-def}

We define here finite-dimensional quotients of the affine Temperley--Lieb algebra~\cite{MartinSaleur,Jones,MartinSaleur1,GL,Green} denoted here by $\ATL{N}(m)$ and spanned by particular diagrams
on an annulus.  A general basis element in the algebra of
diagrams corresponds to a diagram of $N$ sites on the inner and $N$ on
the outer boundary of the annulus (we will restrict in what follows to the case $N$ even, and parametrize $N=2L$.) The sites are connected in pairs,
and only configurations that can be represented using lines inside the
annulus that do not cross are allowed. We will often call all such
diagrams  \textit{affine diagrams}.  Examples of affine diagrams are shown in Fig.~\ref{fig:aff-diag} for $N=4$, where we draw them in a slightly different geometry: we cut the annulus and transform it to a rectangle, which we call \textit{framing}, with the sites labeled from left to right. An important parameter is the number of
 through-lines, which we denote by $2j$, with $j=0,1,\ldots, L$, connecting
 $2j$ sites on the inner and $2j$ sites on the outer boundary of the
 annulus; the $2j$ sites on the inner boundary we call free or
 non-contractible.   Multiplication $a\cdot b$ of two affine diagrams $a$ and $b$ is defined in a natural
way, by joining an inner boundary of $a$ to an outer boundary of the annulus of $b$, and
removing the interior sites.
Whenever a closed contractible loop is
produced when diagrams are multiplied together, this loop must be
replaced by a numerical factor~$m$. We also note that the  diagrams
in this algebra allow winding of through-lines around the annulus any
integer number of times, and different windings result in independent
algebra elements. Moreover, in the ideal of zero through-lines, any number of
non-contractible loops (like in the fourth diagram in Fig.~\ref{fig:aff-diag}) is allowed.

\begin{figure}
\begin{equation*}
 \begin{tikzpicture}
 	\draw[thick, dotted] (-0.05,0.5) arc (0:10:0 and -7.5);
 	\draw[thick, dotted] (-0.05,0.55) -- (2.65,0.55);
 	\draw[thick, dotted] (2.65,0.5) arc (0:10:0 and -7.5);
	\draw[thick, dotted] (-0.05,-0.85) -- (2.65,-0.85);
	\draw[thick] (0.3,0.5) arc (0:10:20 and -3.75);
	\draw[thick] (2.3,-0.81) arc (0:10:-20 and 3.75);

	\draw[thick] (0.9,0.5) arc (0:10:40 and -7.6);
	\draw[thick] (1.55,0.5) arc (0:10:40 and -7.6);
	\draw[thick] (2.2,0.5) arc (0:10:40 and -7.6);

	\end{tikzpicture}\;\;,
	\qquad\qquad
 \begin{tikzpicture}
 	\draw[thick, dotted] (-0.05,0.5) arc (0:10:0 and -7.5);
 	\draw[thick, dotted] (-0.05,0.55) -- (2.65,0.55);
 	\draw[thick, dotted] (2.65,0.5) arc (0:10:0 and -7.5);
	\draw[thick, dotted] (-0.05,-0.85) -- (2.65,-0.85);
	\draw[thick] (0,0) arc (-90:0:0.5 and 0.5);
	\draw[thick] (0.9,0.5) arc (0:10:0 and -7.6);
	\draw[thick] (1.65,0.5) arc (0:10:0 and -7.6);
	\draw[thick] (2.6,0) arc (-90:0:-0.5 and 0.5);

	\draw[thick] (0.5,-0.8) arc (0:90:0.5 and 0.5);
	\draw[thick] (2.1,-0.8) arc (0:90:-0.5 and 0.5);
	\end{tikzpicture}\;\;,
	\qquad\qquad
 \begin{tikzpicture}
 	\draw[thick, dotted] (-0.05,0.5) arc (0:10:0 and -7.5);
 	\draw[thick, dotted] (-0.05,0.55) -- (2.65,0.55);
 	\draw[thick, dotted] (2.65,0.5) arc (0:10:0 and -7.5);
	\draw[thick, dotted] (-0.05,-0.85) -- (2.65,-0.85);
	\draw[thick] (0,0.1) arc (-90:0:0.5 and 0.4);
	\draw[thick] (0,-0.1) arc (-90:0:0.9 and 0.6);
	\draw[thick] (2.6,-0.1) arc (-90:0:-0.9 and 0.6);
	\draw[thick] (2.6,0.1) arc (-90:0:-0.5 and 0.4);
	
	\draw[thick] (0.5,-0.8) arc (0:90:0.5 and 0.5);
	\draw[thick] (1.8,-0.8) arc (0:180:0.5 and 0.5);
	\draw[thick] (2.1,-0.8) arc (0:90:-0.5 and 0.5);
	\end{tikzpicture}\;\;,
	\qquad\qquad
 \begin{tikzpicture}
 	\draw[thick, dotted] (-0.05,0.5) arc (0:10:0 and -7.5);
 	\draw[thick, dotted] (-0.05,0.55) -- (2.65,0.55);
 	\draw[thick, dotted] (2.65,0.5) arc (0:10:0 and -7.5);
	\draw[thick, dotted] (-0.05,-0.85) -- (2.65,-0.85);
	\draw[thick] (0,0.05) arc (-90:0:0.5 and 0.45);
	\draw[thick] (0.8,0.5) arc (-180:0:0.5 and 0.45);
	\draw[thick] (2.6,0.05) arc (-90:0:-0.5 and 0.45);

	\draw[thick] (0.0,-0.15) arc (-180:0:1.3 and 0.0);

	\draw[thick] (0.5,-0.8) arc (0:90:0.5 and 0.45);
	\draw[thick] (1.8,-0.8) arc (0:180:0.5 and 0.45);
	\draw[thick] (2.1,-0.8) arc (0:90:-0.5 and 0.45);

	\end{tikzpicture}
\end{equation*}
\caption{Examples of affine diagrams for $N=4$, with the left and right sides of the framing rectangle identified. The first diagram represents the translation generator $u$ while the second diagram is for the generator $e_4\in\ATL{4}(m)$. The third and fourth ones are examples of $j=0$ diagrams.
}
\label{fig:aff-diag}
\end{figure}
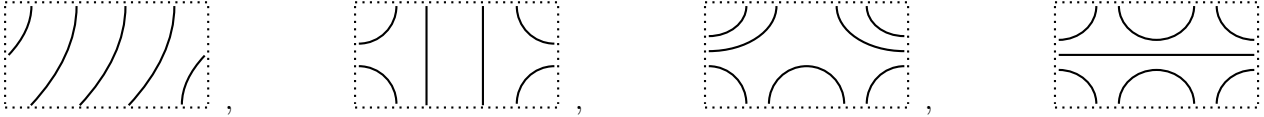

In terms of generators and relations,
the algebra $\ATL{N}(m)$ is generated by the $e_i$'s together with the identity, subject to the usual relations~\cite{MartinSaleur,Jones,Green}
\begin{eqnarray}
e_j^2&=&me_j,\nonumber\\
e_je_{j\pm 1}e_j&=&e_j,\nonumber\\
e_je_k&=&e_ke_j\qquad(j\neq k,~k\pm 1),\label{TL}
\end{eqnarray}
where $j=1,\ldots,N$ and the indices are interpreted modulo $N$,
and
 by generators $u$ and $u^{-1}$ of translations by one site to the right
and to the left, respectively.  The following additional defining
relations are then obeyed,
\begin{eqnarray}
ue_ju^{-1}&=&e_{j+1},\nonumber\\
u^2e_{N-1}&=&e_1\ldots e_{N-1},\nonumber
\end{eqnarray}
and $u^{\pm N}$ is a central element.
 The algebra generated by the $e_i$ and $u^{\pm1}$ together with these  relations is usually called the \textit{affine} Temperley--Lieb algebra $\ATL{N}(m)$.

We will consider translations by an even number of sites only,
{\it i.e.}, restrict to powers of $u^2$, and  replace a
non-contractible loop by a numerical factor $m$, as for the contractible loops. This
constraint (see 4.2.2 and 4.3.1 in~\cite{MartinSaleur}) together with
taking a quotient by the ideal generated by $u^N-1$ makes the algebra
finite dimensional.
We call the resulting object the  \textit{augmented Jones--Temperley--Lieb} algebra
$\JTL{2L}(m)$, where recall we have set   $N=2L$. This
algebra is slightly bigger than the one used
in~\cite{ReadSaleur07-1,GRS2}, called the \textit{Jones--Temperley--Lieb} algebra, which we denote by $\rJTL{2L}(m)$. The  difference is entirely in the sector/ideal
with zero through-lines. The algebra $\JTL{2L}(m)$ introduced here
contains in this ideal all affine diagrams
while the algebra $\rJTL{2L}(m)$ has only planar (or usual TL) diagrams in
this sector\footnote{The last algebra
is also known as oriented Jones annular subalgebra in the Brauer
algebra~\cite{Jones}.}.
In other words, in $\rJTL{2L}(m)$, one only keeps track of which points are connected to which in the diagrams, while in $\JTL{2L}(m)$, one also keeps information of how the connectivities wind around the annulus (the ambiguity does not arise when there are through-lines propagating). Formally, we have a covering homomorphism (surjection) of algebras
\begin{equation}\label{psi-hom}
\psi: \JTL{2L}(m) \longrightarrow \rJTL{2L}(m)
\end{equation}
which acts non-trivially only in the zero through-lines subalgebra and maps the diagrams as
\begin{equation}\label{psi-ex}
  \xymatrix@C=8pt@R=1pt@M=-5pt@W=-2pt{
  &&	\mbox{}\quad\xrightarrow{{\mbox{}\quad\psi\quad} }\quad &\\
  & {
 \begin{tikzpicture}
 	\draw[thick, dotted] (-0.05,0.5) arc (0:10:0 and -7.5);
 	\draw[thick, dotted] (-0.05,0.55) -- (2.65,0.55);
 	\draw[thick, dotted] (2.65,0.5) arc (0:10:0 and -7.5);
	\draw[thick, dotted] (-0.05,-0.85) -- (2.65,-0.85);
	\draw[thick] (0,0.1) arc (-90:0:0.5 and 0.4);
	\draw[thick] (0,-0.1) arc (-90:0:0.9 and 0.6);
	\draw[thick] (2.6,-0.1) arc (-90:0:-0.9 and 0.6);
	\draw[thick] (2.6,0.1) arc (-90:0:-0.5 and 0.4);
	\draw[thick] (0.5,-0.8) arc (0:90:0.5 and 0.5);
	\draw[thick] (1.8,-0.8) arc (0:180:0.5 and 0.5);
	\draw[thick] (2.1,-0.8) arc (0:90:-0.5 and 0.5);
	\end{tikzpicture}
\quad}&
	& {	\quad
  \begin{tikzpicture}
 	\draw[thick, dotted] (-0.05,0.5) arc (0:10:0 and -7.5);
 	\draw[thick, dotted] (-0.05,0.55) -- (2.65,0.55);
 	\draw[thick, dotted] (2.65,0.5) arc (0:10:0 and -7.5);
	\draw[thick, dotted] (-0.05,-0.85) -- (2.65,-0.85);
	\draw[thick] (0.5,0.5) arc (-180:0:0.8 and 0.56);
	\draw[thick] (0.8,0.5) arc (-180:0:0.5 and 0.4);
	\draw[thick] (2.1,-0.8) arc (0:180:0.8 and 0.56);
	\draw[thick] (1.8,-0.8) arc (0:180:0.5 and 0.4);
	\end{tikzpicture}}
  }
\end{equation}

The first
algebra  $\JTL{2L}(m)$ has  dimension (as follows from dimensions of generically irreducible modules described below)
\begin{equation}
\dim\bigl(\JTL{N}(m)\bigr) = \binom{2L}{L}^2 + \sum_{j=1}^L j\binom{2L}{L-j}^2,
\end{equation}
while the second algebra $\rJTL{N}(m)$
 has  dimension
\begin{equation}\label{dim-rJTL}
\dim\bigl(\rJTL{N}(m)\bigr) = \left(\binom{2L}{L}-\binom{2L}{L-1}\right)^2 + \sum_{j=1}^L j\binom{2L}{L-j}^2.
\end{equation}

We will only be concerned in this paper with the case $m=1$ for which
 the algebra $\JTL{2L}(m)$ is non semi-simple; in the following we usually  suppress all
 reference to $m$.

\subsection{Standard and co-standard modules}\label{sec:st-mod}
It is useful to  go back for a little while to the case of the full affine
 Temperley Lieb algebra $\ATL{N}(m)$. Set $m=\q+\q^{-1}$. For generic
 $\q\neq 1$, the irreducible representations we shall need are parametrized by two
 numbers. In terms of diagrams, the first is the number of
 through-lines $2j$, with $j=0,1,\ldots, L$. Using the natural action (by stacking affine diagrams) of the algebra  discussed in the previous subsection,
 we now decide that the result of  this action is zero whenever the affine diagrams  obtained have a number of
 through lines less than $2j$. Furthermore, for a given
 non-zero value of $j$, it is possible using the  action of the algebra,  to cyclically
 permute the free sites: this gives rise to the introduction of a
  pseudomomentum $K$. Whenever $2j$ through-lines wind counterclockwise around
 the annulus $l$ times, we can decide to unwind them at the price of a factor
 $e^{2ijlK}$; similarly, for clockwise winding, the phase is $e^{-i
 2jlK}$ \cite{MartinSaleur,MartinSaleur1}. This action gives rise to a generically
irreducible module, which we denote by
$\AStTL{j}{e^{2iK}}$.
 In the parametrization $(t,z)$ chosen
in~\cite{GL}, this corresponds to $t=2j$ and the twist parameter $z^2=e^{2iK}$.

 The dimensions of these modules $\AStTL{j}{e^{2iK}}$ over $\ATL{2L}(m)$  are then given by
 \begin{equation}\label{eq:dj}
 \hat{d}_{j}=
 \binom{2L}{L+j},\qquad j>0.
 \end{equation}
Note that the numbers do not depend on $K$ (but representations with
different $e^{iK}$ are not isomorphic).
These generically irreducible modules
$\AStTL{j}{e^{2iK}}$ are known also as
\textit{ standard (or cell)} $\ATL{N}(m)$-modules~\cite{GL}.

 Keeping $\q$ generic, degeneracies in the standard modules appear whenever
 \begin{eqnarray}\label{deg-st-mod}
 e^{2iK}&=&\q^{2j+2k},\qquad
k\hbox{ is a strictly positive integer.}
 \end{eqnarray}
 The representation $\AStTL{j}{\q^{2j+2k}}$ then becomes reducible, and contains a submodule isomorphic to
 $\AStTL{j+k}{\q^{2j}}$. The quotient is generically irreducible, with
 dimension $\hat{d}_j-\hat{d}_{j+k}$. The degeneracy
~\eqref{deg-st-mod} is well-known \cite{MartinSaleur1,GL}~\footnote{Note that the twist term in~\cite{PasquierSaleur}, which was denoted there
   $q^{2t}$, reads in these notations as $e^{2iK}$. It  corresponds to $z^2$ in
   the Graham--Lehrer work~\cite{GL}, and to the parameter $x$ in the work of  Martin--Saleur \cite{MartinSaleur1}. The case where $k=1$ is special, and related with braid translation of the blob algebra theory.
We note that in the $\JTL{N}$ case, $2j$ lines going around the cylinder pick up a phase $e^{i 2jK}=1$. In \cite{MartinSaleur1}, this corresponds to $\alpha_h=x^h=1$.}.  When $\q$ is a root of unity, there are infinitely many solutions to the equation \eqref{deg-st-mod}, leading to a complex pattern of degeneracies to which we turn below.


The case $j=0$ is a bit special. There is no pseudomomentum, but representations are still characterized by another parameter, related with  the weight given to  non contractible loops. Parametrizing this weight as $z+z^{-1}$, the corresponding standard module of  $\ATL{2L}(m)$ is denoted  $\AStTL{0}{z^2}$ and has  dimension $\binom{2L}{L}$.

\bigskip

We now specialize to the Jones--Temperley--Lieb algebra $\rJTL{N}(m)$.
 In this case, the rule that winding through-lines can simply be unwound means that the pseudomomentum must satisfy
$jK\equiv 0~\hbox{mod}~\pi$ \cite{Jones}.
 All possible values of the parameter $z^2=e^{2iK}$ are thus $j$-th
 roots of unity ($z^{2j}=1$,~\cite{Green}). The kernel of the
 homomorphism $\psi$ described by~\eqref{psi-hom} and~\eqref{psi-ex} (and the ideal  in $\ATL{N}(m)$ generated by $u^N-1$, in
 particular) acts trivially on these modules if $j>0$. In what
 follows, we will thus  use
 the same notation $\AStTL{j}{z^2}$, with $j>0$, for the standard $\rJTL{N}(m)$-modules.
 We note that two standard $\rJTL{N}$-modules having only different signs in the $z$ parameter are isomorphic.

\medskip

If $j=0$, requiring the weight of the non contractible loops to be $m$ as well
leads to the standard $\JTL{N}(m)$-module $\AStTL{0}{\q^2}$ which is reducible
even for generic $\q$ -- it contains a submodule isomorphic to
$\AStTL{1}{1}$. Meanwhile, on the standard module $\AStTL{0}{\q^2}$ the kernel of the homomorphism
$\psi$ is non-trivial: the standard module over $\rJTL{N}(m)$ for $j=0$ is obtained
precisely by taking the quotient $\AStTL{0}{\q^2}/\AStTL{1}{1}$ as
in~\cite{GL}. This module is now simple for generic $\q$, has the
dimension
$$\hat{d}_0=\binom{2L}{L}-\binom{2L}{L-1}$$
 and is denoted by~$\bAStTL{0}{\q^2}$. The   standard $\JTL{N}$-module
$\AStTL{0}{\q^2}$ is of  dimension $\binom{2L}{L}$.

 In the full  construction of direct summands of our spin-chains -- the tilting modules -- we shall also need a concept of so-called \textit{co}-standard modules. They are defined as the duals $\bigl(\StJTL{j}{P}\bigr)^*$ or vector spaces of linear functionals on $\StJTL{j}{P}$ endowed with the $\JTL{N}$ action by
\begin{equation}
\bigl(\StJTL{j}{P}\bigr)^*:\qquad\qquad  A v^*(\cdot) = v^*(A^* \cdot),\qquad \text{with}\quad A\in\JTL{N}(m), \quad v^*\in\bigl(\StJTL{j}{P}\bigr)^*,
\end{equation}
where $\cdot^*$ is an anti-automorphism on $\JTL{N}(m)$ defined by interchanging the inner and outer boundaries of the affine diagram.
Equivalently, we can say that the basis in co-standards $\bigl(\StJTL{j}{P}\bigr)^*$ is defined by reflecting the framing in a horizontal line. Then, the value of $v^*$ on $\StJTL{j}{P}$ is given by the bilinear form defined in~\cite[Sec 2.6]{GL}. It was shown in~\cite{GL} that $\bigl(\StJTL{j}{z}\bigr)^*$ is generically isomorphic to $\StJTL{j}{z^{-1}}$. At critical values of the parameters, special pseudomomenta $K$ and roots of unity values of $\q$, there is a non-trivial homomorphism from $\bigl(\StJTL{j}{z}\bigr)^*$ to $\StJTL{j}{z^{-1}}$ controlled by the bilinear form.

\medskip

\subsection{The structure of the standard $\JTL{N}(m)$ modules}\label{sec:st-mod-str}
A central point in  representation theory of the algebras $\JTL{N}(m)$ and $\rJTL{N}(m)$ is
 their cellular-algebra structure and the full power of cellular algebras technique~\cite{GL0,GL}
can be used for studying many important indecomposable modules -- including the so-called
projective modules.
We now review shortly these important concepts.

\newcommand{\setW}{\mathcal{S}}
\newcommand{\KK}{P}

\subsubsection{Cellular algebras and (co)cell modules}
We first give a brief review of cellular algebras and general aspects
of their representation theory. These algebras were initially introduced in~\cite{GL0}.
Roughly, cellular algebras are defined as those with a special basis having particularly nice properties under multiplication.
To be more precise, a \textit{cellular} algebra $A$ over $\oC$ is an associative algebra equipped with a finite partially ordered set $\setW$ (the set of weights with an order on it)
and a finite set $W(\lambda)$, for any $\lambda\in\setW$, such that the algebra has a basis $C^{\lambda}_{i,j}$, where $\lambda\in\setW$ and $i$ and $j$ run through all elements in $W(\lambda)$, with the following two properties: (i) for
each $a \in A$ the product $a C^{\lambda}_{i,j}$
can be written as $\bigl(\sum_{k\in W(\lambda)}c_a(i,k) C^{\lambda}_{k,j}\bigr) + b$, where $b$
is a linear combination of basis elements with upper index
$\mu$ strictly smaller than $\lambda$
 and where the coefficients $c_a(i,k)$
 do not depend on $j$; (ii) there exists an anti-automorphism $\cdot^*$
 on $A$ such that its square is the identity, and  that  it sends $C^{\lambda}_{i,j}$ to $C^{\lambda}_{j,i}$.
The diagram algebras such that the ordinary Temperley-Lieb algebras $\TL{N}(m)$ are examples of cellular algebras. The special basis $C^{\lambda}_{i,j}$  in the latter case is given by   Temperley-Lieb diagrams  of arcs and through lines without crossings. Elements from the set $\setW$ of weights are just numbers of through lines in the diagrams and the partial order is just the ordinary order on natural numbers. The algebras $\JTL{N}(m)$ are just another example of a class of cellular algebras. Now the weights $\lambda$ are the pairs of numbers $(j,P)$, with  $2j$ is the number of through lines and $P$ is the exponent of the pseudomomentum. The anti-automorphims $\cdot^*$ from the definition is just the reflection of a diagram through the horizontal axis. The proof of cellularity is rather straightforward following the definition of the multiplication and the basis in $\JTL{N}(m)$ given in the previous section.

A cornerstone of cellular algebras theory is the notion of
{\textit cell modules}. These modules are parametrized by elements $\lambda$ from the set $\setW$ of weights. A cell (or standard) module $\StTL{\lambda}$ over a cellular algebra $A$
is a vector space with a basis $\left\{C_j \, |\, j\in W(\lambda)\right\}$, with  the action of any $a\in A$ given by
$aC_{j} = \sum_{k\in M(\lambda)}c_a(j,k) C_{k}$,  where  the numbers $c_a(j,k)$ are those from the definition of the cellular algebras. Similarly, we introduce \textit{co-cell} modules as duals to the cell modules, and defined with the use of the algebra anti-automorphism $\cdot^*$. The modules $\StJTL{j}{P}$, with $P$ being all $j$-th
roots of unity, introduced in the previous section give all the cell modules for
$\JTL{N}(m)$. The co-cell modules are given by the corresponding costandards described in the end of Sec.~\ref{sec:st-mod}.

A cellular algebra itself and its projective modules   are
filtered by the cell modules, or in other words, can be constructed as
appropriate glueings or extensions between cell modules.
Precise statements about projective modules will be given below. Now, we go into details
for the class of the diagram algebras $\JTL{N}(m)$.

\subsubsection{The partial order on cell $\JTL{N}$-modules weights}
We now turn to the details of the representation theory of $\rJTL{N}(1)$, which
we will then apply to the study of our super-spin chain.
We first analyze  the representation theory for the augmented algebra $\JTL{N}(1)$ - which is  technically easier -  and then  restrict ourselves to the $\rJTL{N}(1)$ algebra.

The subquotient structure of the cell
$\JTL{N}(1)$ modules can be easily obtained using results of~\cite{GL}.
The set of weights for the cell $\JTL{N}(m)$-modules $\StJTL{j}{P}$ is
\begin{equation}\label{def-set-weight-st}
\setW = \bigl\{(0,\q^2),(j,\KK)|\;1\leq j\leq N,\, \KK=e^{\frac{2i\pi}{j}l}, 1\leq
l\leq j\bigr\}.
\end{equation}
We introduce a (weak) partial order $\poless$ (due to~\cite{GL}) on this set of weights such that $(j_1,P_1)\poless(j_2,P_2)$ if $j_2-j_1=k$ for a non-negative integer $k$ and the pairs $(j_1,P_1)$ and $(j_2,P_2)$ satisfy
\begin{equation}
P_1 = \q^{2\epsilon j_2}\qquad \text{and}\qquad P_2 = \q^{2\epsilon j_1},\qquad \epsilon=\pm1.
\end{equation}

Note that the partial order $\poless$ generates equivalence classes on the set of the weights $(j,P)$ -- two weights are in the same equivalence class if and only if they are in the relation $\poless$. The result of~\cite{GL} is that there  exist non-trivial homomorphisms only between  cell (or standard) modules having weights from the same equivalence class.
The idea thus is that simple JTL modules can be glued with each other (or extended  by each other) non-trivially only if their weights are from the same equivalence class.  There are non-trivial classes (containing two or more weights) only when $\q$ is a root of unity. In this case, many cell modules are reducible but indecomposable.

We give now several examples of the partial order on the set of weights.
For the case $\q=i$, the equivalence classes (generated by the partial order $\preceq$) are given by the disjoint
oriented graphs on the diagram on Fig.~\ref{fig:part-ord-secondroot}, where it is apparent  that there are two non-trivial classes, denoted by arrows of different types.
 \begin{figure}\centering
\begin{equation*}
\xymatrix@C=15pt@R=15pt@W=2pt@M=2pt{%
    {\IrJTL{0}{-1}}   &\IrJTL{1}{1}\ar[l]
    &\IrJTL{2}{1}
    &\IrJTL{3}{1} \ar[dl]
    &\IrJTL{4}{1} \ar@{-->}@/_1pc/[ll]
    &\IrJTL{5}{1} \ar@[][ddl]& \IrJTL{6}{1} \ar@{-->}@/_1pc/[ll]&\quad\dots\ar@[][dddl]&\ar@{-->}@/_1pc/[ll]\\
    &&\IrJTL{2}{-1} \ar@[][ul] &\IrJTL{3}{e^{\frac{2i\pi}{3}}}
    &\IrJTL{4}{i} &\IrJTL{5}{e^{\frac{2i\pi}{5}}}&\IrJTL{6}{e^{\frac{i\pi}{3}}}&\quad\dots&\\
    &&&\IrJTL{3}{e^{\frac{4i\pi}{3}}}
    &\IrJTL{4}{-1} \ar@[][uul] &\IrJTL{5}{e^{\frac{4i\pi}{5}}}&\IrJTL{6}{e^{\frac{2i\pi}{3}}}&\quad\dots&\\
    &&&
    &\IrJTL{4}{-i} &\IrJTL{5}{e^{\frac{6i\pi}{5}}}&\IrJTL{6}{-1}\ar@[][uuul]&\quad\dots&\\
    &&&
    &&\IrJTL{5}{e^{\frac{8i\pi}{5}}}&\IrJTL{6}{e^{\frac{4i\pi}{3}}}&\quad\dots&\\
    &&&
    &&&\IrJTL{6}{e^{\frac{5i\pi}{3}}}&\quad\dots&
    }
\end{equation*}
     \caption{The partial order
$\preceq$ on the set of weights at $\q=i$. Two
nodes $a$ and $b$ are connected by an arrow $a\to b$ if and only if $a\succeq
b$.}
    \label{fig:part-ord-secondroot}
    \end{figure}
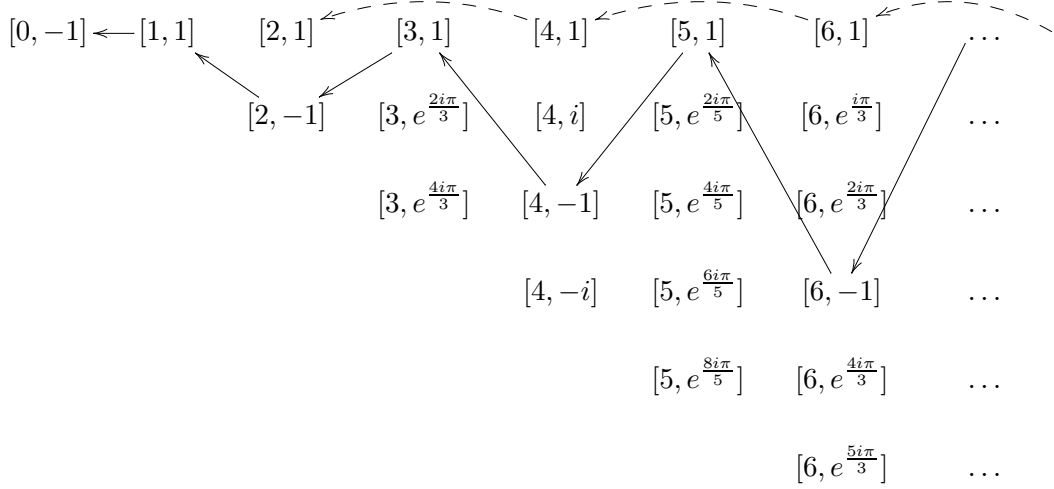
For the case $\q=e^{\frac{i\pi}{3}}$, the equivalence classes  are given in the diagram on Fig.~\ref{fig:part-ord-thirdroot}.
 \begin{figure}\centering
\begin{equation*}
  \xymatrix@C=12pt@R=15pt@W=2pt@M=2pt{%
    {\IrJTL{0}{e^{\frac{2i\pi}{3}}}} &\IrJTL{1}{1}\ar@[]@/^1pc/[l]
    &\IrJTL{2}{1} \ar@[]@/_1pc/[ll]
    &\IrJTL{3}{1}
    &\IrJTL{4}{1} \ar@[][dl]\ar@[][ddl]
    &\IrJTL{5}{1} \ar@[][dll]\ar@[][ddll]
    & \IrJTL{6}{1}  \ar@{-->}@/_1pc/[lll] &\quad\dots\ar@[][ddl]\ar@[][ddddl]&\quad& \ar@{-->}@/_1pc/[lll] \\
    &&\IrJTL{2}{-1} &\IrJTL{3}{e^{\frac{2i\pi}{3}}} \ar@[][ull] \ar@[][ul]
    &\IrJTL{4}{i} &\IrJTL{5}{e^{\frac{2i\pi}{5}}}&\IrJTL{6}{e^{\frac{i\pi}{3}}}&\quad\dots&\quad&\\
    &&&\IrJTL{3}{e^{\frac{4i\pi}{3}}}\ar@[][uul] \ar@[][uull]
    &\IrJTL{4}{-1}
    &\IrJTL{5}{e^{\frac{4i\pi}{5}}}&\IrJTL{6}{e^{\frac{2i\pi}{3}}}
    \ar@[][uull] \ar@[][uul]  &\quad\dots&\quad&\\
    &&&
    &\IrJTL{4}{-i} &\IrJTL{5}{e^{\frac{6i\pi}{5}}}&\IrJTL{6}{-1}&\quad\dots&\quad&\\
    &&&
    &&\IrJTL{5}{e^{\frac{8i\pi}{5}}}&\IrJTL{6}{e^{\frac{4i\pi}{3}}}\ar@[][uuuull] \ar@[][uuuul]
    &\quad\dots&\quad&\\
    &&&
    &&&\IrJTL{6}{e^{\frac{5i\pi}{3}}}&\quad\dots&\quad&
    }
\end{equation*}
      \caption{The partial order
$\preceq$ on the set of weights at $\q=e^{\frac{i\pi}{3}}$. Two
nodes $a$ and $b$ are connected by an arrow $a\to b$ if and only if $a\succeq
b$.}
    \label{fig:part-ord-thirdroot}
    \end{figure}
In this case, we also have only two non-trivial equivalence classes
containing two and more
weights, while  all the other classes (that is,  all the nodes without in- and
out-going arrows) are trivial, and  contain only one weight. For a trivial class, the corresponding cell
module is simple. The non-trivial class with  dashed arrows on the figure contains all weights $(j,1)$ with $j\;\textrm{mod}\;3=0$ and each node has a single in-going or out-going arrow. We call this class \textit{singly critical}. The second class (with  solid arrows)  involves the weights $(0,\q^{2})$ and  $(j,1)$ with $j\;\textrm{mod}\;3\in\{1,2\}$
and $(j,\q^{\pm2})$ whenever $j\;\textrm{mod}\;3=0$ and $j>0$. Nodes in this class have double arrows and we call it \textit{doubly critical}. As we shall see below, the cell modules from the singly critical class are of chain type while the cell modules having weights from the doubly critical class have a ``braid-type'' subquotient structure.

Each cell module $\StJTL{j}{P}$ is known to be indecomposable: we  denote its top simple subquotient -- the  quotient by its maximal submodule -- by $\IrrJTL{j}{P}$; we will also use the short hand notation $\AIrrTL{j}{P}$.  Using the partial order $\poless$ on $\setW$  described earlier, the simple-module content of the cell $\JTL{N}(1)$-modules
can be deduced~\cite{GL}
\begin{equation}
\left[\StJTL{j}{P}\right] = \bigoplus_{(j',P')\pomore(j,P)} \IrJTL{j'}{P'},
\end{equation}
where the notation $\left[\StJTL{j}{P}\right]$ means that the corresponding module, as a vector space, is given by the right hand side, and that it has the simple subquotients denoted by $\AIrrTL{j}{P}$ that appear in the sum. We have thus described the structure of $\StJTL{j}{P}$ up to glueings or arrows between simple subquotients. The information about arrows can be also deduced from results of~\cite{GL}:
the cell modules over $\JTL{N}(1)$ have  a sequence of {\sl embeddings} (see the description of injective
 homomorphisms from Thm.~3.4 and also the proof of Thm.~5.1 in~\cite{GL})
\begin{equation}\label{cell-emb}
\StJTL{j_1}{P_1}\hookrightarrow\StJTL{j_2}{P_2}\hookrightarrow\dots\hookrightarrow\StJTL{j_n}{P_n}
\end{equation}
for any sequence of weights
$(j_1,P_1)\pomore(j_2,P_2)\pomore\dots\pomore(j_n,P_n)$ in Fig.~\ref{fig:part-ord-thirdroot}.
This sequence of embeddings tells us that modules corresponding to the doubly-critical  class have a braid-type subquotient structure, that can be extracted from
Fig.~\ref{fig:cell-thirdroot}, showing the subquotient structure for
the $\JTL{N}(1)$-module $\StJTL{0}{\q^2}$. The other $\StJTL{j}{P}$,
with $P=1$ for $j\;\textrm{mod}\;3\in\{1,2\}$ and $P=\q^{\pm2}$ when
$j\;\textrm{mod}\;3=0$, are obtained as the submodules `emanating' from
the corresponding simple subquotient. The cell modules $\StJTL{j}{1}$,
with $j\;\textrm{mod}\;3=0$, from the singly critical class are also presented
in Fig.~\ref{fig:cell-thirdroot} and are of chain type.
 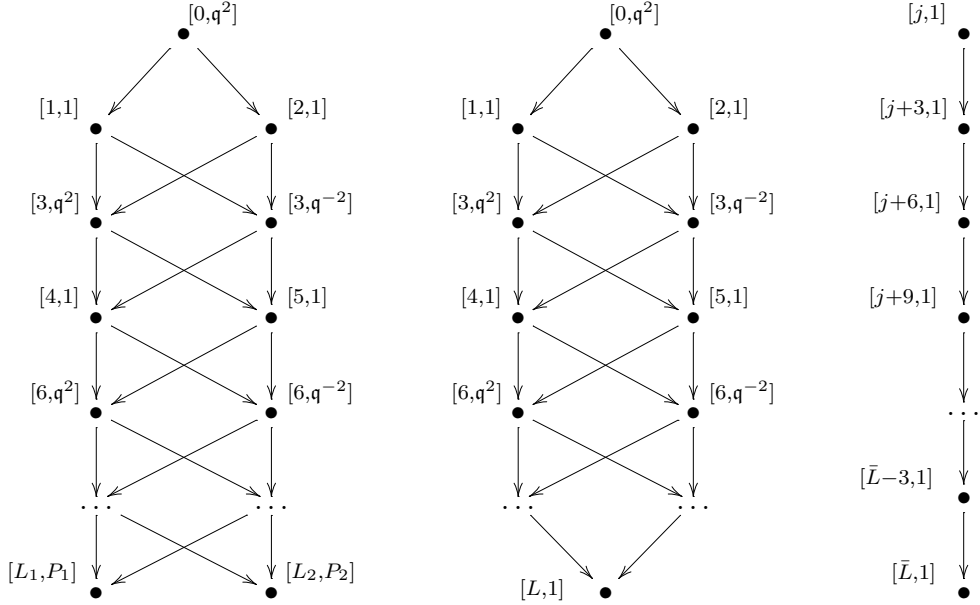
\begin{figure}\centering
 \begin{equation*}
   \xymatrix@R=25pt@C=18pt
   {{}&\bullet\ar@{}|{\substack{\AIrrTL{0}{\q^2}}\kern-7pt}[]+<15pt,15pt>\ar[dr]\ar[dl]&&\\
     {\bullet}\ar@{}|{\substack{\AIrrTL{1}{1}}\kern-7pt}[]+<-35pt,15pt>\ar[d]\ar[drr]
     &&{\bullet}\ar@{}|{\substack{\AIrrTL{2}{1}}\kern-7pt}[]+<20pt,15pt>\ar[d]\ar[dll]\\
     {\bullet}\ar@{}|{\substack{\AIrrTL{3}{\q^2}}\kern-7pt}[]+<-38pt,15pt>\ar[d]\ar[drr]
     &&{\bullet}\ar@{}|{\substack{\AIrrTL{3}{\q^{-2}}}\kern-7pt}[]+<30pt,15pt>\ar[d]\ar[dll]\\
     {\bullet}\ar@{}|{\substack{\AIrrTL{4}{1}}\kern-7pt}[]+<-35pt,15pt>\ar[d]\ar[drr]
     &&{\bullet}\ar@{}|{\substack{\AIrrTL{5}{1}}\kern-7pt}[]+<20pt,15pt>\ar[d]\ar[dll]\\
     {\bullet}\ar@{}|{\substack{\AIrrTL{6}{\q^2}}\kern-7pt}[]+<-38pt,15pt>\ar[d]\ar[drr]
     &&{\bullet}\ar@{}|{\substack{\AIrrTL{6}{\q^{-2}}}\kern-7pt}[]+<30pt,15pt>\ar[d]\ar[dll]\\
     {\dots}\ar[d]\ar[drr]
     &&{\dots}\ar[d]\ar[dll]\\
     {\bullet}\ar@{}|{\substack{\AIrrTL{L_1}{P_1}}\kern-7pt}[]+<-48pt,15pt>
     &&{\bullet}\ar@{}|{\substack{\AIrrTL{L_2}{P_2}}\kern-7pt}[]+<30pt,15pt>
     } \qquad
   \xymatrix@R=25pt@C=18pt
   {{}&\bullet\ar@{}|{\substack{\AIrrTL{0}{\q^2}}\kern-7pt}[]+<15pt,15pt>\ar[dr]\ar[dl]&&\\
     {\bullet}\ar@{}|{\substack{\AIrrTL{1}{1}}\kern-7pt}[]+<-35pt,15pt>\ar[d]\ar[drr]
     &&{\bullet}\ar@{}|{\substack{\AIrrTL{2}{1}}\kern-7pt}[]+<20pt,15pt>\ar[d]\ar[dll]\\
     {\bullet}\ar@{}|{\substack{\AIrrTL{3}{\q^2}}\kern-7pt}[]+<-38pt,15pt>\ar[d]\ar[drr]
     &&{\bullet}\ar@{}|{\substack{\AIrrTL{3}{\q^{-2}}}\kern-7pt}[]+<30pt,15pt>\ar[d]\ar[dll]\\
     {\bullet}\ar@{}|{\substack{\AIrrTL{4}{1}}\kern-7pt}[]+<-35pt,15pt>\ar[d]\ar[drr]
     &&{\bullet}\ar@{}|{\substack{\AIrrTL{5}{1}}\kern-7pt}[]+<20pt,15pt>\ar[d]\ar[dll]\\
     {\bullet}\ar@{}|{\substack{\AIrrTL{6}{\q^2}}\kern-7pt}[]+<-38pt,15pt>\ar[d]\ar[drr]
     &&{\bullet}\ar@{}|{\substack{\AIrrTL{6}{\q^{-2}}}\kern-7pt}[]+<30pt,15pt>\ar[d]\ar[dll]\\
     {\dots}\ar[dr]
     &&{\dots}\ar[dl]\\
     &{\bullet}\ar@{}|{\substack{\AIrrTL{L}{1}}\kern-7pt}[]+<-55pt,2pt> &
     } \quad
   \xymatrix@R=25pt@C=18pt
   {{}&\bullet\ar@{}|{\substack{\AIrrTL{j}{1}}\kern-7pt}[]+<-35pt,15pt>\ar[d]&&\\
     &{\bullet}\ar@{}|{\substack{\AIrrTL{j+3}{1}}\kern-7pt}[]+<-45pt,15pt>\ar[d]&&\\
     &{\bullet}\ar@{}|{\substack{\AIrrTL{j+6}{1}}\kern-7pt}[]+<-50pt,15pt>\ar[d]&&\\
     &{\bullet}\ar@{}|{\substack{\AIrrTL{j+9}{1}}\kern-7pt}[]+<-53pt,15pt>\ar[d]&&\\
     &{\dots}\ar[d]&&\\
     &{\bullet}\ar@{}|{\substack{\AIrrTL{\bar{L}-3}{1}}\kern-7pt}[]+<-58pt,15pt> \ar[d]&&\\
     &{\bullet}\ar@{}|{\substack{\AIrrTL{\bar{L}}{1}}\kern-7pt}[]+<-45pt,15pt> &&
     }
 \end{equation*}
      \caption{The structure of the cell $\JTL{N}$-modules
      $\StJTL{j}{P}$ with $2j$ through lines at
      $\q=e^{\frac{i\pi}{3}}$. The two modules on the left are cell
      modules from the doubly critical class: on the first one $L_i=L-2 + i$ and
      $P_i=1$ when $L\;\textrm{mod}\;3=2$, and $L_i=L$, $P_1=\q^2$,
      and $P_2=\q^{-2}$ when $L\;\textrm{mod}\;3=0$; the second
      module is for the case $L\;\textrm{mod}\;3=1$; the rightmost
      module corresponds to the doubly critical class (with
      $j\;\textrm{mod}\;3=0$ and $\bar{L}=L-L\;\textrm{mod}\;3$). We denote simple subquotients by their weights $[j,P]$.
}
    \label{fig:cell-thirdroot}
    \end{figure}

We note  now that the cell modules  over the quotient $\rJTL{N}(1)$ (and not $\JTL{N}(1)$)  for $j>0$ are
given by the same diagrams in Fig.~\ref{fig:cell-thirdroot}  and we
use the same notation $\StJTL{j}{P}$ for them. The only difference  is
for $j=0$. The cell $\rJTL{N}(1)$-module with the weight $(0,\q^2)$
is given by the quotient $\StJTL{0}{\q^2}/\StJTL{1}{1}$, as we discussed previously, and has only two irreducible subquotients: $\bAStTL{0}{\q^2} = \AIrrTL{0}{\q^2}\to\AIrrTL{2}{1}$.

\medskip

 We conclude  this subsection with a comment on the structure of the costandard  modules $\bigl(\StJTL{j}{P}\bigr)^*$ introduced at the end of Sec.~\ref{sec:st-mod}. Generically, these modules are isomorphic to $\StJTL{j}{P^{-1}}$ while for our choice of $\q$ and parameters $P$ we have only homomorphisms $\psi_{j,P}:\StJTL{j}{P^{-1}}\to\bigl(\StJTL{j}{P}\bigr)^*$ that map the irreducible head $\IrrJTL{j}{P^{-1}}$ to the (unique) irreducible submodule in  $\bigl(\StJTL{j}{P}\bigr)^*$ and the kernel of $\psi_{j,P}$ can be studied using the bilinear form in~\cite[Sec 2.6]{GL}. We thus have that the diagrams for $\StJTL{j}{1}^*$ (with $j\;\textrm{mod}\;3$ equal $1$ or $2$) have in the bottom the irreducible submodule $\IrrJTL{j}{1}$ and all the arrows reversed while modules $\StJTL{j}{\q^{\pm2}}^*$ (with $j\;\textrm{mod}\;3=0$) have in the bottom the irreducible submodule $\IrrJTL{j}{\q^{\mp2}}$.

\subsection{Dimensions of simple JTL modules}

The subquotient structure of the standard modules  allows us to compute  the dimensions of all simple JTL modules.
It will be convenient in what follows to introduce the notation
\begin{equation}
\hat{d}_0'= \binom{2L}{L}
 \end{equation}
 and to recall the notation $\hat{d_j}$ from \eqref{eq:dj} for dimensions of the cell modules. Dimensions of the corresponding simples will be denoted as
 \begin{equation*}
\hat{d}_{j,P}^{(0)}\equiv \dim \IrrJTL{j}{P}.
 \end{equation*}
Using the subquotient structure from Fig. \ref{fig:cell-thirdroot}, we can write immediately the dimension of the simple associated with $j=0$ and $e^{2iK}=e^{2i\pi/3}$:
\begin{eqnarray}
\dim \IrrJTL{0}{\q^2} \equiv \hat{d}_{0,e^{2i\pi/3}}^{(0)}=\hat{d}_0'-\hat{d}_1-\hat{d}_2+2\hat{d}_3-\hat{d}_4-\hat{d}_5\ldots\nonumber\\
=\sum_{n=0}^\infty (d_n-d_{n+2}),
\end{eqnarray}
where $d_j=\hat{d}_j-\hat{d}_{j+1}$, with $j>0$, and $d_0=\hat{d}_0'-\hat{d}_1=\hat{d}_0$ are the dimensions of the cell modules for the ordinary Temperley--Lieb algebra.  We thus see that $\hat{d}_0^{(0)}$ coincides with the dimension of the simple for $j=0$ in the open or TL case:
\begin{equation}
\hat{d}_{0,e^{2i\pi/3}}^{(0)}=d_0^0=1.
\end{equation}
From the diagrams in Fig. \ref{fig:cell-thirdroot} one can obviously derive the following more general results:
\begin{eqnarray}
 \hat{d}_{2+3n,1}^{(0)}=\hat{d}_{2+3n}-\hat{d}_{3+3n}-\hat{d}_{3+3n,e^{2i\pi/3}}^{(0)},\nonumber\\
\hat{d}_{3+3n,e^{\pm 2i\pi/3}}^{(0)}=\hat{d}_{3+3n}-\hat{d}_{4+3n}-\hat{d}_{5+3n,1}^{(0)}.\nonumber
\end{eqnarray}
The final result can be thus obtained recursively.
Comparing these relations with those in the open TL case \cite{ReadSaleur07-2} we see that
\begin{eqnarray}
\hat{d}_{2+3n,1}^{(0)}&=&d_{2+3n}^0,\nonumber\\
\hat{d}_{3+3n,e^{\pm 2i\pi/3}}^{(0)}&=&d_{3+3n}^0,\nonumber
\end{eqnarray}
where $d_j^0$ are dimensions of simple TL modules corresponding to tops of the cell modules with $2j$
through lines.

\section{The periodic $s\ell(2|1)$ spin chain and its decomposition}\label{sec:4}
 Recall that in Sec. \ref{sec:susy} we introduced a family of periodic supersymmetric spin chains with nearest neighbour interaction $e_j$ given by \eqref{ej-1} and \eqref{ej-2} and the Hamiltonian is defined in \eqref{hamil}.
 In the rest of the paper, we will consider only the first member of this family -- the $s\ell(2|1)$ (or $g\ell(2|1)$) spin-chain which is an alternating tensor product of the fundamental representation $\oC^{2|1}$ on even sites and its dual on odd sites. We refer to App. \ref{appSl21} for the definition of $s\ell(2|1)$ and basics of its representation theory. The interaction operators
$e_j$ together with a translation operator give a representation of the JTL algebra $\rJTL{N}(1)$, see an explicit matrix representation in more general context of periodic
$s\ell(n+m|n)$ spin chains in App. \ref{sec:TL-faith}.

In this section, we describe a decomposition of the periodic $s\ell(2|1)$ spin chain onto indecomposable $\rJTL{N}(1)$ modules. The important difference from the $\gl(1|1)$ spin chain, as a representation of the JTL algebra $\rJTL{N}(0)$, studied in~\cite{GRS1,GRS2} is that the $s\ell(2|1)$ one turns out to be a \textsl{faithful} representation of $\rJTL{N}(1)$, i.e. the kernel of the representation is trivial. We give a proof of this fact including higher-rank cases, {\it i.e.}, all $s\ell(n+m|n)$ spin-chains with $n,m\geq1$, in App.~\ref{sec:TL-faith}. The faithfulness of our spin chains motivates the study of projective modules for $\rJTL{N}(1)$: because of the triviality of the kernel, all projective covers have to appear as submodules in the periodic $s\ell(2|1)$ spin chain.
 We give a brief review of this important concept (projectivity and projective covers). We then describe the structure of projective covers using general results in the theory of cellular algebras.

\subsection{Indecomposable modules: general definitions}\label{sec:indecgen}

We collect in this subsection the definitions of several important mathematical concepts that are needed in order to fully appreciate the rest of this paper, such as projectiveness, self-duality, and tilting modules.

We begin with a definition of what is called an injective hull of a simple module. It is the maximal indecomposable that can contain this simple module as a submodule. Then, any injective module is by definition a direct sum of injective hulls and if a submodule in a bigger module is injective then it is a direct summand. In
 contrast, projective modules are defined as direct sums of projective covers, where the projective cover of a simple module is the unique (for finite-dimensional algebras) indecomposable module of maximal dimension that can cover the simple module, {\it i.e.}, the projective cover contains it as a top subquotient. Then, if a subquotient of a bigger module is projective then it is a direct summand. Note that the distinction between the notions \textit{subquotient} and \textit{submodule} is crucial here. Therefore, the projectiveness property does not necessarily imply injectiveness and {\it vice-versa}.

Projective modules appear as direct summands of free modules like the regular representation of an algebra, but for spin chains -- and thus presumably LCFTs -- the direct summands are, more generally, \textit{tilting} modules which are defined and discussed in details below. In some cases, the tiltings are also indecomposable projective, but this is not necessarily the case. In general, there are tilting modules which are not projective, and projective modules which are not tilting.

In the theory of cellular algebras~\cite{GL0,GL1}, there is a general theorem that allows one  to obtain the subquotient structure of  projective covers knowing the subquotient structure of the standard (cell) modules. The essential part of this theorem can be expressed as a reciprocity property of projectives.  Let $[\StTL{}:\IrrTL{}]$ and  $[\PrTL{}:\StTL{}]$ denote the number of appearance of $\IrrTL{}$ in a diagram for a standard module $\StTL{}$ and the number of appearance of $\StTL{}$ in a diagram for the projective cover $\PrTL{}$, respectively. Then, the reciprocity property reads
\begin{equation}\label{recip}
[\PrTL{}:\StTL{}] = [\StTL{}:\IrrTL{}],
\end{equation}
{\it i.e.}, the projective cover
 $\PrTL{}$ that covers $\IrrTL{}$ is composed of those standard modules $\StTL{}$ that have the simple module
$\IrrTL{}$ as a subquotient. This statement can be
deduced from the proof of Thm. 3.7 in~\cite{GL0}.

\medskip
Having an indecomposable (and reducible) module $M$, we call {\it socle} its maximum semisimple submodule -- in terms of nodes and arrows in the subquotient diagram for $M$, the socle is the direct sum of all nodes having only ingoing arrows. Similarly, the \textit{top} of a module $M$ is the maximal subquotient with respect to the property that a quotient of $M$ is a semisimple module, {\it i.e.}, it is the subquotient of $M$ having only outgoing arrows.

\medskip

The spin-chains we consider have a non-degenerate bilinear form given explicitly, for example, in terms of spins. These spin-chains provide a special class of representations with two essential properties: (i) they are filtered by standard modules and (ii) they are self-dual, {\it i.e.}, $\Hilb_N^*\cong\Hilb_N$.
Direct summands in such representations are called tilting modules.

We define a \textit{tilting} module over a cellular algebra
as a module  that has a filtration by standard modules -- these are $\StJTL{j}{P}$ in our case -- and an inverse
filtration by the corresponding duals  -- the costandard  modules $\bigl(\StJTL{j}{P}\bigr)^*$, which have reversed arrows in their subquotient diagram. Note that both standard and costandard modules are introduced in the context of the JTL algebras in Sec.~\ref{sec:st-mod} and their structure is described in Sec.~\ref{sec:st-mod-str}. We recall that a
  filtration of an $A$-module $M$ by $A$-modules $W_i$, with $0\leq
  i\leq n-1$, is a sequence
  of embeddings
  \begin{equation}\label{eq:filtr}
  0=M_0\subset M_{1}\subset\dots \subset M_i
  \subset\dots \subset M_{n-1}\subset M_n=M
   \end{equation}
   such that the quotient $M_{i+1}/M_i$ of ``neighbor'' submodules (called the $i$th section) is isomorphic to $W_i$, or in simple words we can say that $M$ is a glueing of $W_i$'s.
The tilting modules are thus self-dual by our definition.
Several explicit examples will be given below.

We will also show that  tilting JTL modules satisfy a uniqueness property:
one can introduce the tilting module $\TilJTL{j}{P}$ \textit{generated} from a standard module $\StJTL{j}{P}$
as the indecomposable tilting module containing this standard module as a submodule. This property of having a standard submodule
uniquely defines the tilting module, up to an isomorphism. One should replace each simple subquotient of this standard module by a costandard module
having this simple module in its socle --- the unique simple subquotient that has only incoming arrows. The result is then
automatically a tilting module, by construction.

Before discussing  spin-chains and the reasons why tilting modules are more important objects for applications, we give some results about projective modules for the JTL algebra.

\subsection{The projective modules over $\rJTL{N}(1)$}\label{sec:proj}

We  first describe the cell content of projective covers
for all simple modules over $\rJTL{N}(1)$.
Let $\poless$ be the partial order on the set $\setW$ of weights
of cell modules introduced above and let greek letters ($\lambda,\mu,\nu$) denote the weights $[j,P]$ for simplicity.
The projective cover $\PrCA{\lambda}$
of a simple module $\IrCA{\lambda}$ has  the cell content
\begin{equation}\label{cell-cont-proj}
\left[\PrCA{\lambda}\right] = \bigoplus_{\nu\poless\lambda}\left[\StCA{\nu}:\IrCA{\lambda}\right]\StCA{\nu},
\end{equation}
where by $\left[\StCA{\nu}:\IrCA{\lambda}\right]$ we denote the
multiplicity of the appearance of the simple module $\IrCA{\lambda}$ in
$\StCA{\nu}$, i.e. $\left[\StCA{\nu}:\IrCA{\lambda}\right] =
\dim\Hom(\PrCA{\lambda},\StCA{\nu})$. The cell content~\eqref{cell-cont-proj} of the projective covers is actually
 a consequence of the general reciprocity
result $\left[\StCA{\nu}:\IrCA{\lambda}\right] =
\left[\PrCA{\lambda}:\StCA{\nu}\right]$ discussed above.
Note that using the
subquotient structure of the cell modules given in Fig.~\ref{fig:cell-thirdroot}
we see that all the numbers $\left[\StCA{\nu}:\IrCA{\lambda}\right]$ are zero or one.

Moreover, the projective module $\PrCA{\lambda}$ has a filtration by cell
modules respecting the cell-filtration of the cellular algebra, i.e. the
cell modules with the lowest weight (w.r.t to the order $\poless$)
are submodules; quotienting by them gives a module with cell
submodules corresponding to the next-to-lowest weight and so on.
This general property of the cell filtrations of projective covers allows now to describe
subquotient structure of  the projective covers $\PrCA{\lambda}$ in terms of  cell modules.
 We
assume  for simplicity that $N\;\textrm{mod}\;3=0$. Then,
the subquotient structure of $\PrJTL{L}{\q^{\pm 2}}$, where we set $N=2L$, in terms of cell modules, is as simple as on the
left of Fig.~\ref{fig:proj-thirdroot}.
Further, the cell filtration of the
projective covers $\PrJTL{j}{P}$, with $j\ne 0,2$,
is obtained from Fig.~\ref{fig:proj-thirdroot}
(taking $N=2L\geq M$)
as the  submodule generated from $\StJTL{j}{P}$.

Now, it is easy to see that the projective covers have  a
sequence of embeddings (opposite of the one in~\eqref{cell-emb})
\begin{equation}
\PrJTL{j_1}{P_1}\hookleftarrow\PrJTL{j_2}{P_2}\hookleftarrow\dots\hookleftarrow\PrJTL{j_n}{P_n}
\end{equation}
for any sequence of weights
$(j_1,P_1)\pomore(j_2,P_2)\pomore\dots\pomore(j_n,P_n)$ in
Fig.~\ref{fig:part-ord-thirdroot} with $j_k\ne 0,2$.
The projective
$\rJTL{N}$-modules $\PrJTL{1}{1}$ and $\PrJTL{0}{q^2}$ are isomorphic
to the cell modules $\StJTL{1}{1}$ and $\StJTL{0}{q^2}$, respectively;
the module $\PrJTL{2}{1}$ is given in the middle diagram in
Fig.~\ref{fig:proj-thirdroot}.

 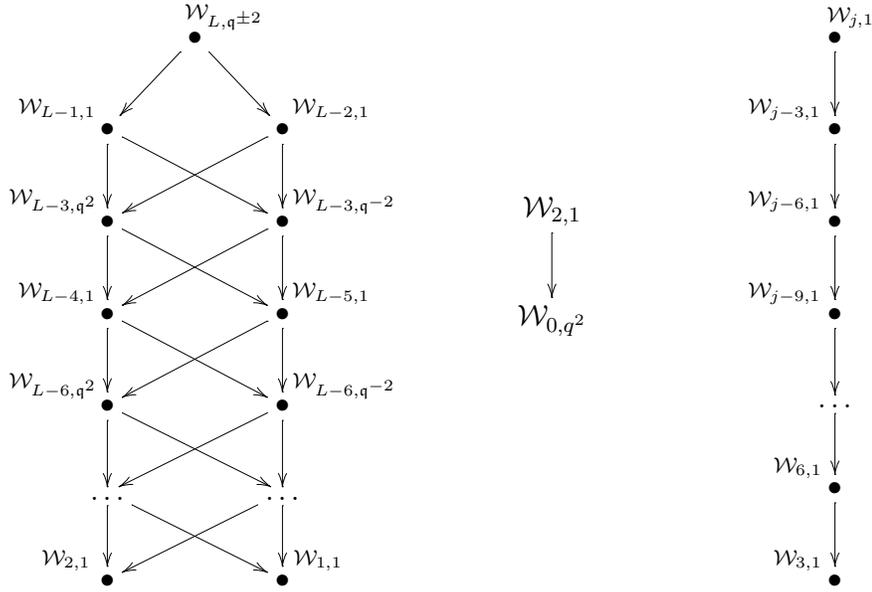
\begin{figure}\centering
 \begin{equation*}
   \xymatrix@R=24pt@C=18pt
   {{}&\bullet\ar@{}|{\substack{\StJTL{L}{\q^{\pm 2}}}\kern-7pt}[]+<15pt,15pt>\ar[dr]\ar[dl]&&\\
     {\bullet}\ar@{}|{\substack{\StJTL{L-1}{1}}\kern-7pt}[]+<-45pt,15pt>\ar[d]\ar[drr]
     &&{\bullet}\ar@{}|{\substack{\StJTL{L-2}{1}}\kern-7pt}[]+<30pt,15pt>\ar[d]\ar[dll]\\
     {\bullet}\ar@{}|{\substack{\StJTL{L-3}{\q^2}}\kern-7pt}[]+<-48pt,15pt>\ar[d]\ar[drr]
     &&{\bullet}\ar@{}|{\substack{\StJTL{L-3}{\q^{-2}}}\kern-7pt}[]+<40pt,15pt>\ar[d]\ar[dll]\\
     {\bullet}\ar@{}|{\substack{\StJTL{L-4}{1}}\kern-7pt}[]+<-45pt,15pt>\ar[d]\ar[drr]
     &&{\bullet}\ar@{}|{\substack{\StJTL{L-5}{1}}\kern-7pt}[]+<30pt,15pt>\ar[d]\ar[dll]\\
     {\bullet}\ar@{}|{\substack{\StJTL{L-6}{\q^2}}\kern-7pt}[]+<-48pt,15pt>\ar[d]\ar[drr]
     &&{\bullet}\ar@{}|{\substack{\StJTL{L-6}{\q^{-2}}}\kern-7pt}[]+<40pt,15pt>\ar[d]\ar[dll]\\
     {\dots}\ar[d]\ar[drr]
     &&{\dots}\ar[d]\ar[dll]\\
     {\bullet}\ar@{}|{\substack{\StJTL{2}{1}}\kern-7pt}[]+<-38pt,15pt>
     &&{\bullet}\ar@{}|{\substack{\StJTL{1}{1}}\kern-7pt}[]+<20pt,15pt>
     } \qquad
   \xymatrix@R=24pt@C=18pt
   {&&\\
     &&\\
     &\StJTL{2}{1}\ar[d]&\\
     &\StJTL{0}{q^2}&\\
     &&\\
     &&\\
     &&
     } \qquad
   \xymatrix@R=24pt@C=18pt
   {{}&\bullet\ar@{}|{\substack{\StJTL{j}{1}}\kern-7pt}[]+<5pt,15pt>\ar[d]&&\\
     &{\bullet}\ar@{}|{\substack{\StJTL{j-3}{1}}\kern-7pt}[]+<-45pt,15pt>\ar[d]&&\\
     &{\bullet}\ar@{}|{\substack{\StJTL{j-6}{1}}\kern-7pt}[]+<-45pt,15pt>\ar[d]&&\\
     &{\bullet}\ar@{}|{\substack{\StJTL{j-9}{1}}\kern-7pt}[]+<-45pt,15pt>\ar[d]&&\\
     &{\dots}\ar[d]&&\\
     &{\bullet}\ar@{}|{\substack{\StJTL{6}{1}}\kern-7pt}[]+<-35pt,15pt> \ar[d]&&\\
     &{\bullet}\ar@{}|{\substack{\StJTL{3}{1}}\kern-7pt}[]+<-35pt,15pt>
     }
 \end{equation*}
      \caption{The cell-filtraion of projective $\rJTL{N}(1)$-modules
      $\PrJTL{j}{P}$. The two modules on the left are projective
      covers from the doubly critical class: a projective cover
      $\PrJTL{j}{P}$, with $j\ne 0,2$, is the submodule generated from
      $\StJTL{j}{P}$ on the left-most diagram given for
      $2L\;\textrm{mod}\;3=0$; the second module is $\PrJTL{2}{1}$;
      the right-most projective module corresponds to the
      singly critical class with $j\;\textrm{mod}\;3=0$.}
    \label{fig:proj-thirdroot}
    \end{figure}

The diagrams allow us to conclude that the indecomposable projectives
$\PrJTL{j}{P}$ (excepting $\PrJTL{2}{1}$ for $N=6$) are not self-dual
because their socles contain the direct sum
$\IrJTL{L}{1}\oplus\IrJTL{L}{1}$ when $L\;\textrm{mod}\;3=1$ or
$2\IrJTL{L}{1}\oplus2\IrJTL{L-1}{1}$ when $L\;\textrm{mod}\;3=2$ or
the sum $2\IrJTL{L}{\q^2}\oplus2\IrJTL{L}{\q^{-2}}$ when
$L\;\textrm{mod}\;3=0$. Therefore, the projective modules are not injective hulls
and can be  embedded into larger  (of course decomposable) modules.

We will see below that the spin-chain representation is a self-dual JTL module and
is decomposed onto a special class of so-called tilting modules. On the other hand
 it is a faithful representation and thus all projective covers should appear in the spin-chain,
though not as direct summands. It turns out that the projective covers are embedded into  a direct
sum of tilting (self-dual) modules.

\newcommand{\repgll}{\pi}

\subsection{Self-duality of the spin chain representation $\chVv$}\label{sec:self dual}

In this section, we show that the faithful representation
of $\rJTL{N}(1)$ on the periodic $s\ell(2|1)$ spin chain $\chVv$ is in addition self-dual.
 To show this we use a non-degenerate bilinear form $(\cdot,\cdot)$ on
$\chVv\times\chVv$, which can be  given explicitly in the Fock space
realization~\cite{ReadSaleur07-1}.
Recall that the representation and even a large family of representations $\repgl$ of
$\rJTL{N}(m)$ are explicitly defined in App.~\ref{sec:TL-faith} using this Fock realization. Our case corresponds to $\pi_{1,1}$ which we will denote simply by $\pi$.
The generators $e_j$ of
$\rJTL{N}(1)$ are self-adjoint with respect to the bilinear form,
i.e. $\repgll(e_j)^{\dagger}=\repgll(e_j)$. The
adjoint of the translation operator $\repgll(u^2)$ is the inverse
$\repgll(u^{N-2})$. Together with non-degeneracy of the bilinear form,
this means that the representation $\repgll$ is isomorphic to the dual
one on the space $\chVv^*$ of linear functionals. Indeed an
isomorphisms $\psi$ between $\rJTL{N}(1)$-modules $\chVv$ and $\chVv^*$
is given by
\begin{equation}
\psi: \chVv\to \chVv^*, \qquad \psi(v)(\cdot) = (v,\cdot),
\end{equation}
where the $\rJTL{N}(1)$-action on $\chVv^*$ is defined as
\begin{equation*}
A v^*(\cdot)
= v^*(\repgll(A)^{\dagger}\,\cdot),\quad \text{with}\quad A\in\rJTL{N}(1),
\; v^*\in\chVv^*,
\end{equation*}
 and
$\cdot$ stands for an argument. The non-degeneracy of the bilinear
 form implies that the kernel of $\psi$ is zero. The statement on the self-duality is obviusly true in the general case of $\repgl$ representations.

Self-duality of the module $\chVv$ implies that the subquotient
structure (with simple subquotients) is not affected by reversing all the arrows representing the
$\rJTL{N}$-action. On the other hand, the faithfullness of the
representation of $\rJTL{N}(1)$ implies that all the projective covers of simples
should be present in the spin-chain decomposition. As we saw above, the projective $\rJTL{N}(1)$-modules (those which are not simple) are not self-dual and
therefore they are not injective modules and can in principle be non-direct-summand
submodules in some `bigger' modules. These bigger and self-dual $\rJTL{N}$-modules
indeed exist: they turn up to be  \textsl{tilting} modules, and we show below that the projectives can be
embedded into a direct sum of tilting modules.
We can thus say that, for our needs at least,  the tilting
modules are more
fundamental objects than projective modules in the sense that tilting modules
are the  building blocks (direct summands)
of the spin-chains.

\subsection{Tilting modules over $\rJTL{N}(1)$}\label{sec:tilt}
Recall that a tilting module $\mathcal{T}$ is a module (over a cellular algebra) with a
filtration by cell modules $\mathcal{W}$ and with the dual filtration by the duals to the  cell
modules -- the co-cell (or costandard)
modules $\mathcal{W}^*$. The tilting modules are thus
self-dual modules by this definition, i.e. $\mathcal{T}^*\cong\mathcal{T}$.

 For any weight
$(j,P)\in\setW$, we define a tilting $\JTL{N}$-module $\TilJTL{j}{P}$ \textit{generated}
from the cell module $\StJTL{j}{P}$ as \textsl{the indecomposable} tilting module
containing this cell module as a submodule. We will see that if such a tilting module exists this property indeed
fixes it uniquely, see Thm. \ref{thm:tilt-mod-key} below. In most cases, our results on the structure of the tilting modules generated from a cell module can be expressed by very simple rules: to construct $\TilJTL{j}{P}$ we take the cell module $\StJTL{j}{P}$ and each simple subquotient $\IrrJTL{j'}{P'}$ (in the diagram for $\StJTL{j}{P}$) replace by its co-cell module $\StJTL{j'}{P'}^*$ such that different co-cells are glued following the pattern for the diagram for the original cell module $\StJTL{j}{P}$.
The result is then obviously a module with a filtration by the duals to cell modules. We then check the dual filtration and the module is what we call the tilting module $\TilJTL{j}{P}$.

It will be shown below that tilting modules $\TilJTL{j}{P}$ indeed exist and exhaust all possible indecomposable tilting modules (of course up to an isomorphism). We will now describe our results on the subquotient structure.
The diagrams in Fig.~\ref{fig:tilt-thirdroot} describe  cell and
co-cell filtrations of the `biggest' tilting modules $\TilJTL{1}{1}$
and $\TilJTL{2}{1}$. The structure of $\TilJTL{1}{1}$ is given just by the substitution $2\to 1$.  The
two tilting modules correspond to the cases when $L\;\textrm{mod}\;3$ equals
$2$ or $0$:
in the first case $L_i=L-2 + i$ and $P_i=1$, in the second -- $L_i=L$, $P_1=\q^2$, and $P_2=\q^{-2}$.
The tilting module $\TilJTL{2}{1}$ in the
case $L\;\textrm{mod}\;3=1$ is obtained by identifying the two nodes
$\StJTL{L_1}{P_1}$ and $\StJTL{L_2}{P_2}$ with $\StJTL{L}{1}$ (with their
two arrows) at the top on the left and similarly for duals at the bottom on the right.

 An important consistency check for the existence of these tilting modules (in addition to their
self-duality) is that the projective modules $\PrJTL{j}{P}$ defined
above in Fig.~\ref{fig:proj-thirdroot} cover any cell-subquotient in
the tiltings, where the kernel of the projection is isomorphic to
$\StJTL{1}{1}$. The tiltings $\TilJTL{1}{1}$ and $\TilJTL{2}{1}$
themselves can be alternatively obtained as a quotient of the
projective module $\PrJTL{L_1}{P_1}\oplus\PrJTL{L_2}{P_2}$.

 \begin{figure}\centering
 \begin{equation*}
   \xymatrix@R=24pt@C=18pt
   {
     {\bullet}\ar@{}|{\substack{\StJTL{L_1}{P_1}}\kern-7pt}[]+<-45pt,15pt>\ar[d]\ar[drr]
     &&{\bullet}\ar@{}|{\substack{\StJTL{L_2}{P_2}}\kern-7pt}[]+<30pt,15pt>\ar[d]\ar[dll]\\
     {\dots}\ar[d]\ar[drr]
     &&{\dots}\ar[d]\ar[dll]\\
     {\bullet}\ar@{}|{\substack{\StJTL{6}{\q^2}}\kern-7pt}[]+<-48pt,15pt>\ar[d]\ar[drr]
     &&{\bullet}\ar@{}|{\substack{\StJTL{6}{\q^{-2}}}\kern-7pt}[]+<40pt,15pt>\ar[d]\ar[dll]\\
     {\bullet}\ar@{}|{\substack{\StJTL{4}{1}}\kern-7pt}[]+<-45pt,15pt>\ar[d]\ar[drr]
     &&{\bullet}\ar@{}|{\substack{\StJTL{5}{1}}\kern-7pt}[]+<30pt,15pt>\ar[d]\ar[dll]\\
     {\bullet}\ar@{}|{\substack{\StJTL{3}{\q^2}}\kern-7pt}[]+<-48pt,15pt>\ar[dr]
     &&{\bullet}\ar@{}|{\substack{\StJTL{3}{\q^{-2}}}\kern-7pt}[]+<40pt,15pt>\ar[dl]\\
     &{\bullet}\ar@{}|{\substack{\StJTL{2}{1}}\kern-7pt}[]+<-40pt,5pt> &
     }
   \xymatrix@R=24pt@C=0pt
   {
     &&\\
     &&\\
     &\quad\cong&\\
     &&\\
     &&\\
     &&
     }
   \xymatrix@R=24pt@C=18pt
   {{}&\bullet\ar@{}|{\substack{\StJTL{2}{1}^*}\kern-7pt}[]+<15pt,15pt>\ar[dr]\ar[dl]&&\\
     {\bullet}\ar@{}|{\substack{\StJTL{3}{\q^2}^*}\kern-7pt}[]+<-38pt,15pt>\ar[d]\ar[drr]
     &&{\bullet}\ar@{}|{\substack{\StJTL{3}{\q^{-2}}^*}\kern-7pt}[]+<30pt,15pt>\ar[d]\ar[dll]\\
     {\bullet}\ar@{}|{\substack{\StJTL{4}{1}^*}\kern-7pt}[]+<-35pt,15pt>\ar[d]\ar[drr]
     &&{\bullet}\ar@{}|{\substack{\StJTL{5}{1}^*}\kern-7pt}[]+<20pt,15pt>\ar[d]\ar[dll]\\
     {\bullet}\ar@{}|{\substack{\StJTL{6}{\q^2}^*}\kern-7pt}[]+<-38pt,15pt>\ar[d]\ar[drr]
     &&{\bullet}\ar@{}|{\substack{\StJTL{6}{\q^{-2}}^*}\kern-7pt}[]+<30pt,15pt>\ar[d]\ar[dll]\\
     {\dots}\ar[d]\ar[drr]
     &&{\dots}\ar[d]\ar[dll]\\
     {\bullet}\ar@{}|{\substack{\StJTL{L_1}{P_1}^*}\kern-7pt}[]+<-48pt,15pt>
     &&{\bullet}\ar@{}|{\substack{\StJTL{L_2}{P_2}^*}\kern-7pt}[]+<30pt,15pt>
     }
 \end{equation*}
      \caption{The structure of the tilting $\rJTL{N}$-module
      $\TilJTL{2}{1}$ with the co-cell modules filtration on the
      right. The two isomorphic modules correspond to $L_i=L-2 + i$ and
      $P_i=1$ if $L\;\textrm{mod}\;3=2$, and $L_i=L$, $P_1=\q^2$,
      and $P_2=\q^{-2}$ when $L\;\textrm{mod}\;3=0$; the tilting
      module $\TilJTL{2}{1}$ in the case $L\;\textrm{mod}\;3=1$ is
      obtained by identifying the two nodes
      $\StJTL{L_1}{P_1}=\StJTL{L_2}{P_2}$ with $\StJTL{L}{1}$ and the two
      arrows at the top on the left and their duals at the bottom on
      the right. The structure of $\TilJTL{1}{1}$ is given just by the substitution $2\to 1$. }
    \label{fig:tilt-thirdroot}
    \end{figure}
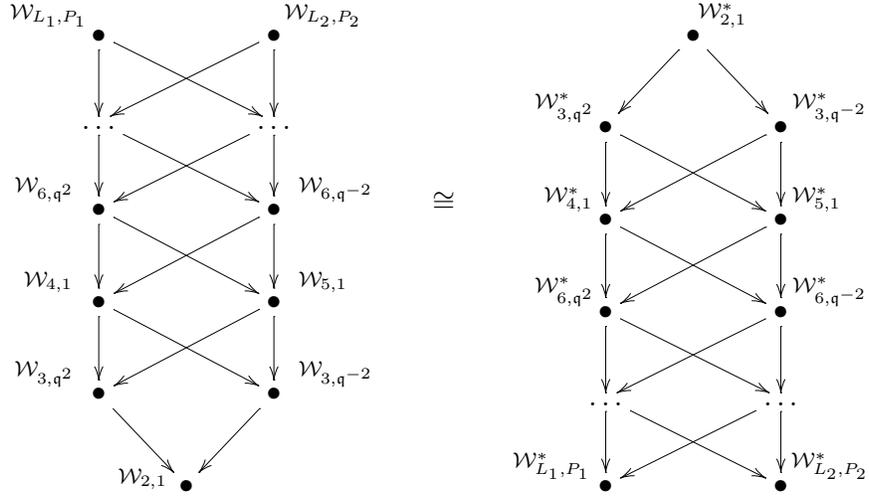

A tilting module $\TilJTL{j}{P}$ with $j>2$ and appropriate $P$ can be extracted also from
Fig.~\ref{fig:tilt-thirdroot}. It is the submodule generated from the
subquotient $\StJTL{j}{P}^*$ on the right diagram or, dually, it
is the corresponding quotient
containing all $\StJTL{j'}{P'}$, with $[j',P']\pomore[j,P]$, on
the left diagram.
The tilting modules from
the singly critical class with $j\;\textrm{mod}\;3=0$ and $P=1$ are
constructed in a similar way and they are of chain type.

A peculiarity happens with the tilting $\rJTL{N}$-module containing the cell
submodule $\bStJTL{0}{\q^2}$. If we were to proceed with the construction
used so far, we would get a non self-dual module with the subquotient structure
$\bStJTL{0}{\q^2}^*\to\StJTL{2}{1}^*$,  so it does not
work in this case (meanwhile, the construction would still  work for the former algebra
$\JTL{N}(1)$). However, there turns out to be  a universal construction of tilting modules~\cite{AR,Ringel} that requires knowledge of extension groups associated with indecomposable modules. This construction suggests  to begin with a glueing (an extension)  corresponding to an exact non-split sequence
\begin{equation}
0\to \bStJTL{0}{\q^2} \to M \to \StJTL{2}{1} \to 0,
\end{equation}
{\it i.e.}, the first step of the universal construction produces a module that has a submodule isomorphic to $\bStJTL{0}{\q^2}$ and as a  top  $\StJTL{2}{1}$ with multiplicity given by dimension of $\ExtJTL(\StJTL{2}{1},\bStJTL{0}{\q^2})$. We know from the structure of the projective cover for $\StJTL{2}{1}$ (see above) that this dimension equals $1$. In the second step, we should compute  the extension group $\ExtJTL(\StJTL{3}{\q^{\pm2}},M)$. We checked on the first few values of $N$ that a $\rJTL{N}$-module having $M$ as a submodule and with a cell filtration containing sections isomorphic to $\StJTL{3}{\q^{\pm2}}$ can not be self-dual, and actually can not have a filtration by duals to cell modules. For example, it is quite easy to see that at $N=6$ a $\rJTL{N}$-module with the subquotient structure $(n_+\StJTL{3}{\q^{2}}\oplus n_-\StJTL{3}{\q^{-2}})\to M$ is self-dual if and only if  $n_{\pm}=0$ (see also the example below). Hence, we can conjecture that the extension groups $\ExtJTL(\StJTL{3}{\q^{\pm2}},M)$ are trivial. It turns out that for $N=8$, the only self dual module having cell and co-cell filtrations is the one having the subquotient structure $\StJTL{4}{1}\to M$ (see examples for $N=8$ below). We call this module $M'$. Then, we repeat our analysis for $N=10$ and we get a module $M'' = \StJTL{5}{1}\to M'$, etc.
So, the only way we see that the tilting $\rJTL{N}$-module $\TilJTL{0}{\q^2}$
containing $\bStJTL{0}{\q^2}$ as a submodule might have  cell as well as co-cell filtrations is the one given on Fig.~\ref{fig:tilt-JTL-vac}.
 \begin{figure}\centering
\begin{equation*}
   \xymatrix@R=24pt@C=0pt
   {
     &&\\
     &&\\
     &\TilJTL{0}{\q^2} \;=\; &\\
     &&\\
     &&\\
     &&
     }
   \xymatrix@R=24pt@C=18pt
   {
     {\bullet}\ar@{}|{\substack{\StJTL{L_1}{P_1}}\kern-7pt}[]+<-45pt,15pt>\ar[d]\ar[drr]
     &&{\bullet}\ar@{}|{\substack{\StJTL{L_2}{P_2}}\kern-7pt}[]+<30pt,15pt>\ar[d]\ar[dll]\\
     {\dots}\ar[d]\ar[drr]
     &&{\dots}\ar[d]\ar[dll]\\
     {\bullet}\ar@{}|{\substack{\StJTL{6}{\q^2}}\kern-7pt}[]+<-48pt,15pt>\ar[d]\ar[drr]
     &&{\bullet}\ar@{}|{\substack{\StJTL{6}{\q^{-2}}}\kern-7pt}[]+<40pt,15pt>\ar[d]\ar[dll]\\
     {\bullet}\ar@{}|{\substack{\StJTL{4}{1}}\kern-7pt}[]+<-45pt,15pt>\ar[dr]
     &&{\bullet}\ar@{}|{\substack{\StJTL{5}{1}}\kern-7pt}[]+<30pt,15pt>\ar[dl]\\
     &{\bullet}\ar@{}|{\substack{\StJTL{2}{1}}\kern-7pt}[]+<-40pt,5pt>
   \ar[d]&\\
     &{\bullet}\ar@{}|{\substack{\bStJTL{0}{\q^2}}\kern-7pt}[]+<-40pt,5pt> &
     }
   \xymatrix@R=24pt@C=0pt
   {
     &&\\
     &&\\
     &\quad\cong&\\
     &&\\
     &&\\
     &&
     }
   \xymatrix@R=24pt@C=18pt
   {{}&\bullet\ar@{}|{\substack{\bStJTL{0}{\q^2}^*}\kern-7pt}[]+<15pt,15pt>\ar[d]&&\\
     {}&\bullet\ar@{}|{\substack{\StJTL{2}{1}^*}\kern-7pt}[]+<25pt,15pt>\ar[dr]\ar[dl]&&\\
     {\bullet}\ar@{}|{\substack{\StJTL{4}{1}^*}\kern-7pt}[]+<-35pt,15pt>\ar[d]\ar[drr]
     &&{\bullet}\ar@{}|{\substack{\StJTL{5}{1}^*}\kern-7pt}[]+<20pt,15pt>\ar[d]\ar[dll]\\
     {\bullet}\ar@{}|{\substack{\StJTL{6}{\q^2}^*}\kern-7pt}[]+<-38pt,15pt>\ar[d]\ar[drr]
     &&{\bullet}\ar@{}|{\substack{\StJTL{6}{\q^{-2}}^*}\kern-7pt}[]+<30pt,15pt>\ar[d]\ar[dll]\\
     {\dots}\ar[d]\ar[drr]
     &&{\dots}\ar[d]\ar[dll]\\
     {\bullet}\ar@{}|{\substack{\StJTL{L_1}{P_1}^*}\kern-7pt}[]+<-48pt,15pt>
     &&{\bullet}\ar@{}|{\substack{\StJTL{L_2}{P_2}^*}\kern-7pt}[]+<30pt,15pt>
     }
\end{equation*}
      \caption{The structure of the tilting $\rJTL{N}$-module
      $\TilJTL{0}{\q^2}$ with its filtration by cell modules on the left side and the co-cell modules filtration  on the
      right.}
    \label{fig:tilt-JTL-vac}
    \end{figure}
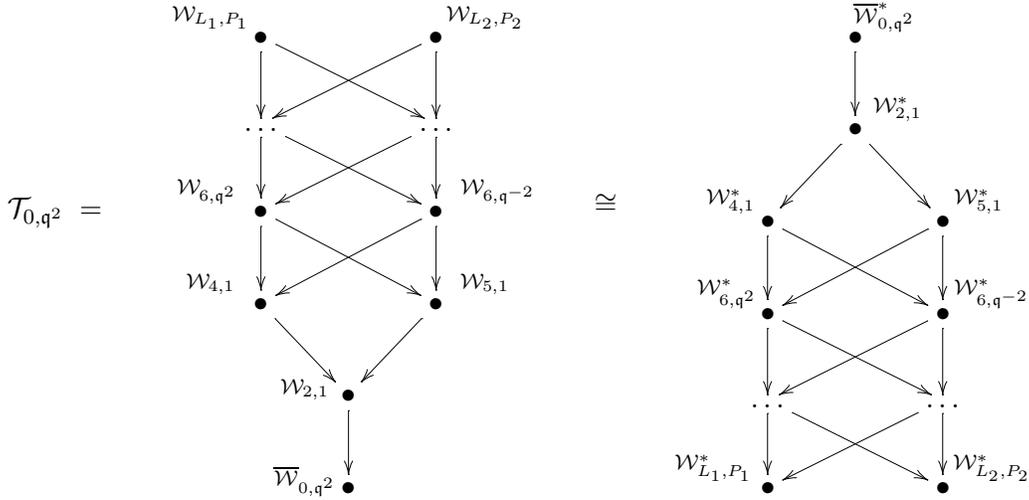
We checked for $N\leq 18$ that a module with this cell filtration is self-dual and we believe that its existence can be proved by  taking a quotient of the
projective module
$\PrJTL{L_1}{P_1}\oplus\PrJTL{2}{1}\oplus\PrJTL{L_2}{P_2}$ (see
definition of $L_i$, $P_i$ above).
Note also that this module contains the projective cover $\PrJTL{2}{1}$,
which is an important module for our spin-chain.

\medskip

Finally, we claim that the tilting modules  just described  \textit{exhaust} all indecomposable tilting
modules over $\rJTL{N}(1)$.
To show this we  use the important result about the $\rJTL{N}(1)$ algebra that it is a \textsl{quasi-hereditary} algebra. We first recall the corresponding definition~\cite{DlabR, Donkin, tilt-book}.

\medskip
\textbf{Definition}
\textit{Let $A$ be a finite dimensional associative algebra over $\mathbb{C}$, $\setW$ an indexing set for the isomorphism classes of simple $A$-modules with correspondence $\lambda \to \mathcal{X}_{\lambda}$, and $\leq$ a partial order on $\setW$. We say that $(A,\leq)$ or simply $A$ is a quasi-hereditary algebra if and only if for all $\lambda \in \setW$ there exists a left $A$-module, $\mathcal{W}_{\lambda}$, called a standard module such that}
\begin{itemize}
\item
\textit{ there is a surjection $\mathcal{W}_{\lambda} \to \mathcal{X}_{\lambda}$ and the composition factors (subquotients), $\mathcal{X}_{\mu}$, of the kernel satisfy $\mu < \lambda$.}

\item
\textit{
the indecomposable projective cover $\mathcal{P}_{\lambda}$ of $\mathcal{X}_{\lambda}$ maps onto $\mathcal{W}_{\lambda}$ via a map $\psi : \mathcal{P}_{\lambda} \to \mathcal{W}_{\lambda}$ whose kernel is filtered by modules $\mathcal{W}_{\lambda}$ with $\mu > \lambda$.}
\end{itemize}

In our setting, the cell modules $\StJTL{j}{z^2}$ for $\rJTL{N}(1)$ are the standard modules with the weight $\lambda=(j,z^2)$.
Using our results on the projective JTL modules described above and choosing for  the quasi-hereditarity partial order $\leq$  the one opposite to the cellular partial order $\poless$, {\it i.e.}, if  $\leq=\pomore$, we readily see that $\rJTL{N}(1)$ is a quasi-hereditary algebra (it is actually quasi-hereditary for any $\q\ne i$). Then, as for any quasi-hereditary algebra~\cite{AR,Ringel}, we have a  key theorem on tiliting JTL modules\footnote{We note that a similar theorem appears for reductive algebraic groups over a finite field~\cite{Math}.}:

 \begin{Thm}\label{thm:tilt-mod-key}\mbox{}
 \begin{itemize}
 \item
 For any weight $[j,P]\in\setW$, there is a unique indecomposable tilting module $\TilJTL{j}{P}$ such that
 $\left[\TilJTL{j}{P}:\StJTL{j}{P}\right]=1$ and $\StJTL{j}{P}$ is a submodule in $\TilJTL{j}{P}$, and
 \begin{equation}\label{eq:tilt-criteria}
\left[\TilJTL{j}{P}:\StJTL{j'}{P'}\right]\ne 0
\qquad \text{only if}\quad[j',P']\pomore[j,P];
 \end{equation}
 \item Any indecomposable tilting module is isomorphic to some $\TilJTL{j}{P}$.
 \end{itemize}
 \end{Thm}

This theorem guarantees existence and uniqueness of indecomposable tilting $\rJTL{N}$-modules generated from a cell module. It also gives a good criteria on whether a given  cell module appears in the cell filtration of a tilting module: if $[j',P']\prec[j,P]$ then the multiplicity $\left[\TilJTL{j}{P}:\StJTL{j'}{P'}\right]=0$. Using also the subquotient structure of projective covers for cell modules, it is easy to see that the numbers $\left[\TilJTL{j}{P}:\StJTL{j'}{P'}\right]$
 are less than $2$ ($0$ or $1$) if $[j',P']\succ[j,P]$. Further, the structure for $\TilJTL{j}{P}$
modules we proposed above does give a filtration by cell and duals to the cell modules, as indicated in Fig.~\ref{fig:tilt-thirdroot} and Fig.~\ref{fig:tilt-JTL-vac}.
To prove that these modules for $j>0$ are indeed indecomposable we observe, using the universal construction as we did for $\TilJTL{0}{\q^2}$ above, that if one of the multiplicities $\left[\TilJTL{j}{P}:\StJTL{j'}{P'}\right]$ is $0$, for $[j',P']\succ[j,P]$, the corresponding module would not be able to have a filtration by co-cell modules. We thus have  the criteria \eqref{eq:tilt-criteria} with the stronger condition ``if and only if''. By  uniqueness we finally obtain
that the modules with the subquotient structure proposed  are indecomposable\footnote{We have to note that our definition for tilting modules is slightly stronger than the  one for quasi-hereditary algebras~\cite{AR,Ringel}. Tiltings are usually not required to have necessary
   {\sl dual} or opposite filtration by co-standrard modules, just any filtration by co-standards.  It turned out that our tiltings do have a filtration by co-standards opposite to the one by the standard or cell modules.}.

\subsection{Centralizer and formal decomposition}
\label{secsl21resuts}

In our previous works~\cite{GRS1,GRS2,GRS3} on the $\gl(1|1)$ spin chain, the analysis of the spin chain and the scaling limit properties was  based on the structure of the centralizer of the $\rJTL{N}(0)$ algebra, dubbed $\LQGodd$. The centralizer was found in a rather straightforward way as a proper subalgebra of the open-case centralizer -- $\LQG$ at $\q=i$. The subalgebra $\LQGodd$ was identified with appropriate polynomials in the finite number of generators of~$\LQG$. Unfortunately, the centralizer $\mathcal{A}_{2|1}(N)$   for the open $s\ell(2|1)$ spin-chains is not explicitly described in terms of generators and defining relations: it is not $\LQG$ at $\q=e^{i\pi/3}$ but only a Morita equivalent algebra. At  best, we have only its cellular basis description~\cite{ReadSaleur07-1} and it is complicated  to identify in a straightforward way a subalgebra in $\mathcal{A}_{2|1}(N)$ that commutes only with elements from $\rJTL{N}(1)$. The lack of suitable description of the open-case centralizer thus makes the centralizer approach used in the $\gl(1|1)$ case less promising in the present case of faithful representations of JTL algebras.

Although the faithfulness makes, at first glance, the analysis of the periodic spin chains much harder than in the non-faithful case of $\gl(1|1)$ spin chains, it also provides  many advantages. First of all, we know what kind of `complexity' of JTL modules we might expect in the decompositions. They should be as complicated as the projective covers $\PrJTL{j}{z^2}$ described above. So, the structure of the periodic $s\ell(2|1)$ spin chain is apparently much more involved than the one described in~\cite{GRS2}. Further, we have seen that projective covers are still not  good candidates for direct summands in the spin chain, since they are non-self dual. The best candidates are the tilting JTL modules which contain the projective covers as submodules. Once again, each indecomposable tilting module is self-dual and `smaller' than a projective cover but we can always find two non-isomorphic tilting modules that contain the projective cover in their direct sum, as  follows from their structure. It is thus quite reasonable to expect that the periodic spin chain, as a {\sl self-dual} and {\sl faithful} $\rJTL{N}$-module, is decomposed onto tilting modules.

 We should emphasize that the important missing step in this new analysis of periodic spin chains based on the theory of tilting modules is a proof\footnote{We believe that one could repeat steps in the Martin's paper~\cite{MartinFaithful2} for the mirror spin-chains which are representations of the blob algebra (also some quotient of affine TL). The crucial technical problem is that one should find a proper embedding of $\JTL{N}$ into $\JTL{N+1}$ for any $N$.}
 that our spin-chain is a full tilting module. To show this point one should prove that the periodic $s\ell(2|1)$ spin chain indeed has a filtration by cell modules: using the fact that all $e_j$'s are self-adjoint operators, it would be then straightforward  to show that the full spin chain is a full tilting module. There are strong arguments suggesting that the spin chain has a filtration by cell modules. In particular, the spin chain is deeply related with a generic loop model, which is defined for arbitrary values of  $\q$~\cite{ReadSaleur01}. In this model, non contractible loops   get the weight $m$ if they wind around the small (space) cycle, and a modified weight $\q'+\q'^{-1}$ if they wind around the long (imaginary time) cycle (the particular case we are interested in corresponds to $m=1$ and $\q'+\q'^{-1}=3$).  For generic values of $\q$, it is possible to show that the generating function of levels expands only on traces over cell modules, so the corresponding `spin chain' has a filtration by cell modules indeed. However, this argument cannot be made rigorous because, for $\q$ generic, there is in fact no such spin chain, and the generating function of levels is only a formal object. Hence the existence of a filtration by cell modules cannot be proven by `analytic continuation' (see the appendix of~\cite{[ABNKS]}, however).

 \textit{Assuming this crucial assumption about cell-modules filtration however, we see  that  the spin-chain itself should be a
full tilting module}, and the decomposition, up to multiplicities, follows from the foregoing discussion of these modules.

The next step in our analysis is to obtain the multiplicity of each tiliting module. Since we know the structure of all tilting modules in terms of standard modules, and taking into account Thm.~\ref{thm:tilt-mod-key},
 we could compute each of these multiplicities iteratively if  we knew the numbers $[\Hilb_N:\AStTL{j}{z^2}]$ of appearance of each JTL
 standard module in the full spin-chain. For example, the multiplicity of $\TilJTL{0}{\q^2}$ equals $[\Hilb_N:\bAStTL{0}{\q^2}]$ because
 the subquotient isomorphic to $\bAStTL{0}{\q^2}$ is contained only in $\TilJTL{0}{\q^2}$. Then  the multiplicity of $\TilJTL{2}{1}$
 equals $[\Hilb_N:\AStTL{2}{1}]-[\Hilb_N:\bAStTL{0}{\q^2}]$ because  subquotients isomorphic to $\AStTL{2}{1}$ were already counted in the $\TilJTL{0}{\q^2}$ modules $[\Hilb_N:\bAStTL{0}{\q^2}]$ times, {\it etc.}

 The numbers $[\Hilb_N:\AStTL{j}{z^2}]$ can in fact be computed assuming the
 possibility  of an analytical continuation from semi-simple cases where these numbers are known~\cite{ReadSaleur07-1}.
We have an infinite family of super-symmetric spin chains defined in the same way as in the $s\ell(2|1)$ case but with the $s\ell(n+m|n)$ symmetry such that $m>2$ and $n\geq1$, see App.~\ref{sec:TL-faith}. The nearest neighbour coupling or the projection onto a singlet for these spin chains also defines a faithful representation of the $\rJTL{N}(m)$ algebra with $\q+\q^{-1}=m>2$.
The decomposition of the spin chains for these cases can be written as\footnote{We note that it is known~\cite{Jones} that the JTL algebras are semisimple when $\q+\q^{-1}>2$, though the Hamiltonian is not critical for these values of $\q$.}
   \begin{eqnarray}
       {\Hilb_N}=\bAStTL{0}{\q^2}\oplus\bigoplus_{j>0}~\hat{D}'_{j,0}~\AStTL{j}{1}\oplus\bigoplus_{\substack{j>0, M>1\\M|j}}~\bigoplus_{\substack{0<p<M\\
       p\wedge M=1}}\hat{D}'_{j,\frac{\pi p}{M}}~\AStTL{j}{e^{2i\pi p/M}}
	   \end{eqnarray}
 where we have used the notation $\hat{D}'_{j,K}$, with $K=\pi p/M$, for  the dimensions of the irreducible representations of the centralizer $\hat{\mathcal{A}}_{n+m|n}(2L)$. This centralizer is discussed in~\cite{ReadSaleur07-1} for the semi-simple cases. We will discuss it briefly below, and for now just
  recall the dimensions
\begin{equation}\label{hatD}
\hat{D}_{j,K}'={1\over j}\sum_{r=0}^{j-1}e^{2iKr}w(j,j\wedge r),
\end{equation}
where $j\wedge r$ is the greatest common divisor, $j\wedge 0=j$ and
\begin{equation}
w(j,d)=\left(\q^{2d}+\q^{-2d}\right)\delta_{j/d\equiv 0}+\left(\q'^{2d}+\q'^{-2d}\right)\delta_{j/d\equiv 1},
\end{equation}
with congruences being taken modulo $2$ and we set $\q'+\q'^{-1}=m+2n$.

Let us give a few examples of these multiplicities. For the case $2L=6$,
the decomposition of the semisimple  Hilbert space reads
\begin{multline}\label{Hilb-dec-form}
\Hilb_6 = \bAStTL{0}{\q^2} \oplus \hat{D}'_{1,0}\AStTL{1}{1} \oplus \hat{D}'_{2,0}\AStTL{2}{1} \oplus \hat{D}'_{3,0}\AStTL{3}{1} \oplus \hat{D}'_{2,\frac{\pi}{2}} \AStTL{2}{-1} \\\oplus \hat{D}'_{3,\frac{\pi}{3}}\AStTL{3}{e^{2i\pi/3}} \oplus \hat{D}'_{3,\frac{2\pi}{3}}\AStTL{3}{e^{4i\pi/3}}.
\end{multline}

We can actually think  about the expression~\eqref{hatD} formally as a polynomial in the complex variable~$\q$.
We note then that the multiplicities are well defined in the critical cases when $\q+\q^{-1}$ equals $0$, $1$ or~$2$, {\it i.e.}, the polynomials $\hat{D}_{j,K}'$ give positive integer numbers for $\q=i$, $e^{i\pi/3}$ and $1$ as well. Of course, we will not have a direct sum decomposition as in~\eqref{Hilb-dec-form} at these critical values of $\q$. Some of the direct summands become reducible but indecomposable and they are glued with other direct summands in~\eqref{Hilb-dec-form} in order to make a self-dual module (recall that our  space of states $\Hilb_N$ is always a self-dual module). We should thus think about  the number $\hat{D}'_{j,K}$ as the number of appearance of (subquotients isomorphic to) $\AStTL{j}{e^{2i K}}$ in the full space of states $\Hilb_N$. In other words we assume the equality
\begin{equation}
[\Hilb_N:\AStTL{j}{e^{2i K}}] = \hat{D}'_{j,K}.
\end{equation}
It will be shown below (by a numerical analysis) that this assumption indeed gives  correct results on the number of higher-rank Jordan cells, in particular.

 So, in the $s\ell(n+1|n)$ case we get the following filtration by cell modules $\StJTL{j}{z^2}$, setting $l\equiv 2n+1$:
\begin{eqnarray}
\bigl[\Hilb_6\bigr]=\bAStTL{0}{\q^2}+(l^2-1)~\AStTL{1}{1}+\half(l^4-4l^2+1)~\AStTL{2}{1}+\ffrac{1}{3}(l^6-6l^4+11l^2-6)~\AStTL{3}{1}\nonumber\\
+\half(l^4-4l^2+3) \AStTL{2}{-1}+\ffrac{1}{3}(l^6-6l^4+8l^2) \left(\AStTL{3}{e^{2i\pi/3}}+\AStTL{3}{e^{4i\pi/3}}\right),
\end{eqnarray}
where by $\bigl[\Hilb_N\bigr]$ we denote a formal decomposition on cell modules modulo glueings (this is why we do note use here the direct sum `$\oplus$' symbol)\footnote{More formally, $\bigl[\Hilb_N\bigr]$ is the image of $\Hilb_N$ in the Grothendieck group $Gr_N$ generated by cell modules:  let $F_N$ be the free abelian group with generators the isomorphism classes of $\rJTL{N}$-modules filtered by cell modules, and let $[V]$ be the element of  $F_N$ corresponding to a module $V$, then $Gr_N$ is the quotient of $F_N$ by the relations $[W]= [U]+[V]$ for all short exact sequences $0 \to U\to W \to V\to0$
 of $\rJTL{N}$-modules $U$, $V$, $W$ having cell filtration. We also abused notations denoting $[\AStTL{j}{K}]$ simply by $\AStTL{j}{K}$.}.

For $l=3$ meanwhile (that is, the $s\ell(2|1)$ spin-chain) we get
\begin{eqnarray}\label{Hsix}
\bigl[\Hilb_6\bigr]=\bAStTL{0}{\q^2}+8~\AStTL{1}{1}+23~\AStTL{2}{1} +24 \AStTL{2}{-1} +112~\AStTL{3}{1}
+105 \left(\AStTL{3}{e^{2i\pi/3}}+\AStTL{3}{e^{4i\pi/3}}\right).
\end{eqnarray}

Having these filtrations (or formal decompositions) by cell modules, we already see that the tilting module $\TilJTL{0}{\q^2}$ appears always with the multiplicity one (it is the trivial $s\ell(2|1)$ module) which is quite important for our analysis -- it means that we have only one ground state. The $s\ell(2|1)$ content of the other multiplicities in front of $\TilJTL{j}{z^2}$ will be discussed below.

We now give examples for $N=6,8,14$ detailing the subquotient structure (with
simple subquotients) of
the tilting
modules $\TilJTL{j}{P}$.

\subsubsection{$N=6$ example}\label{sec:N6-ex}

We begin with a simple example on $N=6$ sites. We  give first the structure of indecomposable tiltings, using Fig.~\ref{fig:tilt-thirdroot} and Fig.~\ref{fig:cell-thirdroot}.
 \begin{align*}
   \xymatrix@R=16pt@C=0pt
   {
     &&\\
     &\TilJTL{2(1)}{1}:&\\
     &&
     }
   \xymatrix@R=14pt@C=14pt
   { &&\\
     {\bullet}\ar@{}|{\substack{\StJTL{3}{\q^2}}\kern-7pt}[]+<-20pt,25pt>\ar[dr]
     &&{\bullet}\ar@{}|{\substack{\StJTL{3}{\q^{-2}}}\kern-7pt}[]+<10pt,20pt>\ar[dl]\\
     &{\circ}\ar@{}|{\substack{\StJTL{2(1)}{1}}\kern-7pt}[]+<-50pt,5pt> &
     }\quad&
   \xymatrix@R=16pt@C=0pt
   {
     &&\\
     &=&\\
     &&
     }
        \xymatrix@R=16pt@C=14pt
   {
     {\bullet}\ar@{}|{\substack{\IrJTL{3}{\q^2}}\kern-7pt}[]+<-40pt,15pt>\ar[dr]
     &&{\bullet}\ar@{}|{\substack{\IrJTL{3}{\q^{-2}}}\kern-7pt}[]+<40pt,15pt>\ar[dl]\\
     &{\circ}\ar@{}|{\substack{\IrJTL{2(1)}{1}}\kern-7pt}[]+<-50pt,5pt>\ar[dr]\ar[dl] &\\
      {\circ}\ar@{}|{\substack{\IrJTL{3}{\q^2}}\kern-7pt}[]+<-40pt,15pt>
     &&{\circ}\ar@{}|{\substack{\IrJTL{3}{\q^{-2}}}\kern-7pt}[]+<40pt,15pt>
     }\\
        \xymatrix@R=16pt@C=0pt
   {
     &&\\
     &\TilJTL{0}{\q^2}:&\\
     &&
     }   \xymatrix@R=14pt@C=14pt
   { &&\\
     &{\bullet}\ar@{}|{\substack{\StJTL{2}{1}}\kern-7pt}[]+<-20pt,25pt>\ar[d]
     &\\
     &{\circ}\ar@{}|{\substack{\bAStTL{0}{\q^2}}\kern-7pt}[]+<-40pt,5pt> &
     }&
   \xymatrix@R=16pt@C=0pt
   {
     &&\\
     &=&\\
     &&
     }
        \xymatrix@R=20pt@C=18pt
   {&{\bullet}\ar@{}|{\substack{\IrJTL{2}{1}}\kern-7pt}[]+<-20pt,20pt>\ar[dl]\ar[d]\ar[dr]&\\
     {\bullet}\ar@{}|{\substack{\IrJTL{3}{\q^2}}\kern-7pt}[]+<-40pt,15pt>\ar[dr]
     &{\circ}\ar@{}|{\substack{\IrJTL{0}{\q^2}}\kern-7pt}[]+<-35pt,0pt>\ar[d]
     &{\bullet}\ar@{}|{\substack{\IrJTL{3}{\q^{-2}}}\kern-7pt}[]+<30pt,15pt>\ar[dl]\\
      &{\circ}\ar@{}|{\substack{\IrJTL{2}{1}}\kern-7pt}[]+<-35pt,0pt>&
           }
 \end{align*}
where different types of nodes denoted by $\bullet$ and
$\circ$ show the cell-filtration of the tilting modules: symbols are assigned to cell modules that appear on the left from equalities and  all simple subquotients (on the right) from a particular cell module are denoted by the  corresponding symbol.
Here, the notation $2(1)$ means that any of the two numbers $2$ or
$1$ can be  the first index in the subscript.
All other tiltings over $\rJTL{6}(1)$ are irreducible. We thus get the following dimensions
\begin{equation}\label{dim-tilt-N6}
\dim \TilJTL{0}{\q^2} = 11,\quad \dim \TilJTL{1}{1} = 17,\quad \dim \TilJTL{2}{1} = 8,\quad
\dim \TilJTL{2}{-1} = 6,\quad \dim \TilJTL{3}{\q^{\pm2}} = 1.
\end{equation}
 Using the diagrams we also obtain that the only nontrivial homomorphisms
between the tiltings are
\begin{gather}
\Hom(\TilJTL{2}{1},\TilJTL{0}{\q^2})=\oC,\quad
\Hom(\TilJTL{0}{\q^2},\TilJTL{2}{1})=\oC,\;\;
\Hom(\TilJTL{2}{1},\TilJTL{3}{\q^{\pm2}})=\oC,\quad
\Hom(\TilJTL{1}{1},\TilJTL{3}{\q^{\pm2}})=\oC,\label{hom-tilt-N6-1}\\
\Hom(\TilJTL{3}{\q^{\pm2}},\TilJTL{2}{1})=\oC,\quad
\Hom(\TilJTL{3}{\q^{\pm2}},\TilJTL{1}{1})=\oC.\qquad\label{hom-tilt-N6-2}
\end{gather}

Next, using the filtration of $\Hilb_N$  by cell modules given in~\eqref{Hsix}
and the subquotient structure of tiltings we deduce the spin-chain decomposition for $N=6$
over the tilting modules $\TilJTL{j}{P}$:
\begin{equation}\label{decomp-H6-JTL}
\Hilb_6 = \TilJTL{0}{\q^2} \oplus 8 \TilJTL{1}{1}
\oplus 22 \TilJTL{2}{1}  \oplus 24\TilJTL{2}{-1}
\oplus 112 \TilJTL{3}{1}
 \oplus 75 \TilJTL{3}{\q^{\pm2}}.
\end{equation}
Together with \eqref{dim-tilt-N6}-\eqref{hom-tilt-N6-2}, this gives the  dimensions of simple modules over
the centralizer of $\rJTL{6}(1)$ and its indecomposable tiltings.
Because the centralizer is bigger than  the $s\ell(2|1)$ symmetry,  the multiplicities in \eqref{decomp-H6-JTL} arise  in general
from direct sums of atypical/typical
representations of $s\ell(2|1)$.

 Let us describe the  $s\ell(2|1)$ content of these multiplicities here.
Using the formulas in~\cite{sl21rep,superDic}, we find the following decomposition of the Hilbert space over $s\ell(2|1)$
\begin{multline}\label{decomp-H6-sl21}
\Hilb_6|_{s\ell(2|1)} = \atyp{0,0} \oplus \atyp{0,3} \oplus
\atyp{\pm1,2} \oplus 2\,\atyp{\pm1/2,5/2} \oplus 9\,\atyp{0,2} \oplus
9\,\atyp{\pm1/2,3/2}\\ \oplus 18\,\atyp{0,1} \oplus 2\,\slPr_{\pm}(1/2)
\oplus 5\,\slPr(0),
 \end{multline}
where we use the notations of appendix~\ref{appSl21}. Recall that the dimension of typicals $\atyp{b,j}$, with $b\ne\pm j$, is $8j$, while the projective $s\ell(2|1)$-modules $\slPr_{\pm}(j)$ have dimension $16j + 4$ for $j>0$ and $8$ for $j=0$.


First of all, it is clear that the multiplicity $1$ in front of $\TilJTL{0}{\q^2}$ in~\eqref{decomp-H6-JTL} corresponds to the trivial atypical $s\ell(2|1)$-module $\atyp{0,0}$. This is because $\TilJTL{0}{\q^2}$ contains the groundstate of the Hamiltonian which transforms trivially with respect to $s\ell(2|1)$. This accounts for the direct summand $\atyp{0,0}$ in the sum~\eqref{decomp-H6-sl21} plus the $5$ copies of the top of $\slPr(0)$ and the $5$ copies of the (bottom) submodule $\atyp{0,0}\subset\slPr(0)$ -- we have thus counted  $11$ copies which is precisely the dimension of $\TilJTL{0}{\q^2}$, see~\eqref{dim-tilt-N6}. The multiplicity $8$ of $\TilJTL{1}{1}$ corresponds to the adjoint $s\ell(2|1)$ representation $\atyp{0,1}$: we shall see in the following that the scaling limit of this tilting module contains the Noether currents associated with the  $s\ell(2|1)$ symmetry. This accounts for $17$ out of the $18$ $\atyp{0,1}$ modules in the decomposition~\eqref{decomp-H6-sl21}.

Interpreting the multiplicity $22$ in front of $\TilJTL{2}{1}$ is a bit more difficult. We  use the decomposition
on $4$ sites,
\begin{equation}\label{decomp-H4-sl21}
\Hilb_4|_{s\ell(2|1)} = \atyp{0,0}  \oplus 4\,\atyp{0,1}  \oplus \,\atyp{0,2} \oplus \,\atyp{\pm1/2,3/2}
\oplus \,\slPr(0),
 \end{equation}
 where this multiplicity also appears in front of $\TilJTL{2}{1}$,
and observe that  it is enough to look at the bimodule for JTL and its centralizer -- this  gives many constraints for the $s\ell(2|1)$ content. The bimodule can be easily constructed using the subquotient structure of the tilting modules and information about possible non trivial Hom spaces between them, as those in~\eqref{hom-tilt-N6-1}. 
This way, we find that the multiplicity $22$  corresponds to $\atyp{0,2} \oplus \lbrace \frac{1}{2} \rbrace_{+}\oplus \lbrace \frac{1}{2} \rbrace_{-}$. Meanwhile, the multiplicity $24$ in front of $\TilJTL{2}{-1} $ corresponds to $\atyp{\pm1/2,3/2}$: this  we obtained
because the only remaining $s\ell(2|1)$ representations are $\atyp{\pm1/2,3/2}$ indeed.

Now, the remaining multiplicities $112$ and $2 \times 75$ on $6$ sites must be interpreted in terms of the $s\ell(2|1)$ representations
\begin{equation}
 \atyp{0,3} \oplus \atyp{0,2} \oplus \atyp{0,1} \oplus
\atyp{\pm1,2} \oplus 2\,\atyp{\pm1/2,5/2}  \oplus
3\,\atyp{\pm1/2,3/2} \oplus 4 \lbrace 0 \rbrace  \oplus 2 \lbrace 1 \rbrace_{\pm}\oplus \lbrace \frac{1}{2} \rbrace_{\pm}
 \end{equation}
where the last pieces come from breaking up the projectives $\slPr_{\pm}(1/2)$. One can check that the dimension indeed corresponds to $262=112+2\times 75$. Since we expect the results for the two multiplicities $75$ to be somewhat symmetric, it is reasonable to expect that the part $\atyp{0,3} \oplus \atyp{0,2} \oplus \atyp{0,1}$ that cannot be cut in half should contribute to the multiplicity $112$.   From the bimodule analysis, as we did for $N=4$,  and asuming the symmetry for $75$ we thus get
$$112 = \atyp{0,3} \oplus \atyp{0,2} \oplus \atyp{0,1} \oplus \atyp{\pm\half,\ffrac{3}{2}} \oplus \atyp{\pm1,2} \oplus P(0)$$
It means that the remaining $2\times 75$ multiplicities are given by `taking a half of' the $150$-dimensional representation
\begin{equation}
75 =
\atyp{\pm\half,\ffrac{3}{2}} \oplus \atyp{\pm\ffrac{1}{2},\ffrac{5}{2}}  \oplus \lbrace 0 \rbrace  \oplus  \lbrace 1 \rbrace_{\pm}.
 \end{equation}
Note that we obtained that both multiplicities $75$ correspond to isomorphic $s\ell(2|1)$ modules; they are non-isomorphic only as modules over the JTL's centralizer.

\subsection{General structure of tilting modules and Hamiltonians's Jordan cells}\label{sec:Jcells}
In this section, we give more examples of the subquotient structure of tilting modules, and provide finally the corresponding  general pattern. We also discuss some of our results on Jordan cells for the spin chain Hamiltonian $H$. Remarkably, we not only  observe  the Jordan cells of rank higher than $2$, but in fact show that  the maximum rank in the spin chain grows with the number of sites.

\subsubsection{$N=8$ example}\label{sec:N8}
We analyze first the more interesting case of $8$ sites,  where  Jordan cells of rank $3$
 appear for the first time. They involve now  states in the simple subquotients $\IrrJTL{4}{1}$. To justify this, we describe
the structure of indecomposable tiltings with simple subquotients, using Fig.~\ref{fig:tilt-thirdroot} and Fig.~\ref{fig:cell-thirdroot}.
 \begin{align*}
   \xymatrix@R=16pt@C=0pt
   {
     &&\\
     &&\\
     &\TilJTL{2(1)}{1}:&\\
     &&\\
     &&
     }
   \xymatrix@R=14pt@C=14pt
   {     &&\\
     &{\square}\ar@{}|{\substack{\StJTL{4}{1}}\kern-7pt}[]+<-30pt,25pt>\ar[dr]\ar[dl]&\\
     {\bullet}\ar@{}|{\substack{\StJTL{3}{\q^2}}\kern-7pt}[]+<-20pt,25pt>\ar[dr]
     &&{\bullet}\ar@{}|{\substack{\StJTL{3}{\q^{-2}}}\kern-7pt}[]+<10pt,20pt>\ar[dl]\\
     &{\circ}\ar@{}|{\substack{\StJTL{2(1)}{1}}\kern-7pt}[]+<-50pt,5pt> &\\
     &&
     }\quad&
   \xymatrix@R=16pt@C=0pt
   {
     &&\\
     &&\\
     &=&\\
     &&\\
     &&
     }
        \xymatrix@R=16pt@C=14pt
   {    & & {\square}\ar@{}|{\substack{\IrJTL{4}{1}}\kern-7pt}[]+<-40pt,15pt>\ar[dr]\ar[dl]
     &&\\
    & {\bullet}\ar@{}|{\substack{\IrJTL{3}{\q^2}}\kern-7pt}[]+<-40pt,15pt>\ar[dr]\ar[dl]
     &&{\bullet}\ar@{}|{\substack{\IrJTL{3}{\q^{-2}}}\kern-7pt}[]+<40pt,15pt>\ar[dl]\ar[dr]&\\
    {\bullet}\ar@{}|{\substack{\IrJTL{4}{1}}\kern-7pt}[]+<-30pt,10pt>\ar[dr]
    &&{\circ}\ar@{}|{\substack{\IrJTL{2(1)}{1}}\kern-7pt}[]+<-50pt,5pt>\ar[dr]\ar[dl]
    &&    {\bullet}\ar@{}|{\substack{\IrJTL{4}{1}}\kern-7pt}[]+<20pt,10pt>\ar[dl]\\
      &{\circ}\ar@{}|{\substack{\IrJTL{3}{\q^2}}\kern-7pt}[]+<-40pt,0pt>\ar[dr]
     &&{\circ}\ar@{}|{\substack{\IrJTL{3}{\q^{-2}}}\kern-7pt}[]+<30pt,0pt>\ar[dl]&\\
    & & {\circ}\ar@{}|{\substack{\IrJTL{4}{1}}\kern-7pt}[]+<-40pt,0pt>&&
     }\\
        \xymatrix@R=16pt@C=0pt
   {
     &&\\
     &\TilJTL{0}{\q^2}:&\\
     &&
     }   \xymatrix@R=14pt@C=14pt
   { &{\square}\ar@{}|{\substack{\StJTL{4}{1}}\kern-7pt}[]+<-20pt,25pt>\ar[d]&\\
     &{\bullet}\ar@{}|{\substack{\StJTL{2}{1}}\kern-7pt}[]+<-40pt,15pt>\ar[d]
     &\\
     &{\circ}\ar@{}|{\substack{\bAStTL{0}{\q^2}}\kern-7pt}[]+<-40pt,5pt> &
     }&
   \xymatrix@R=16pt@C=0pt
   {
     &&\\
     &=&\\
     &&
     }
        \xymatrix@R=20pt@C=18pt
   {&{\bullet}\ar@{}|{\substack{\IrJTL{2}{1}}\kern-7pt}[]+<-20pt,20pt>\ar[dl]\ar[d]\ar[dr]
     &{\square}\ar@{}|{\substack{\IrJTL{4}{1}}\kern-7pt}[]+<-20pt,20pt>\ar[d]\ar[dll]\\
     {\bullet}\ar@{}|{\substack{\IrJTL{3}{\q^2}}\kern-7pt}[]+<-40pt,15pt>\ar[dr]\ar[drr]
     &{\circ}\ar@{}|{\substack{\IrJTL{0}{\q^2}}\kern-7pt}[]+<-35pt,0pt>\ar[d]
     &{\bullet}\ar@{}|{\substack{\IrJTL{3}{\q^{-2}}}\kern-7pt}[]+<30pt,15pt>\ar[dl]\ar[d]\\
      &{\circ}\ar@{}|{\substack{\IrJTL{2}{1}}\kern-7pt}[]+<-35pt,0pt>
     &{\bullet}\ar@{}|{\substack{\IrJTL{4}{1}}\kern-7pt}[]+<-20pt,20pt>
           }
 \end{align*}
where different types of nodes denoted by $\square$, $\bullet$ and
$\circ$ show the cell-filtration of the tilting modules: symbols are assigned to cell modules that appear on the left from equalities while all simple subquotients (on the right) from a particular cell module are denoted by the  corresponding symbol.  The notation $2(1)$ means that any of the two numbers $2$ or $1$ can occur as  the first index in the subscript.

 Given this structure for the tilting modules $\TilJTL{2(1)}{1}$, we expect the Hamiltonian $H$ to show rank-3 Jordan cells involving (generalized) eigenstates from simple subquotients $\IrJTL{4}{1}$.  Note that this is not a rigorous result, since the indecomposability does not necessarily mean that all elements of the algebra have Jordan cells -- only that some might do. Nevertheless, there is overwhelming evidence that, from an algebraic point of view, the Hamiltonian behaves in a very `generic' fashion, as we now verify numerically.  On $N=8$ sites for instance, the module $\IrrJTL{4}{1}$ is one-dimensional and corresponds to the eigenvalue $0$. Since the tilting modules $\TilJTL{1}{1}$ and $\TilJTL{2}{1}$ appear with multiplicities $8$ and $22$ respectively, we would expect the Hamiltonian $H$ on $8$ sites to show $22+8=30$ rank-3 Jordan cells for the eigenvalue $0$.  Using a formal computation software (Mathematica$^\copyright$), we have computed exactly the null-spaces ${\rm Ker} \ H$, ${\rm Ker} \ H^2$, ${\rm Ker} \ H^3$ on various $s\ell(2|1)$ sectors labeled by the quantum numbers $(B,S_z)$. This gives us all the information with need on the Jordan cell structure for the eigenvalue $0$. For instance, in the sector $(B,S_z)=(0,0)$, we find ${\rm dim} \ {\rm Ker} \ H = 155$, ${\rm dim} \ {\rm Ker} \ H^2 = 170$, and ${\rm dim} \ {\rm Ker} \ H^3 = 174$, $174$ being the multiplicity of the eigenvalue $0$ in that sector. Hence, once can clearly see that there are $4$ rank-3 Jordan cells in that sector, which is exactly what is expected from our analysis as $(B,S_z)=(0,0)$ occurs once in the multiplicity $8=\atyp{0,1}$ of  $\TilJTL{1}{1}$, and three times in  the multiplicity $22=\atyp{0,2} \oplus \lbrace \frac{1}{2} \rbrace_{+}\oplus \lbrace \frac{1}{2} \rbrace_{-}$ of  $\TilJTL{2}{1}$. We have actually checked the presence of the $30$ rank-3 Jordan cells in the whole spectrum. Note that we have also carefully analyzed the multiplicities of the rank-2 Jordan cells on small sizes, and found a perfect agreement with our algebraic results.

\subsubsection{$N=14$ example}\label{sec:N14}
We next analyze the case of $14$ sites  where Jordan cells of rank $4$
should appear for states from simple subquotients $\IrrJTL{7}{1}$. The first time rank $4$ Jordan cells for the Hamiltonian might appear is actually for $N=12$ but diagrams for subquotient structure in the case of $N=14$ look somewhat nicer and we have decided to discuss this case instead.
As above, using the structure of tilting modules in Fig.~\ref{fig:tilt-thirdroot} given in terms of cell modules from Fig.~\ref{fig:cell-thirdroot}, we can describe the structure of indecomposable tiltings in terms of simple subquotients. For this, one should also use self-duality arguments. An example is given   in Fig.~\ref{fig:N14-rank4tilt}  for $\TilJTL{3}{\q^{\pm2}}$.

We also notice in the diagram a pattern of appearance of isomorphic simple subquotients. From the previous case we learned that these subquotients at different sections (levels) of the diagram are connected by the action of our Hamiltonian. For example,  simple subquotients $\IrrJTL{7}{1}$ indicated in   bold  appear $6$ times but at $4$ levels. So, we expect that these modules allow Jordan cells for the Hamiltonian $H$ of (at least) rank~$4$.
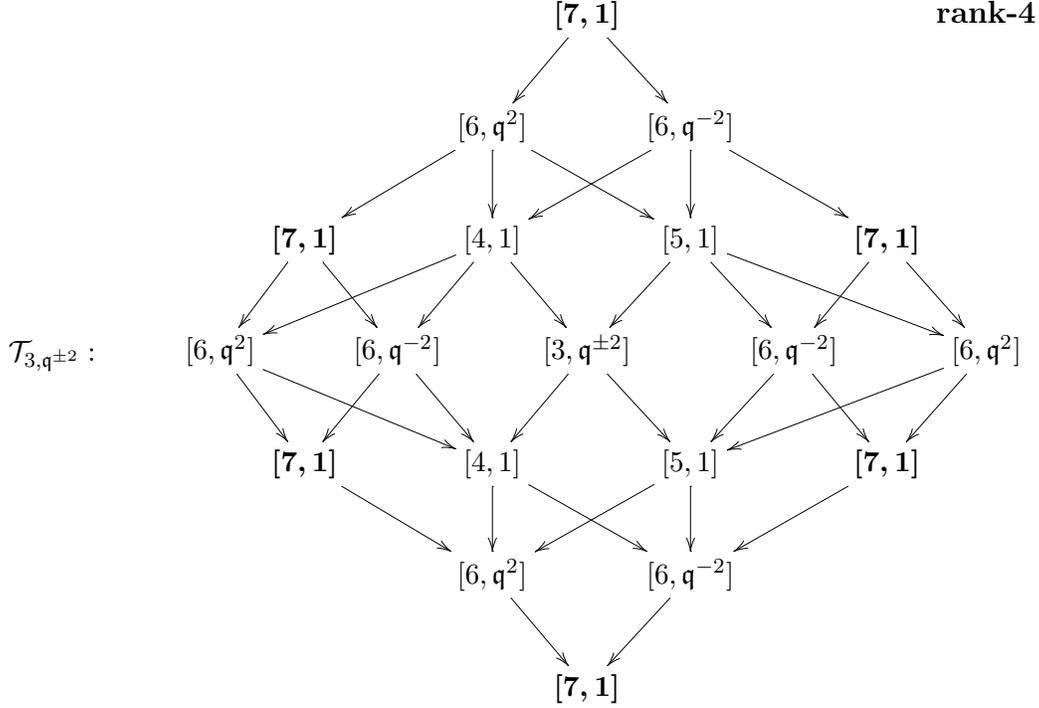
\begin{figure}
 \begin{equation*}
   \xymatrix@R=55pt@C=0pt
   {
     &&\\
     &&\\
     &\TilJTL{3}{\q^{\pm2}}:\qquad&\\
     &&\\
     &&
     }
\xymatrix@R=25pt@C=0pt
   {
    &&& & {\boldsymbol{\IrJTL{7}{1}}}\ar[dr]\ar[dl]     &&&&{\large{\textbf{rank-$\boldsymbol{4}$}}}\\
 &&&{\IrJTL{6}{\q^2}}\ar[drr]\ar[dll]\ar[d]
     &&{\IrJTL{6}{\q^{-2}}}\ar[dll]\ar[drr]\ar[d]&&\\
    & {\boldsymbol{\IrJTL{7}{1}}}\ar[dr]\ar[dl]
    &&{\IrJTL{4}{1}}\ar[dr]\ar[dl]\ar[dlll]
     &&{\IrJTL{5}{1}}\ar[dl]\ar[dr]\ar[drrr]
     &&{\boldsymbol{\IrJTL{7}{1}}}\ar[dl]\ar[dr]\\
    {\IrJTL{6}{\q^{2}}}\ar[dr]\ar[drrr]
    &&{\IrJTL{6}{\q^{-2}}}\ar[dr]\ar[dl]
    &&{\IrJTL{3}{\q^{\pm2}}}\ar[dr]\ar[dl]
    &&    {\IrJTL{6}{\q^{-2}}}\ar[dl]\ar[dr]
    &&{\IrJTL{6}{\q^{2}}}\ar[dl]\ar[dlll]\\
    & {\boldsymbol{\IrJTL{7}{1}}}\ar[drr]
     &&{\IrJTL{4}{1}}\ar[drr]\ar[d]
     &&{\IrJTL{5}{1}}\ar[dll]\ar[d]
          &&{\boldsymbol{\IrJTL{7}{1}}}\ar[dll]\\
 &&&{\IrJTL{6}{\q^2}}\ar[dr]
     &&{\IrJTL{6}{\q^{-2}}}\ar[dl]&&\\
    &&& & {\boldsymbol{\IrJTL{7}{1}}}&&&&
     }
 \end{equation*}
      \caption{The structure of the tilting $\rJTL{N}$-module
      $\TilJTL{3}{\q^{\pm2}}$ for $N=14$. These modules   allow Jordan cells of rank~$4$ for the hamiltonian $H$ acting on  states from simple subquotients $\IrJTL{7}{1}$ indicated in bold. }
    \label{fig:N14-rank4tilt}
    \end{figure}

  Next, the tilting $\rJTL{N}$-modules
      $\TilJTL{2}{1}$ and $\TilJTL{1}{1}$ for $N=14$  allow Jordan cells of rank~$5$ for the Hamiltonian $H$ acting on  states from simple subquotients $\IrrJTL{7}{1}$. The rank is actually stabilized, e.g., for all $N>14$ the maximum rank on the whole spin chain (and at least for first low lying states) for subquotients  $\IrrJTL{4}{1}$ is $3$,  for $\IrrJTL{7}{1}$ is $5$, etc.

Finally, we give the structure of the ``vacuum'' (we call it so since  it contains the vacuum state of the spin chain) tilting module $\TilJTL{0}{\q^2}$ at $N=14$
\begin{equation}\label{Tilt-N14}
        \xymatrix@R=42pt@C=26pt@W=2pt@M=2pt
   {
   &&&& {\circ}\ar@{}|{\substack{\IrJTL{7}{1}}\kern-7pt}[]+<0pt,20pt>\ar[drrr]\ar[dlll] &&\\
   &{\bullet}\ar@{}|{\substack{\IrJTL{6}{\q^{-2}}}\kern-7pt}[]+<-30pt,15pt>\ar[dr]\ar[dl]\ar[drrrrr]&&&&
   &&{\bullet}\ar@{}|{\substack{\IrJTL{6}{\q^2}}\kern-7pt}[]+<15pt,15pt>\ar[dl]\ar[dr]\ar[dlllll]\\
   {\bullet}\ar@{}|{\substack{\IrJTL{7}{1}}\kern-7pt}[]+<-20pt,20pt>\ar[dr]\ar[drr]
   &&{\square}\ar@{}|{\substack{\IrJTL{4}{1}}\kern-7pt}[]+<15pt,20pt>\ar[dr]\ar[drrr]\ar[d]\ar[dl]
   &&{\bullet}\ar@{}|{\substack{\IrJTL{2}{1}}\kern-7pt}[]+<0pt,20pt>\ar[dl]\ar[d]\ar[dr]
     &&{\square}\ar@{}|{\substack{\IrJTL{5}{1}}\kern-7pt}[]+<-20pt,20pt>\ar[dl]\ar[dlll]\ar[d]\ar[dr]
     &&{\bullet}\ar@{}|{\substack{\IrJTL{7}{1}}\kern-7pt}[]+<10pt,20pt>\ar[dl]\ar[dll]\\
   &{\square}\ar@{}|{\substack{\IrJTL{6}{\q^{-2}}}\kern-7pt}[]+<-45pt,5pt>\ar[dr]\ar[dl]
   &{\square}\ar@{}|{\substack{\IrJTL{6}{\q^2}}\kern-7pt}[]+<-40pt,3pt>\ar[d]\ar[dll]
   &{\bullet}\ar@{}|{\substack{\IrJTL{3}{\q^2}}\kern-7pt}[]+<-35pt,5pt>\ar[dr]\ar[drrr]\ar[dl]
     &{\circ}\ar@{}|{\substack{\IrJTL{0}{\q^2}}\kern-7pt}[]+<-30pt,5pt>\ar[d]
     &{\bullet}\ar@{}|{\substack{\IrJTL{3}{\q^{-2}}}\kern-7pt}[]+<25pt,5pt>\ar[dl]\ar[dr]\ar[dlll]
     &{\square}\ar@{}|{\substack{\IrJTL{6}{\q^{-2}}}\kern-7pt}[]+<30pt,3pt>\ar[d]\ar[drr]
     &{\square}\ar@{}|{\substack{\IrJTL{6}{\q^{2}}}\kern-7pt}[]+<25pt,5pt>\ar[dr]\ar[dl]\\
   {\square}\ar@{}|{\substack{\IrJTL{7}{1}}\kern-7pt}[]+<-20pt,-20pt>\ar[dr]
   && {\bullet}\ar@{}|{\substack{\IrJTL{4}{1}}\kern-7pt}[]+<-35pt,5pt>\ar[dl]\ar[drrrrr]
    &  &{\circ}\ar@{}|{\substack{\IrJTL{2}{1}}\kern-7pt}[]+<0pt,-20pt>
     &&{\bullet}\ar@{}|{\substack{\IrJTL{5}{1}}\kern-7pt}[]+<20pt,0pt>\ar[dr]\ar[dlllll]
     && {\square}\ar@{}|{\substack{\IrJTL{7}{1}}\kern-7pt}[]+<20pt,-15pt>\ar[dl]\\
     &{\bullet}\ar@{}|{\substack{\IrJTL{6}{\q^{-2}}}\kern-7pt}[]+<-35pt,-15pt>\ar[drrr]&&&
     &&&{\bullet}\ar@{}|{\substack{\IrJTL{6}{\q^2}}\kern-7pt}[]+<20pt,-15pt>\ar[dlll]&
     \\&&&& {\bullet}\ar@{}|{\substack{\IrJTL{7}{1}}\kern-7pt}[]+<0pt,-20pt> &&
     }
\end{equation}
where $\circ$'s denote simple subquotients from the standard modules $\StJTL{0}{\q^2}$ and $\StJTL{7}{1}$,  $\bullet$'s denote subquotients from  $\StJTL{2}{1}$ and $\StJTL{6}{\q^{\pm2}}$, and $\square$'s are for the standard modules $\StJTL{4}{1}$ and $\StJTL{5}{1}$. Using these notations, one can easily see the filtration by standard modules proposed before.  One can also note that the structure of $\TilJTL{0}{\q^2}$ is essentially (but not completely) fixed by the filtration proposed and the self-duality requirement. The structure~\eqref{Tilt-N14} for the tilting module $\TilJTL{0}{\q^2}$ is obviously invariant under the duality operation. We drew only the minimum number of arrows between the simple subquotients --  those fixed by the duality -- though there  might be additional arrows, for example, between $\IrrJTL{7}{1}$ and $\IrrJTL{6}{\q^{\pm2}}$.

 One can also check that the projective covers found in previous sections, see Fig.~\ref{fig:proj-thirdroot}, indeed cover any submodule in $\TilJTL{0}{\q^2}$. For example, the top node/subquotient $\IrrJTL{7}{1}$ is covered by the projective module $\PrJTL{7}{1}$ in the following way: the kernel of the map $\PrJTL{7}{1}\to\TilJTL{0}{\q^2}$ contains the submodule $\StJTL{1}{1}$, the maximum proper submodule of $\StJTL{2}{1}$ (considered itself as a submodule in $\PrJTL{7}{1}$, see Fig.~\ref{fig:proj-thirdroot}) and a linear combination of the maximal proper submodules in the two subquotients $\StJTL{3}{\q^2}$ and $\StJTL{3}{\q^{-2}}$ as they are presented in Fig.~\ref{fig:proj-thirdroot}.

 We finally note in this example, that the vacuum tilting module $\TilJTL{0}{\q^2}$ has two irreducible tops -- the subquotients $\IrrJTL{7}{1}$ marked by `$\circ$' and $\IrrJTL{2}{1}$ marked by `$\bullet$'. This happens because we can not have arrows connecting  the nodes $\IrJTL{6}{\q^{\pm2}}$, lying higher in the diagram, with the node $\IrJTL{2}{1}$ -- otherwise the projective modules $\PrJTL{6}{\q^{\pm2}}$ could not cover submodules growing from these nodes $\IrJTL{6}{\q^{\pm2}}$ which is a contradiction. Similarly, one can exclude many other arrows and end up with the diagram we present in~\eqref{Tilt-N14}. This property that the states from $\IrrJTL{2}{1}$ and the vacuum state from $\IrrJTL{0}{\q^2}$ are somehow disconnected from the rest of the module is peculiar to the vacuum tilting module. All other indecomposable tiltings for $N=14$ have a unique top subquotient.

\subsubsection{General structure for $\TilJTL{0}{\q^2}$ and Jordan cells for $H$}\label{sec:vac-tilt-gen}

We now discuss the general structure of the vacuum tilting module in terms of irreducible subquotients. It turns out that this structure is best formulated through the example of   $N=18$. Using the same ideas as before, we obtain the diagram for the structure of $\TilJTL{0}{\q^2}$ presented in Fig.~\ref{fig:tilt-N18},
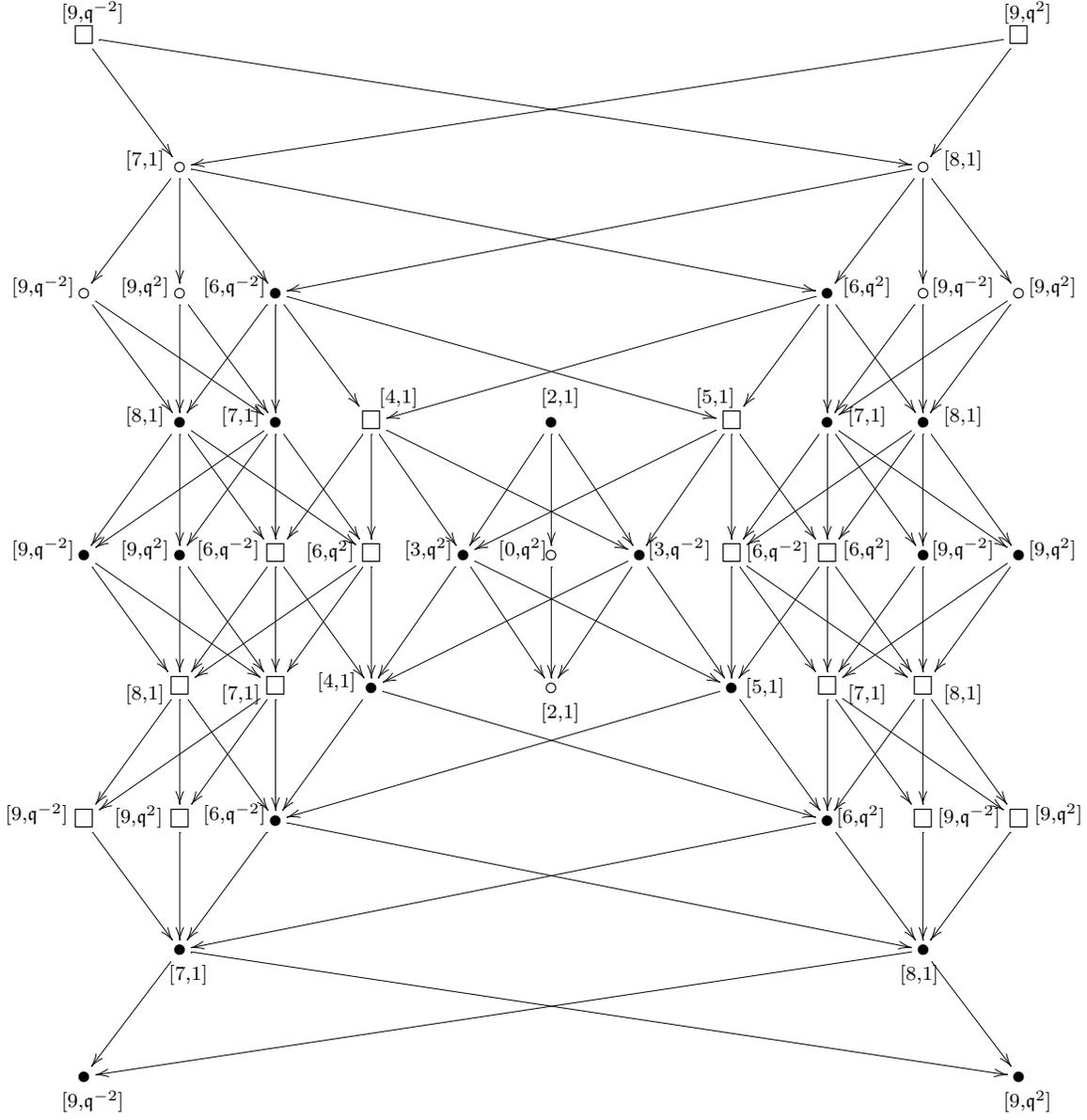
\begin{figure}
\begin{equation*}
        \xymatrix@R=42pt@C=26pt@W=2pt@M=2pt
   {
    {\square}\ar@{}|{\substack{\IrJTL{9}{\q^{-2}}}\kern-7pt}[]+<0pt,20pt>\ar[dr]\ar[drrrrrrrrr] &&&& &&
    &&&&{\square}\ar@{}|{\substack{\IrJTL{9}{\q^2}}\kern-7pt}[]+<0pt,20pt>\ar[dl]\ar[dlllllllll]\\
   &{\circ}\ar@{}|{\substack{\IrJTL{7}{1}}\kern-7pt}[]+<-35pt,5pt>\ar[drrrrrrr]\ar[dr]\ar[d]\ar[dl]&&&&  &&
   &&{\circ}\ar@{}|{\substack{\IrJTL{8}{1}}\kern-7pt}[]+<25pt,5pt>\ar[dlllllll]\ar[dl]\ar[d]\ar[dr]\\
   {\circ}\ar@{}|{\substack{\IrJTL{9}{\q^{-2}}}\kern-7pt}[]+<-40pt,5pt>\ar[dr]\ar[drr]
   &{\circ}\ar@{}|{\substack{\IrJTL{9}{\q^2}}\kern-7pt}[]+<-35pt,5pt>\ar[d]\ar[dr]
   &{\bullet}\ar@{}|{\substack{\IrJTL{6}{\q^{-2}}}\kern-7pt}[]+<-40pt,5pt>\ar[dr]\ar[dl]\ar[drrrrr]\ar[d]&&&&
   &&{\bullet}\ar@{}|{\substack{\IrJTL{6}{\q^2}}\kern-7pt}[]+<25pt,5pt>\ar[dl]\ar[dr]\ar[dlllll]\ar[d]
   &{\circ}\ar@{}|{\substack{\IrJTL{9}{\q^{-2}}}\kern-7pt}[]+<25pt,5pt>\ar[d]\ar[dl]
     &{\circ}\ar@{}|{\substack{\IrJTL{9}{\q^2}}\kern-7pt}[]+<20pt,5pt>\ar[dl]\ar[dll]\\
   &{\bullet}\ar@{}|{\substack{\IrJTL{8}{1}}\kern-7pt}[]+<-35pt,5pt>\ar[dr]\ar[drr]\ar[d]\ar[dl]
   &{\bullet}\ar@{}|{\substack{\IrJTL{7}{1}}\kern-7pt}[]+<-35pt,5pt>\ar[d]\ar[dr]\ar[dll]\ar[dl]
   &{\square}\ar@{}|{\substack{\IrJTL{4}{1}}\kern-7pt}[]+<15pt,20pt>\ar[dr]\ar[drrr]\ar[d]\ar[dl]
   &&{\bullet}\ar@{}|{\substack{\IrJTL{2}{1}}\kern-7pt}[]+<0pt,20pt>\ar[dl]\ar[d]\ar[dr]
     &&{\square}\ar@{}|{\substack{\IrJTL{5}{1}}\kern-7pt}[]+<-20pt,20pt>\ar[dl]\ar[dlll]\ar[d]\ar[dr]
     &{\bullet}\ar@{}|{\substack{\IrJTL{7}{1}}\kern-7pt}[]+<25pt,5pt>\ar[d]\ar[dl]\ar[dr]\ar[drr]
     &{\bullet}\ar@{}|{\substack{\IrJTL{8}{1}}\kern-7pt}[]+<25pt,5pt>\ar[dl]\ar[dll]\ar[d]\ar[dr]\\
     {\bullet}\ar@{}|{\substack{\IrJTL{9}{\q^{-2}}}\kern-7pt}[]+<-40pt,5pt>\ar[dr]\ar[drr]
   &{\bullet}\ar@{}|{\substack{\IrJTL{9}{\q^2}}\kern-7pt}[]+<-35pt,5pt>\ar[d]\ar[dr]
   &{\square}\ar@{}|{\substack{\IrJTL{6}{\q^{-2}}}\kern-7pt}[]+<-45pt,5pt>\ar[dr]\ar[dl]\ar[d]
   &{\square}\ar@{}|{\substack{\IrJTL{6}{\q^2}}\kern-7pt}[]+<-40pt,3pt>\ar[d]\ar[dll]\ar[dl]
   &{\bullet}\ar@{}|{\substack{\IrJTL{3}{\q^2}}\kern-7pt}[]+<-35pt,5pt>\ar[dr]\ar[drrr]\ar[dl]
     &{\circ}\ar@{}|{\substack{\IrJTL{0}{\q^2}}\kern-7pt}[]+<-30pt,5pt>\ar[d]
     &{\bullet}\ar@{}|{\substack{\IrJTL{3}{\q^{-2}}}\kern-7pt}[]+<25pt,5pt>\ar[dl]\ar[dr]\ar[dlll]
     &{\square}\ar@{}|{\substack{\IrJTL{6}{\q^{-2}}}\kern-7pt}[]+<30pt,3pt>\ar[d]\ar[drr]\ar[dr]
     &{\square}\ar@{}|{\substack{\IrJTL{6}{\q^{2}}}\kern-7pt}[]+<25pt,5pt>\ar[dr]\ar[dl]\ar[d]
     &{\bullet}\ar@{}|{\substack{\IrJTL{9}{\q^{-2}}}\kern-7pt}[]+<25pt,5pt>\ar[d]\ar[dl]
     &{\bullet}\ar@{}|{\substack{\IrJTL{9}{\q^2}}\kern-7pt}[]+<20pt,5pt>\ar[dl]\ar[dll]\\
   &{\square}\ar@{}|{\substack{\IrJTL{8}{1}}\kern-7pt}[]+<-35pt,-5pt>\ar[dr]\ar[d]\ar[dl]
   &{\square}\ar@{}|{\substack{\IrJTL{7}{1}}\kern-7pt}[]+<-35pt,-5pt>\ar[d]\ar[dl]\ar[dll]
   & {\bullet}\ar@{}|{\substack{\IrJTL{4}{1}}\kern-7pt}[]+<-35pt,5pt>\ar[dl]\ar[drrrrr]
    &  &{\circ}\ar@{}|{\substack{\IrJTL{2}{1}}\kern-7pt}[]+<0pt,-20pt>
     &&{\bullet}\ar@{}|{\substack{\IrJTL{5}{1}}\kern-7pt}[]+<20pt,0pt>\ar[dr]\ar[dlllll]
     &{\square}\ar@{}|{\substack{\IrJTL{7}{1}}\kern-7pt}[]+<25pt,-5pt>\ar[d]\ar[dr]\ar[drr]
     & {\square}\ar@{}|{\substack{\IrJTL{8}{1}}\kern-7pt}[]+<25pt,-5pt>\ar[dl]\ar[d]\ar[dr]\\
    {\square}\ar@{}|{\substack{\IrJTL{9}{\q^{-2}}}\kern-7pt}[]+<-45pt,5pt>\ar[dr]
   &{\square}\ar@{}|{\substack{\IrJTL{9}{\q^2}}\kern-7pt}[]+<-40pt,3pt>\ar[d]
    &{\bullet}\ar@{}|{\substack{\IrJTL{6}{\q^{-2}}}\kern-7pt}[]+<-40pt,5pt>\ar[dl]\ar[drrrrrrr]&&&
     &&&{\bullet}\ar@{}|{\substack{\IrJTL{6}{\q^2}}\kern-7pt}[]+<20pt,3pt>\ar[dr]\ar[dlllllll]
     &{\square}\ar@{}|{\substack{\IrJTL{9}{\q^{-2}}}\kern-7pt}[]+<30pt,3pt>\ar[d]
     &{\square}\ar@{}|{\substack{\IrJTL{9}{\q^{2}}}\kern-7pt}[]+<25pt,5pt>\ar[dl]\\
     &{\bullet}\ar@{}|{\substack{\IrJTL{7}{1}}\kern-7pt}[]+<0pt,-20pt>\ar[dl]\ar[drrrrrrrrr]&&&& &&
     &&{\bullet}\ar@{}|{\substack{\IrJTL{8}{1}}\kern-7pt}[]+<-10pt,-20pt>\ar[dr]\ar[dlllllllll]&&\\
    {\bullet}\ar@{}|{\substack{\IrJTL{9}{\q^{-2}}}\kern-7pt}[]+<0pt,-20pt> &&&& &&&&
    &&{\bullet}\ar@{}|{\substack{\IrJTL{9}{\q^2}}\kern-7pt}[]+<0pt,-20pt>
     }
\end{equation*}
\caption{Subquotient structure ``Eiffel tower'' of the vacuum tilting $\rJTL{18}$-module $\TilJTL{0}{\q^2}$. Here, $\circ$'s denote simple subquotients from the standard modules $\StJTL{0}{\q^2}$, $\StJTL{7}{1}$ and $\StJTL{8}{1}$, the $\bullet$'s denote subquotients from  $\StJTL{2}{1}$ and $\StJTL{6}{\q^{\pm2}}$, and $\square$'s are for the standard modules $\StJTL{4}{1}$, $\StJTL{5}{1}$ and $\StJTL{9}{\q^{\pm2}}$.}\label{fig:tilt-N18}
\end{figure}
where $\circ$'s denote simple subquotients from the standard modules $\StJTL{0}{\q^2}$, $\StJTL{7}{1}$ and $\StJTL{8}{1}$, the $\bullet$'s denote subquotients from  $\StJTL{2}{1}$ and $\StJTL{6}{\q^{\pm2}}$, and $\square$'s are for the standard modules $\StJTL{4}{1}$, $\StJTL{5}{1}$ and $\StJTL{9}{\q^{\pm2}}$. Using these notations, one can easily see the filtration by standard modules proposed before.  We see that in general we have two cones: one consisting of $\bullet$'s, and the other reflected in the horizontal line containing $\IrrJTL{0}{\q^2}$. The vacuum irreducible subquotient $\IrrJTL{0}{\q^2}$ lives in the intersection of the cones while the boundaries of these cones give  bounds for the appearance of simple subquotients $\IrrJTL{j}{P}$, with $j>2$. We denote  this particular structure of the vacuum tilting module  the ``Eiffel tower'' structure.

Now for larger number of sites $N$, the vacuum tilting module $\TilJTL{0}{\q^2}$ has essentially the same pattern of nodes and arrows between them\footnote{A slight difference is only for cases $\frac{N}{2}\mod 3 = 1$ where the diagram follows the pattern from \eqref{Tilt-N14} with one node at the top and one at the bottom.}, the increasing values of $N$ simply giving rise to   longer and longer `ladders' on the left and right sides of the corresponding ``Eiffel tower''.

Finally, we give a short comment on the other tilting modules $\TilJTL{j}{P}$ with $j>0$ that belong to the doubly critical class, i.e., those corresponding to the weights $[2,1]$, $[3,\q^{\pm2}]$, etc. Though it is not easy to write down a general diagram for their subquotient structure in terms of irreducible subquotients (one can imagine a ``ladder-of-ladders" structure, of course), the most important   is actually their general diagram in Fig.~\ref{fig:tilt-thirdroot}  in terms of cell modules.

\medskip

\textit{On Jordan cells.}
 Having now  a better understanding on how the diagrams for tilting modules are organized in terms of simple subquotients (at least for the vacuum tilting module) we are ready to discuss the possible structure of Jordan cells for the spin chain Hamiltonian $H$.
Indeed, following the  discussion at the end of Sec.~\ref{sec:N8} and Sec.~\ref{sec:N14}, and in agreement with numerical experimentation,  the rank of the Hamiltonian  Jordan cells involving states of a particular (generalized) eigenvalue  can be estimated by counting, in the diagram for a given module, the horizontal levels  that contain  the simple subquotients  to which  these states belong. So, for  general values of $N$, we believe  that the Hamiltonian Jordan cells  involving states from $\IrrJTL{j}{P}$ subquotients in $\TilJTL{0}{\q^2}$  have rank at least as large as the number of  appearance of $\IrrJTL{j}{P}$ on the left (or right) part of the corresponding  diagram of $\TilJTL{0}{\q^2}$. Once again, to obtain the diagram for larger values of $N>18$ one should extend   the one in Fig.~\ref{fig:tilt-N18} in a rather obvious way, and, for values $N<18$, just  remove nodes for subquotients with the $j$ index greater than $N/2$.
We also note that we do not assume that the Hamiltonian is diagonalizable in each simple JTL. There are recent results~\cite{MDSA} stating that the Hamiltonian $H=\sum_{i=1}^N e_i$
has Jordan cells of rank $2$ in cell (standard) modules for the periodic Temperley-Lieb algebra, which is rather surprising comparing to the open case. It could  thus be that counting  the horizontal levels in the diagram  only gives lower bounds on  the rank Jordan cells. We give our conjecture on these bounds below.

 Using the general diagram in Fig.~\ref{fig:tilt-thirdroot} and taking into account the fact that  an irreducible subquotient $\IrJTL{j'}{P'}$ appears in a cell module $\StJTL{j}{P}$ only if $j'>j$ and with multiplicity one,  see Fig.~\ref{fig:cell-thirdroot}, we can estimate lower bounds for the rank of the Hamiltonian Jordan cells in each tilting module. We should just count the number of levels (also called Loewy layers) where a particular irreducible subquotient $\IrJTL{j'}{P'}$ appears. So, at large enough $N$ the Jordan cell rank for states from $\IrJTL{j'}{P'}$ in $\TilJTL{j}{P}$ with $j'\geq j$ is bounded by the number of (horizontal) levels counted from the level of the diagram in Fig.~\ref{fig:tilt-thirdroot} containing the node $\StJTL{j}{P}$ up to the level containing the node $\StJTL{j'}{P'}$. One can easily compute this number. For example, \textit{the highest rank  Jordan cells for $j'>2$ should be in  $\TilJTL{2}{1}$ and $\TilJTL{1}{1}$ and states from $\IrJTL{j'}{P'}$ subquotients are expected to be involed into Jordan cells of $H$ of rank at least given   by the following number:}
\begin{equation}\label{eq:min rank-1}
\text{rank of}\, H\; \text{in}\; \TilJTL{2(1)}{1} \geq
\begin{cases}
2 \bigl\lceil \frac{j'-2}{3} \bigr\rceil,& \qquad j'\bmod 3 =0,\quad P'=\q^{\pm2},\\
2 \bigl\lceil \frac{j'-2}{3}\bigr\rceil +1,& \qquad j'\bmod 3 = 1 \; \text{or}\; 2, \quad P'=1.\\
\end{cases}
\end{equation}
Further, using the structure of $\TilJTL{0}{\q^2}$ presented in Fig.~\ref{fig:tilt-JTL-vac}, we also expect that the lower bound for states from $\IrrJTL{2}{1}$ subquotients is $2$, and for states from  $\IrrJTL{3}{\q^{\pm2}}$ is $1$.
For the other subquotients with $j'>3$ we have that the value for the lower bound of the Jordan cells rank differs   by one from that  in~\eqref{eq:min rank-1}. It is because the diagram for $\TilJTL{0}{\q^2}$ in terms of cell modules has no nodes $\StJTL{3}{\q^{\pm2}}$ and thus, the number of corresponding horizontal levels or Loewy layers is lowered  by one unit.

 Tilting modules in the singly critical case where the corresponding cell modules have a chain structure can be studied in a similar fashion. We do not describe explicitly their structure here and we refer the interested reader to~\cite{VGJS} where very similar modules were encountered in the context of the blob algebra.

\medskip

\section{Taking the continuum limit of the $s\ell(2|1)$ spin chain}\label{sec:takecontlim}

We now turn to  the scaling (continuum) limit of the $s\ell(2|1)$ spin chain: our goal is to infer from the foregoing algebraic analysis results about the representations of the product of left and right Virasoro algebras that act on the low energy states, and ultimately, all the properties of the corresponding LCFT. This is a difficult task, which we will only begin in this paper.

A major difference with the case of   $\gl(1|1)$  studied in~\cite{GRS1,GRS3}  is that the $s\ell(2|1)$ spin chain is not free. It cannot be diagonalized using  free fermions (or a combination of free fermions and free bosons), and therefore,
we have much less control on the (generalized) eigenstates and eigenvalues, and consequently, on the scaling limit.

Many properties of the $s\ell(2|1)$ spin chain can nevertheless be obtained exactly, by combining the algebraic analysis with the Bethe ansatz. It is important at this stage to stress that our chain is not the integrable alternating spin chain one would obtain \cite{EsslerFS} from the general inverse scattering construction. To obtain, for instance, the (generalized) eigenvalues of the Hamiltonian,  what one must do is observe that, since the chain provides a representation of the Jones--Temperley--Lieb algebra, the eigenvalues in each standard  module can be obtained using the results from another, Bethe ansatz solvable chain, where the same modules appear -- in our case, the twisted XXZ spin chain (we also refer to~\cite{[ABNKS]} where the spectrum problem for the periodic $s\ell(2|1)$ spin chain was also studied.)

Focussing now on the continuum limit, we are interested in the generating function of gaps  for each standard JTL module. Since we know the  decomposition of the spin chain over tilting modules, and each summand  consists of a glueing of many standard JTL modules, the generating function for gaps in our spin chain -- which will then give information on the Virasoro content -- will be obtained using results of the previous section.
To start, we thus describe the scaling limit of the twisted XXZ models, where each spin sector is isomorphic to a standard module over the affine TL at generic values of parameters and the asymptotics (at large $N$) of the generating functions is known.

\subsection{Twisted XXZ spin chain and continuum limit}\label{sec:twXXZ}

 It is well known
that the 6-vertex model (or the XXZ spin-chain) provides a natural representation of the affine Temperley-Lieb algebra, where the generators read
\begin{equation}
e_i = \mathbb{I} \otimes \mathbb{I} \otimes \dots \otimes \left( \begin{array}{cccc}
0 & 0 & 0 & 0 \\
0 & \q^{-1} & -1 & 0 \\
0 & -1 & \q & 0 \\
0 & 0 & 0 & 0 \end{array} \right) \otimes \dots \otimes \mathbb{I},
\label{e_eiXXZ}
\end{equation}
acting on $(i,i+1)$th tensor components or spins $\{\uparrow\uparrow,\uparrow\downarrow,\downarrow\uparrow,\downarrow\downarrow\}$ of the ``Hilbert'' space $ \mathcal{H}_{\rm XXZ} = (\mathbb{C}^2)^{\otimes 2N} $ and $1\leq i\leq 2N-1$.
In the basis of the last and the first spins, the last generator $e_{2N}$ is
\begin{equation}
\left( \begin{array}{cccc}
0 & 0 & 0 & 0 \\
0 & \q^{-1} & -\mathrm{e}^{i\phi} & 0 \\
0 & -\mathrm{e}^{-i\phi} & \q & 0 \\
0 & 0 & 0 & 0 \end{array} \right),
\label{eq_perioXXZ}
\end{equation}
where of course it is implied that $e_{2N}$ acts as the identity operator on all the other spins.
The resulting Hamiltonian $H=-\sum_{i} e_i$ reads, up to an irrelevant constant
\begin{equation}
\displaystyle H = \frac{1}{2} \sum_{i=1}^{2N} \left( \sigma^x_i \sigma^x_{i+1} + \sigma^y_i \sigma^y_{i+1}  + \frac{\q + \q^{-1}}{2} \sigma^z_i \sigma^z_{i+1}  \right) + \frac{\mathrm{e}^{i\phi}}{4} \sigma_{2N}^{+}  \sigma_{1}^{-}
+ \frac{\mathrm{e}^{-i\phi}}{4} \sigma_{2N}^{-}  \sigma_{1}^{+}.
\end{equation}
We shall refer to this model as  the \textit{twisted  XXZ spin chain}.
 The choice of the twist
$\mathrm{e}^{i\phi}=\mathrm{e}^{2iK}$ allows to select specific generically
irreducible representations of the affine TL algebra. For $j \neq 0$, the Hilbert space of this model
in the sector with the total spin $S_z=j$ is isomorphic to the standard module
$\StTL{|j|,\mathrm{e}^{\pm2iK}}$, where `$+$' is for positive $j$ and  `$-$' is for negative value of $j$. This can be easily checked using the translation generator
$u$ of the affine TL  represented in the twisted chain as
\begin{equation}\label{XXZ-trans}
u = e^{i\frac{\phi}{2}\sigma^z_1}s_1 s_2\dots s_{2N-1},
\end{equation}
where the $s_i$'s are the permutations of the $i$th and $(i+1)$th sites.

The scaling limit of each sector can be inferred using the basic fact~\cite{Cardy} that  the generating function of the energy and momentum spectra is related
 to conformal spectra (for the critical Hamiltonian at $|\q|=1$) as
\begin{equation}
 \hbox{Tr}\, e^{-\beta_R(H-Ne_0)}e^{-i\beta_I P}\;\xrightarrow{\, N\to\infty\,}\; \hbox{Tr}\, q^{L_0-c/24}\bar{q}^{\bar{L}_0-c/24},\label{charform}
 \end{equation}
 where $H$ and $P$ are the lattice hamiltonian  (normalized such that the velocity of sound is unity) and momentum,  $e_0$ is the ground state energy per site in the thermodynamic limit, we also set $q(\bar{q})=\exp\left[-{2\pi\over N}(\beta_R\pm i\beta_I)\right]$ with $\beta_{R,I}$ real and $\beta_R>0$, and $N$ is the length of the chain.
 The trace on the left is taken over  the states of the spin chain in a given sector, and the trace on the right over the  states occurring in this sector  in the continuum limit.

  The traces of the scaling hamiltonian in the generic irreducible representations are easily worked out using the XXZ hamiltonian, to which methods such as the Bethe ansatz or Coulomb gas mappings can be readily applied~\cite{PasquierSaleur}.
Let us introduce the usual notations for the central charge and the conformal weights
\begin{subequations}
\begin{eqnarray}
\displaystyle c &=& 1 - \frac{6}{x (x+1)}, \\
\displaystyle h_{r,s} &=& \frac{ \left[(x+1)r - xs \right]^2 - 1}{4 x (x+1)},
\end{eqnarray}
\end{subequations}
where with this parametrization $\q = \mathrm{e}^{i \pi/ (x+1)}$. For $x=2$ or $c=0$ the conformal weights $h_{r,s}$ are arranged in the Kac table in Fig.~\ref{Kac-t}.

\begin{figure}\centering
    \includegraphics[scale=0.4]{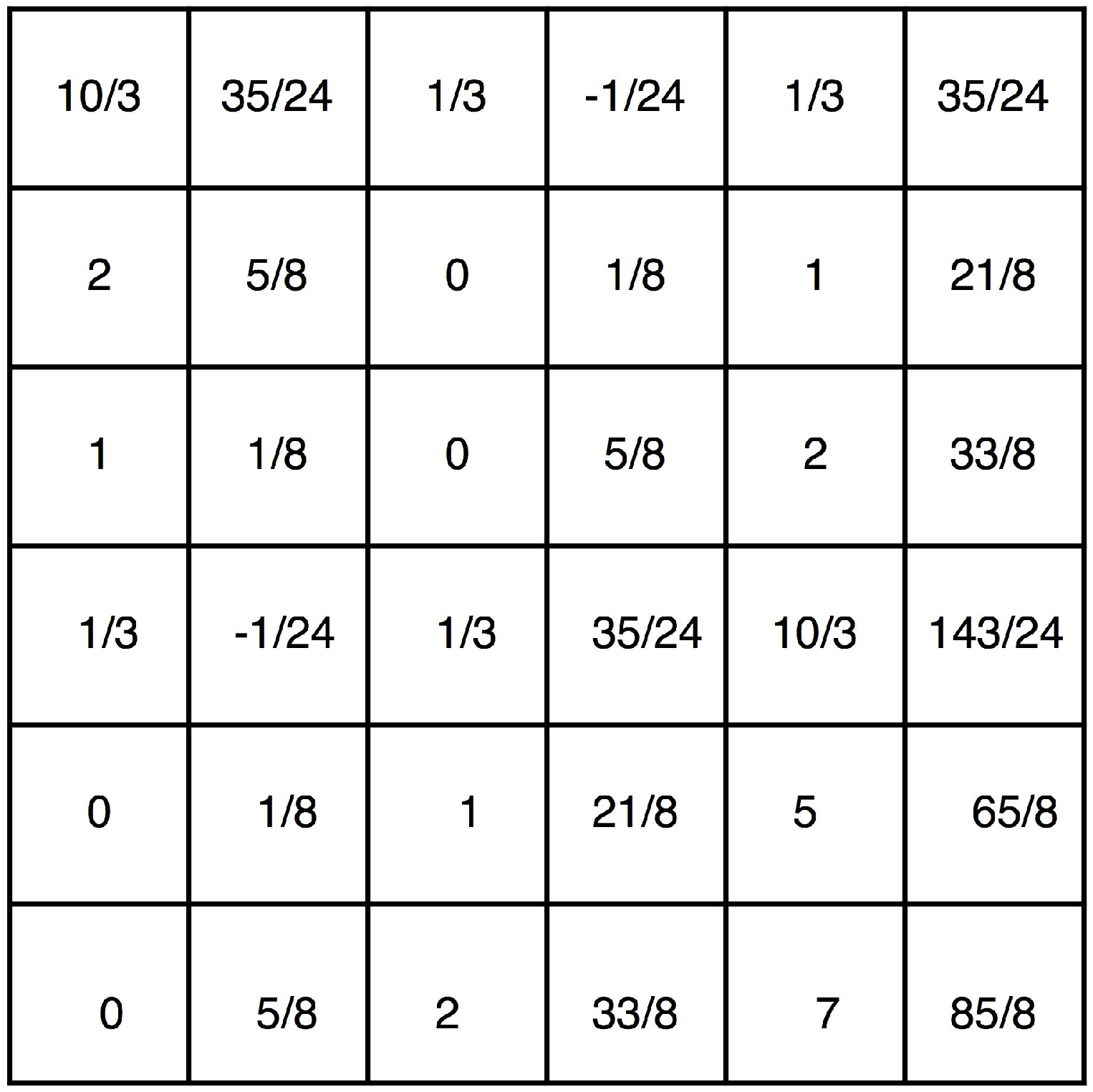}
      \caption{Kac table for $c=0$.}\label{Kac-t}
    \end{figure}

The trace on the left hand side of \eqref{charform} taken over the subspace ${\cal H}_j$ of  spin projection $S^z=j$, with $-L\leq
j\leq L$, in the XXZ chain of length $N=2L$ has the limit \cite{PasquierSaleur,Rittenbergetal}
 \begin{equation}
 \hbox{Tr}_{{\cal H}_j}\, e^{-\beta_R(\HXXZ-2Le_0)}e^{-i\beta_I \PXXZ}\;\xrightarrow{\, N\to\infty\,}\; F_{j,e^{2iK}}
 \end{equation}
where  $K=\frac{\pi}{M} p$ and
 \begin{equation}\label{F-func}
 F_{j,e^{2iK}} = \ffrac{q^{-c/24}\bar{q}^{-c/24}}{P(q)P(\bar{q})}\sum_{n \in \mathbb{Z}} q^{h_{n+p/M,-j}} \bar{q}^{h_{n+p/M,j}},
\end{equation}
and
\begin{equation}
\displaystyle P(q) =  \prod_{n=1}^{\infty} (1 - q^n) = q^{-1/24} \eta (q).
\end{equation}

We note that the expression~\eqref{F-func} is a formally infinite sum over products of characters
\begin{equation}\label{Vir-ch}
	k_{r,s} = \frac{q^{h_{r,s}-c/24}}{P(q)}, \qquad \overline{k}_{r,s} = \frac{\bar{q}^{h_{r,s}-c/24}}{P(\bar{q})}
	\end{equation}
 of the Verma representations of the Virasoro algebra.

Recall that $\StJTL{j}{\mathrm{e}^{2 i K}}$ modules over the $\rJTL{N}$ algebra are irreducible at generic values of $\q$, where the pseudomomenta $\mathrm{e}^{2iK}$ are now taken as $j$th roots of unity, and we have an isomorphism ${\cal H}_j\cong\StTL{j,\mathrm{e}^{2 i K}}$ for $j>0$ and the choice of the twist $\phi = 2K$.
Therefore, the generating function of levels in the scaling limit of the JTL modules  $\StJTL{j}{\mathrm{e}^{2 i K}}$, can be simply written as
\begin{equation}
\mathrm{Tr}_{\,\StTL{j,\mathrm{e}^{2 i K}}} q^{L_0-c/24}\bar{q}^{\bar{L}_0-c/24} \equiv F_{j,e^{2iK}},
\end{equation}
where  we used same notations for the  modules $\StJTL{j}{\mathrm{e}^{2 i K}}$ and their scaling limits\footnote{By the scaling limit of modules, we mean an appropriate inductive limit of them corresponding to $N\to\infty$. Though in general, it is very hard to construct such limits explicitly, see some examples of a rigorous construction in \cite{GRS3}.}.  The character formulas for the scaling limit of the JTL simple modules at generic $\q$ shows that the  scaling limit of the $\rJTL{N}$ algebra should be  an operator algebra containing $\VirN$ as a proper subalgebra.

The case $j=0$ requires more care as one has to be careful about the loops that wrap around the spatial direction due to the periodic boundary conditions which would get a weight $2$ if $\phi=0$ in~\eqref{eq_perioXXZ}.
One then needs to introduce a twist $\mathrm{e}^{i\phi}=\q^2$ to account for this\footnote{It is indeed easy to check that within our twisted XXZ representation, non-contractible loops carry a weight $2 \cos \frac{\phi}{2}$.}.  In this case, the sector $S_z=0$ corresponds to the standard module\footnote{We note that choosing the twist as $\mathrm{e}^{i\phi}=\q^{-2}$ gives $\mathcal{H}_{0}$ as the costandard module $\AStTL{0}{\q^2}^*=   \StJTL{1}{1} \rightarrow \bAStTL{0}{\q^2}$.} $\mathcal{H}_{0} \simeq \AStTL{0}{\q^2} = \bAStTL{0}{\q^2} \rightarrow \StJTL{1}{1}$.
The trace over the scaling limit of the representation $\bAStTL{0}{\q^2}$ thus reads
\begin{equation}
 \mathrm{Tr}_{\,\bAStTL{0}{\q^2}} q^{L_0-c/24}\bar{q}^{\bar{L}_0-c/24} = F_{0,\q^2}-F_{1,1},
\end{equation}
or  introducing
the character of the Kac representation
\begin{equation}\label{Kac-ch}
K_{r,s}={q^{h_{r,s}}-q^{h_{r,-s}}\over P(q)},
\end{equation}
we have
\begin{equation}
\hbox{Tr }_{\bAStTL{0}{\q^2}}q^{L_0-c/24}\bar{q}^{\bar{L}_0-c/24}=F_{0,q^2}-F_{1,1}=\sum_{r\geq 1} K_{r,1}\bar{K}_{r,1}.\label{firstsub}
 \end{equation}
  Note that we have
\begin{equation}
F_{0,e^{2iK}}=F_{0,e^{-2iK}}.
\end{equation}

 Finally, we quote a formula that will play a fundamental role later:
 \begin{equation}
 F_{j,\q^{2j+2k}}-F_{j+k,\q^{2j}}=\sum_{r=1}^\infty K_{r,k}\bar{K}_{r,k+2j}\label{useful}
 \end{equation}
 with, recall,  $\q=e^{i\pi/(x+1)}$. Of course, (\ref{firstsub}) is just a particular case where $j=0$ and $k=1$.
We note that the pseudomomenta ($j$th roots of unity) for our JTL modules appear in this formula only at integer values of $x$.

\medskip

The generating functions of standard modules in~\eqref{F-func} can be obviously written in terms of the characters $k_{r,s}$ of Virasoro Verma modules. Even at generic values of $x$ or the central charge $c$, it does not mean, of course, that the scaling limit of the corresponding JTL modules is a direct sum of products of Verma modules. For example, the Feigin--Fuchs (FF) module $\rep{F}_{r,s}$ has the same character $k_{r,s}$ as the Verma module of weight $h_{r,s}$, but is not Verma (see a review on FF modules in App.~\ref{appFF}). Actually, it is known that XXZ chains are closely related to the  Coulomb Gas model \cite{PasquierSaleur}.
Given that FF modules are constructed directly from a free boson picture, it is rather natural to expect our characters to describe the product of FF modules instead of Verma modules. We will actually see later directly at  $c=0$ that assuming (a filtration by) Virasoro Verma modules in the scaling limit of the JTL  standards will be in contradiction with the algebraic structure, while Feigin--Fuchs modules give a consistent picture.

 We also note that  it is not clear whether at generic values of $\q$  the scaling limit of each JTL standard (also simple) module $\StTL{j,\mathrm{e}^{2 i K}}$ considered as a $\VirN$ module is self-dual. This is related with  the question of whether loop models
at generic values of $\q$ (or $x$) are ``physical'', that is, described by consistent local bulk CFTs\footnote{Recall that one of consistency conditions for bulk CFTs is a non-degeneracy of the two-point function that is defined as the Virasoro-algebra invariant bilinear form on the space $\Hilb$ of all non-chiral fields. The non-degeneracy condition implies that the space of linear forms on $\Hilb$ is isomorphic to $\Hilb$ itself, as $\VirN$ modules, or in other words $\Hilb$ has to be self-dual.}.
The answer is not clear, because for generic $\q$, there exists no (supersymmetric) spin chain formulation, no self-dual `Hilbert space', etc. Self-duality of the limits of
 $\StTL{j,\mathrm{e}^{2 i K}}$, as $\VirN$ modules, at generic values of $\q$ on the other hand would have very important practical consequences: this  will be
discussed elsewhere.

\section{Operator content of simple JTL modules}\label{sec:opconsimmod}

We start our analysis by discussing the $\VirN$ content of the simple $\rJTL{N}(1)$ modules in the continuum limit for $c=0$. By this, we mean the representation content of the states  that contribute to the scaling limit in the $\rJTL{N}$ modules. It is convenient for this purpose to  first evaluate the generating functions of levels in some of the  $\rJTL{N}$ modules we have encountered previously.

 \subsection{Characters for the limit of JTL simple modules}\label{sec:char-JTLsimples}
There are two possible situations, apart from the fully non degenerate one. The first occurs for pairs $(j,e^{2iK})=(0 \hbox{ mod }3,e^{2i\pi/3})$, $(j,e^{2iK})=(0\hbox{ mod }3,e^{4i\pi/3})$, $(1\hbox{ mod }3,1)$ or $(2\hbox{ mod }3,1)$. In this case, the structure of submodules is represented by the first two diagrams on Fig.~\ref{fig:cell-thirdroot}.
%

The characters of the corresponding JTL simples $\IrrJTL{j}{P}$ in the limit
\begin{equation}
F^{(0)}_{j,P} \equiv \mathrm{Tr}_{\IrrJTL{j}{P}} q^{L_0-c/24}\bar{q}^{\bar{L}_0-c/24}
\end{equation}
 are obtained  by a series of subtractions and additions, just like in the  computation of  dimensions for  finite chains: the only difference is that the series is infinite -- but convergent.
As an example, we consider
\begin{eqnarray}
F_{0,e^{2i\pi/3}}^{(0)}=F_{0,e^{2i\pi/3}}-F_{1,1}-F_{2,1}+F_{3,e^{2i\pi/3}}+F_{3,e^{4i\pi/3}}-F_{4,1}-F_{5,1}+\ldots\nonumber\\
=\sum_{n=0}^\infty F_{3n,e^{2i\pi/3}}-F_{3n+1,1}-\sum_{n=0}^\infty \left(F_{3n+2,1}-F_{3n+3,e^{4i\pi/3}}\right).
\end{eqnarray}
We now use the basic identities
\begin{eqnarray}
F_{3n,e^{2i\pi/3}}-F_{3n+1,1}=\sum_{r=1}^\infty K_{r,1}\bar{K}_{r,6n+1},\nonumber\\
F_{3n+2,1}-F_{3n+3,e^{4i\pi/3}}=\sum_{r=1}^\infty K_{r,1}\bar{K}_{r,6n+5}
\end{eqnarray}
together with, for $l=1,2$,
\begin{gather*}
K_{2k-1,l}=\chi_{2k-1,l}+\chi_{2k+1,l},\\
K_{2k,l}=\chi_{2k,l},
\end{gather*}
where the $\chi$'s are as usual characters of the Virasoro algebra simples. By straightforward manipulations we find first
\begin{equation}
F_{0,e^{2i\pi/3}}^{(0)}=\sum_{r=1}^\infty K_{r1}\sum_{n=0}^\infty \left(\bar{K}_{r,6n+1}-\bar{K}_{r,6n+5}\right)=\chi_{11}\bar{\chi}_{11}=1
\end{equation}
as was expected because the dimension of $\IrrJTL{0}{\q^2}$ is one for any even number of sites.
Note also that $\chi_{12}=\chi_{11}=1$.

We then find
\begin{eqnarray}
F_{1,1}^{(0)}=\sum_{k=1}^\infty \chi_{2k,2}\bar{\chi}_{2k,2}
+\left(\chi_{2k-1,2}+\chi_{2k+1,2}\right)\left(\bar{\chi}_{2k-1,2}+\bar{\chi}_{2k+1,2}\right)-\chi_{1,1}\bar{\chi}_{1,1},\nonumber\\
F_{2,1}^{(0)}=\sum_{k=1}^\infty \chi_{2k,1}\bar{\chi}_{2k,1}
+\left(\chi_{2k-1,1}+\chi_{2k+1,1}\right)\left(\bar{\chi}_{2k-1,1}+\bar{\chi}_{2k+1,1}\right)-\chi_{1,1}\bar{\chi}_{1,1}.\label{F21-0}
\end{eqnarray}
Interestingly, we can for instance rewrite
\begin{eqnarray}
F_{1,1}^{(0)}&=&\sum_{k=1}^\infty\chi_{2k,2}\bar{\chi}_{2k,2}+\chi_{2k+1,1}\left(\bar{\chi}_{2k-1,2}+2\bar{\chi}_{2k+1,2}+\bar{\chi}_{2k+3,2}\right)\nonumber\\
&=&\sum_{k=1}^\infty\chi_{2k,2}\bar{\chi}_{2k,2}+\left(\chi_{2k-1,2}+2\chi_{2k+1,2}+\chi_{2k+3,2}\right)
\bar{\chi}_{2k+1,2}.
\end{eqnarray}
We give in Appendix~\ref{AppCharacters} explicit formulas for many other $F_{j,e^{2iK}}$ in terms of the left and right Virasoro characters.

We note that the leading conformal weighs in $F_{2n+1,1}^{(0)}$ are $(h,\bar{h})=(h_{1,2},h_{3+4n,2})$ and the same with $h\leftrightarrow \bar{h}$. In $F_{2n+2,1}^{(0)}$ we have $(h_{1,1},h_{3+4n,1})$. In $F_{3n,e^{2i\pi/3}}^{(0)}$ we get $(h_{1,1},h_{1+4n,2})$ and $(h_{1+4n,1},h_{1,2})$, in $F_{3n,e^{4i\pi/3}}^{(0)}$ we get $(h_{1,2},h_{1+4n,1})$ and $(h_{1+4n,2},h_{1,1})$ (recall $h_{1,1}=h_{1,2}$ so we have left-right symmetry).

\subsection{Left-right Virasoro structure of simple JTL modules}\label{sec:lrVir-simple}
We stress that so far we only computed characters $F_{j,e^{2iK}}$ of the (scaling limit of) the simple JTL modules $\IrrJTL{j}{e^{2iK}}$:  an additional analysis  is required to see whether each of the modules $\IrrJTL{j}{e^{2iK}}$ is a direct sum of simples over left-right Virasoro $\VirN$, as was in the case of $\gl(1|1)$ spin chains~\cite{GRS3}. Unfortunately, we give strong arguments below  that JTL simples (those belonging to an indecomposable block of JTL modules) in the case of the $s\ell(2|1)$ spin chain  are not, in general,  direct sums of Virasoro simples -- and involve instead reducible but indecomposable modules. This suggests that
 the analysis of the scaling limit should involve a  bigger algebra than just $\VirN$ -- in other words, that the  scaling limit of the JTL algebra is more than $\VirN$. This will be discussed in the conclusion, and in subsequent work.

\newcommand{\Kac}{\mathcal{K}}
\newcommand{\bKac}{\bar{\mathcal{K}}}
We begin our analysis  with the standard module $\bStJTL{0}{\q^2}$.  We note that this JTL module is well-defined for generic $\q$ or $x$. For such  values of $x$, using the character~\eqref{firstsub} and the fact that weights $h_{r,1}$ do not differ by an integer for different integer values of $r$, the scaling limit of the vacuum standard module, which we will denote by the same symbol $\bStJTL{0}{\q^2}$ as for a finite lattice, is decomposed over the left-right Virasoro onto the direct sum of simple $\VirN$-modules
\begin{equation}\label{vac-mod-gen}
\bStJTL{0}{\q^2} = \bigoplus_{r\geq1}\Kac_{r,1}\boxtimes\bKac_{r,1}
\end{equation}
where $\Kac_{r,s}$ are  Virasoro modules
 with the Kac characters $K_{r,s}$ given in~\eqref{Kac-ch} and we also introduce the corresponding anti-chiral modules $\bKac_{r,s}$. These modules are  simple at generic central charges $c$ or generic values of our parameter $\q$. Moreover, there cannot be  glueings/extensions among the  modules $\Kac_{r,1}$ in such generic cases, so the limit of $\bStJTL{0}{\q^2}$ must be a direct sum.

Note that the first term $\Kac_{1,1}\boxtimes \bKac_{1,1}$ contains the identity (or the vacuum state of dimension $(0,0)$), for any $x$, and  its descendants with respect to left-right Virasoro, while the next term $\Kac_{2,1}\boxtimes\bKac_{2,1}$ is spanned by the descendants of the primary field of conformal dimension $\bigl(\frac{3+x}{4x},\frac{3+x}{4x}\bigr)$ (or $(\frac{5}{8},\frac{5}{8})$ at $c=0$), the term $\Kac_{3,1}\boxtimes\bKac_{3,1}$ is spanned by the primary field of dimension $\bigl(\frac{2+x}{x},\frac{2+x}{x}\bigr)$ (which equals $(2,2)$ at $c=0$) and its descendants, {\it etc.}

In contrast, when  $\q=e^{i\pi/3}$ or $c=0$,  the vacuum standard module (scaling limit of $\bStJTL{0}{\q^2}$)  is not  a direct sum of simple Virasoro modules anymore. The subsequent analysis requires, for the time being,   a series of natural assumptions that we spell out below. These assumptions are checked a posteriori by the consistency of the full picture.

When  $\q=e^{i\pi/3}$, or $c=0$, each  module $\Kac_{r,1}$, for odd values of $r$, becomes indecomposable but reducible with the subquotient structure
\begin{equation}
   \xymatrix@R=10pt@C=6pt
   { &\\
   \Kac_{r,1}: }\qquad
   \xymatrix@R=19pt@C=16pt
   {
     &{\circ}\ar@{}|{\substack{h_{r,1}}\kern-7pt}[]+<-30pt,20pt>\ar[dr]&\\
     &&{\circ}\ar@{}|{\substack{h_{r+2,1}}\kern-7pt}[]+<30pt,15pt>
     }
\end{equation}
and each tensor product $\Kac_{r,1}\boxtimes\bKac_{r,1}$ in~\eqref{vac-mod-gen}, assuming the decomposition for $\q=e^{i\pi/3}$ involves the `continuation' of the modules present for $\q$ generic, becomes the following left-right Virasoro module
\begin{equation}
   \xymatrix@R=13pt@C=6pt
   { &\\
   \Kac_{r,1}\boxtimes\bKac_{r,1}: }\qquad
   \xymatrix@R=19pt@C=16pt
   {
     &{\circ}\ar@{}|{\substack{(h_{r,1},\bar{h}_{r,1})}\kern-7pt}[]+<-30pt,20pt>\ar[dr]\ar[dl]&\\
     {\circ}\ar@{}|{\substack{(h_{r,1},\bar{h}_{r+2,1})}\kern-7pt}[]+<-60pt,15pt>\ar[dr]
     &&{\circ}\ar@{}|{\substack{(h_{r+2,1},\bar{h}_{r,1})}\kern-7pt}[]+<50pt,15pt>\ar[dl]\\
     &{\circ}\ar@{}|{\substack{(h_{r+2,1},\bar{h}_{r+2,1})}\kern-7pt}[]+<-60pt,-15pt> &\\
     }
\end{equation}

Moreover, these diamonds could be in principle glued with each other: what we  can  say at this stage of our analysis is that the vacuum standard module has a filtration by the left-right Virasoro modules $\Kac_{r,1}\boxtimes\bKac_{r,1}$ (recall our definition of the filtration around~\eqref{eq:filtr}).
This means that $\Kac_{1,1}\boxtimes\bKac_{1,1}$ is a Virasoro submodule (as it contains the vacuum and we expect the trivial action of positive modes on the identity and energy-momentum tensors), that  there might be arrows due to positive Virasoro modes action from $\Kac_{3,1}\boxtimes\bKac_{3,1}$ into $\Kac_{1,1}\boxtimes\bKac_{1,1}$, and so on.
Modulo the arrows connecting different diamonds we have the structure of the full vacuum module $\bStJTL{0}{\q^2}$:
\begin{multline}\label{scal-lim-mod-1}
  \xymatrix@R=13pt@C=1pt
   { &\\
   \bStJTL{0}{\q^2}: }
     \xymatrix@R=28pt@C=1pt
   {
   {\circ}\ar@{}|{\substack{\IrrJTL{0}{\q^2}}\kern-7pt}[]+<-20pt,20pt>\ar[dr]&&\\
    &{\bullet}\ar@{}|{\substack{\IrrJTL{2}{1}}\kern-7pt}[]+<-20pt,-20pt>&
    }
     \xymatrix@R=10pt@C=1pt
   { &\\
   \xrightarrow{\; N\to\infty\;}
   }
   \xymatrix@R=19pt@C=16pt
   {
     &{\circ}\ar@{}|{\substack{(0,0)}\kern-7pt}[]+<-20pt,20pt>\ar[dr]\ar[dl]&\\
     {\bullet}\ar@{}|{\substack{(0,2)}\kern-7pt}[]+<-20pt,15pt>\ar[dr]
     &&{\bullet}\ar@{}|{\substack{(2,0)}\kern-7pt}[]+<15pt,15pt>\ar[dl]\\
     &{\bullet}\ar@{}|{\substack{(2,2)}\kern-7pt}[]+<-20pt,-15pt> &\\
     }
       \xymatrix@R=13pt@C=1pt
   { &\\
   \qquad\oplus }
   \xymatrix@R=19pt@C=16pt
   {
     &{\bullet}\ar@{}|{\substack{(2,2)}\kern-7pt}[]+<-20pt,20pt>\ar[dr]\ar[dl]&\\
     {\bullet}\ar@{}|{\substack{(2,7)}\kern-7pt}[]+<-25pt,15pt>\ar[dr]
     &&{\bullet}\ar@{}|{\substack{(7,2)}\kern-7pt}[]+<15pt,15pt>\ar[dl]\\
     &{\bullet}\ar@{}|{\substack{(7,7)}\kern-7pt}[]+<-20pt,-15pt> &\\
     }
          \xymatrix@R=13pt@C=1pt
   { &\\
   \qquad\oplus\qquad\dots }\\
   \oplus \quad(\ffrac{5}{8},\ffrac{5}{8})\quad \oplus\quad (\ffrac{33}{8},\ffrac{33}{8})\quad\oplus \dots
\end{multline}
Note once again that these diamonds are just products of the indecomposable Kac modules $\Kac_{r,1}$ at the logarithmic $c = 0$ point, and that the existence of a filtration by these products of Kac modules follows from the generic point decomposition~\eqref{vac-mod-gen}. All the modules $\Kac_{r,1}$ can be equivalently obtained as quotients/submodules of the Feigin--Fuchs modules introduced in App. \ref{appFF}. We see from this filtration that the scaling limit of the simple JTL module $\IrrJTL{2}{1}$ is not a direct sum of simples over the Virasoro algebra but a direct sum of indecomposable but reducible modules with subquotients marked by `$\bullet$', and of  those having  fractional conformal dimensions.

Note that the module $\IrrJTL{0}{\q^2}$ is always one dimensional; it is spanned by the vacuum state, and its scaling limit is marked by `$\circ$' and corresponds to the trivial Virasoro representation. The two states generating subquotients $(2,0)$ and $(0,2)$ are respectively the chiral $T$ and anti-chiral $\bar{T}$ energy-momentum tensors, and the state generating $(2,2)$ is their product $T\bar{T}$.
Note finally that we could assume a filtration of $\bStJTL{0}{\q^2}$ by products of modules \textit{dual} to the Kac ones, i.e. by those with reversed arrows in their subquotient structure (the character would be the same). In our case however, the unique vacuum state would be a descendent of the energy-momentum tensors $T$ and $\bar{T}$. Of course, this is not allowed.

We now discuss the scaling limit of all other JTL standards $\StJTL{j}{z^2}$ and the corresponding simples. For these purposes it is technically easier at  $x=2$ (or $c=0$) to introduce intermediate modules which are quotients similar to $\bStJTL{0}{\q^2}$:
\begin{equation}
\bStJTL{j}{\q^{2j+2k}} = \StJTL{j}{\q^{2j+2k}}/\StJTL{j+k}{\q^{2j}}, \qquad j\geq 0,\quad k\in\{1,2\}.
\end{equation}
Now, each module $\bStJTL{j}{\q^{2j+2k}}$ is a glueing  of two JTL simples. For example, we have
\begin{gather}
\bStJTL{1}{1} = \IrrJTL{1}{1}\to\IrrJTL{3}{\q^4},\qquad
\bStJTL{2}{1} = \IrrJTL{2}{1}\to\IrrJTL{3}{\q^2},\\
\bStJTL{3}{\q^2} = \IrrJTL{3}{\q^2}\to\IrrJTL{5}{1},\qquad
\bStJTL{3}{\q^4} = \IrrJTL{3}{\q^4}\to\IrrJTL{4}{1},\\
\bStJTL{4}{1} = \IrrJTL{4}{1}\to\IrrJTL{6}{\q^4},\qquad
\bStJTL{5}{1} = \IrrJTL{5}{1}\to\IrrJTL{6}{\q^2},
\end{gather}
{\it etc}.  In terms of these modules, we can now describe the structure of all JTL standards  by a chain structure
\begin{equation}
\StJTL{2}{1} = \bStJTL{2}{1}\to\bStJTL{3}{\q^4}\to\bStJTL{5}{1}\to\dots.
\end{equation}

The important thing is that the intermediate JTL modules $\bStJTL{j}{z^2}$ have nice scaling limit properties. To start, we can consider the characters  \begin{equation}\label{Tr-Kac-gen}
\mathrm{Tr}_{\,\bAStTL{j}{\q^{2j+2k}}} q^{L_0-c/24}\bar{q}^{\bar{L}_0-c/24} =
  F_{j,\q^{2j+2k}}-F_{j+k,\q^{2j}}=\sum_{r=1}^\infty K_{r,k}\bar{K}_{r,k+2j}.
 \end{equation}
with the right hand side a sum of  products of (Virasoro) Kac modules characters. For example, the generating functions for the scaling limit of $\bStJTL{1}{1}$, $\bStJTL{2}{1}$, and $\bAStTL{3}{\q^2}$ obey
\begin{gather}
 \mathrm{Tr}_{\,\bAStTL{1}{1}} q^{L_0-c/24}\bar{q}^{\bar{L}_0-c/24} 
 =\sum_{r=1}^\infty K_{r,2}\bar{K}_{r,4},\label{Kac-W11}\\
 \mathrm{Tr}_{\,\bAStTL{2}{1}} q^{L_0-c/24}\bar{q}^{\bar{L}_0-c/24}
 =\sum_{r=1}^\infty K_{r,1}\bar{K}_{r,5},\label{Kac-W21}\\
 \mathrm{Tr}_{\,\bAStTL{3}{\q^2}} q^{L_0-c/24}\bar{q}^{\bar{L}_0-c/24} 
 =\sum_{r=1}^\infty K_{r,1}\bar{K}_{r,7}.\label{Kac-W3q}
\end{gather}
We can again say here that the scaling limit of the modules $\bStJTL{1}{1}$, $\bStJTL{2}{1}$, $\bStJTL{3}{\q^2}$, {\it etc.}, is filtered by the product of {\sl indecomposable} Virasoro modules $\Kac_{r,n}$ having the  character equal $K_{r,n}$ (the Kac character),
i.e., the first term in the sums corresponds to a submodule, the third term is glued
 by positive Virasoro modes action to the first one, and so on.
The only difference from the case $\bStJTL{0}{\q^2}$ is that the anti-chiral modules $\bar{K}_{r,n}$, with $n>3$, (they are generically irreducible) are more complicated  in the case $c=0$. To obtain the character $K_{r,n}$, the singular vector in the corresponding Verma module to be set to zero is at the level $nr$. For example, the  module $\Kac_{3,4}$ has the same highest weight $h_{3,4}=0$ as $\Kac_{1,1}$, but its subquotient structure has six nodes (not two) corresponding to  Virasoro irreducibles.  The module $\Kac_{5,7}$ has the same highest weight and its structure is even more complicated, and so on.
The other problem is that the limit $\Kac_{r,n}$ of generically irreducible Virasoro modules might not be a quotient of the corresponding Verma modules. Indeed, we will see below that they are actually quotients/submodules of Feigin--Fuchs modules\footnote{Recall that Feigin--Fuchs modules are defined by Virasoro-module sructure on the Fock spaces in the Coulomb gas formalism or modified free-boson theory.}. Note also that the first component in the sum~\eqref{Tr-Kac-gen} is always $K_{r,1}$ or $K_{r,2}$ and the corresponding  modules have two subquotients only and they are indeed Kac modules (quotients of the Verma modules) and at the same time quotients of the Feigin--Fuchs modules.

To show the complexity of the scaling limit of the quotient JTL modules $\bStJTL{j}{z^2}$ we describe the left-right Virasoro structure of several  terms in the character sum~\eqref{Kac-W21} corresponding to $\bStJTL{2}{1}$. The first term $K_{1,1}\bar{K}_{1,5}$ has the diamond structure on the left part of Fig.~\ref{KK-FF}, {\it i.e.}, four subquotients, a situation encountered before; the next term $K_{2,1}\bar{K}_{2,5}$ corresponds to the left-right Virasoro module of a chain type $(\frac{5}{8},\frac{5}{8})\to (\frac{5}{8},\frac{21}{8})$. The first interesting term is $K_{3,1}\bar{K}_{3,5}$. If we assumed that the corresponding module at $c=0$ was the product of two Kac modules (quotients of Verma) then  the corresponding  structure would be
\begin{equation}\label{KK-Kac}
  \xymatrix@C=18pt@R=15pt@M=2pt@W=1pt{%
    &&{\stackrel{(2,0)}{\bullet}}\ar[dll]\ar[dr]\ar[drr]&&&\\
    {\stackrel{(7,0)}{\circ}}\ar[dr]\ar[drr]&&
    &{\stackrel{(2,1)}{\circ}}\ar[dll]\ar[dr]\ar[drr]&{\stackrel{(2,2)}{\bullet}}\ar[d]\ar[dr]\ar[dll]&\\
    &{\stackrel{(7,1)}{\circ}}\ar[drr]\ar[dr]&{\stackrel{(7,2)}{\bullet}}\ar[dr]\ar[d]&
    &{\stackrel{(2,5)}{\circ}}\ar[dll]\ar[drr]&{\stackrel{(2,7)}{\bullet}}\ar[dll]\ar[dr]\\
    &&  {\stackrel{(7,5)}{\circ}}\ar[drr] &{\stackrel{(7,7)}{\bullet}}\ar[dr]
    &&&{\stackrel{(2,12)}{\circ}}\ar[dll]\\
    && && {\stackrel{(7,12)}{\circ}}&
 }
\end{equation}
where the down-left arrows describe the chiral Virasoro action and
down or down-right ones are for the antichiral part. We also denote simple Virasoro subquotients contributing to the scaling limit of $\IrrJTL{2}{1}$ by `$\bullet$' and those from the limit of $\IrrJTL{3}{\q^2}$ by `$\circ$'. We used here the corresponding Virasoro character expressions in~\eqref{F21-0} and~\eqref{F3q-0}. As we can see from the diagram in~\eqref{KK-Kac}, there would be a problem then -- there are arrows $\circ\to\bullet$ mapping states from $\IrrJTL{3}{\q^2}$ to states in $\IrrJTL{2}{1}$ and this contradicts to the structure of $\bStJTL{2}{1}: \IrrJTL{2}{1} \to \IrrJTL{3}{\q^2}$. To solve this problem we should actually assume that $\bar{\Kac}_{3,5}$, instead of a Kac module,  is the quotient of a Feigin--Fuchs module, where  half  the arrows in the corresponding Verma module are reversed. We then obtain the consistent structure shown on the right of Fig.~\ref{KK-FF}
\begin{figure}[tbp]
\begin{equation*}
{\footnotesize  \xymatrix@R=19pt@C=0pt
   { &\\
   \Kac_{1,1}\boxtimes\bar{\Kac}_{1,5}:& }
   }
   \xymatrix@R=19pt@C=16pt
   {
     &{\bullet}\ar@{}|{\substack{(0,2)}\kern-7pt}[]+<-20pt,20pt>\ar[dr]\ar[dl]&\\
     {\circ}\ar@{}|{\substack{(0,5)}\kern-7pt}[]+<-20pt,15pt>\ar[dr]
     &&{\bullet}\ar@{}|{\substack{(2,2)}\kern-7pt}[]+<15pt,15pt>\ar[dl]\\
     &{\circ}\ar@{}|{\substack{(2,5)}\kern-7pt}[]+<-20pt,-15pt> &\\
     }\qquad
{\footnotesize  \xymatrix@R=19pt@C=1pt
   { &\\
   \Kac_{3,1}\boxtimes\bar{\Kac}_{3,5}:& }
   }
  \xymatrix@C=18pt@R=15pt@M=2pt@W=1pt{%
    &&{\stackrel{(2,0)}{\bullet}}\ar[dll]\ar[dr]&&&\\
    {\stackrel{(7,0)}{\circ}}\ar[dr]&&
    &{\stackrel{(2,1)}{\circ}}\ar[dll]&{\stackrel{(2,2)}{\bullet}}\ar[d]\ar[dr]\ar[dll]\ar[ull]&\\
    &{\stackrel{(7,1)}{\circ}}&{\stackrel{(7,2)}{\bullet}}\ar[dr]\ar[d]\ar[ull]&
    &{\stackrel{(2,5)}{\circ}}\ar[dll]\ar[drr]\ar[ul]&{\stackrel{(2,7)}{\bullet}}\ar[dll]\ar[dr]\ar[ull]\\
    &&  {\stackrel{(7,5)}{\circ}}\ar[drr]\ar[ul] &{\stackrel{(7,7)}{\bullet}}\ar[dr]\ar[ull]
    &&&{\stackrel{(2,12)}{\circ}}\ar[dll]\\
    && && {\stackrel{(7,12)}{\circ}}&
 }
\end{equation*}
\caption{Structure of the first two terms in the filtration of $\bStJTL{2}{1}$.}
\label{KK-FF}
\end{figure}
We  see once again from this analysis that the scaling limit of JTL simple modules is a direct sum of Virasoro indecomposable but reducible modules.

   Actually, all the terms $K_{3+2n,1}\bar{K}_{3+2n,5}$, with $n\geq0$, in the character sum~\eqref{Kac-W21} correspond to the modules $\Kac_{3+2n,1}\bar{\Kac}_{3+2n,5}$ with this $2\times6$-subquotient structure, and they are products of quotients of the Feigin--Fuchs  modules, or of the duals to them. So, there is a sort of stabilization in the structure of diagrams. Nevertheless, the diagrams for terms having the chain-type structure seem to be growing: the module $\Kac_{4,1}\bar{\Kac}_{4,5}$ has the structure $(\frac{33}{8},\frac{1}{8})\xleftarrow{\quad} (\frac{33}{8},\frac{33}{8})\xrightarrow{\quad} (\frac{33}{8},\frac{65}{8})$, where we again used the Feigin--Fuchs pattern of arrows, the module  $\Kac_{6,1}\bar{\Kac}_{6,5}$ has $4$ subquotients and so on.

As a first interesting result of our analysis, we get the Virsoro structure  of the simple JTL module $\IrrJTL{2}{1}$ in the scaling limit by comparing  its filtrations, in~\eqref{scal-lim-mod-1} and the one described in Fig.~\ref{KK-FF}. Note that in the two diagrams the nodes $(2,0)$, $(0,2)$ and twice $(2,2)$ labeled by $\bullet$'s are connected in different ways (both diagrams do not contradict to each other; they would  if we had used Verma and not Feigin--Fuchs structure). Therefore, all these arrows should be present in the structure for $\IrrJTL{2}{1}$. Further, to complete the diagram for $\IrrJTL{2}{1}$ we should recall that the only known existing indecomposable Virasoro module involving irreducible subquotients with the weights $2$ twice, $0$, and~$7$ is the staggered  module \cite{KytRid}
\begin{equation}
\xymatrix@R=10pt@C=8pt
{
& 2\ar[dl]\ar[dr]&\\
0\ar[dr]&&7\ar[dl]\\
&2&
}
\end{equation}
This observation and similar arguments for the antichiral part of the action give us all the additional arrows in the final struture
\begin{equation}\label{VirN-21}
  \xymatrix@R=17pt@C=2pt
   { &\\
   \IrrJTL{2}{1}:& }
   \xymatrix@R=19pt@C=16pt
   {
     &&{(2,2)}\ar[dr]\ar[drrr]\ar[dl]\ar[dll]&&\ldots\ar[dl]\ar[dr]&&\ldots\\
     (0,2)\ar[drr]&(2,0)\ar[dr]
     &&(2,7)\ar[dl]\ar[dr]
     &&(7,2)\ar[dlll]\ar[dl]&\ldots\\
     &&(2,2)&&(7,7)&&\ldots
     }
\end{equation}
 plus a direct sum of terms with non-integer (rational) weights. The dots in the diagram mean higher terms like $(7,7)$, etc. Following this structure, we also conclude that there is a Jordan cell of rank $2$ for the Hamiltonian in the scaling limit of $\IrrJTL{2}{1}$. We note that both the indecomposability parameters $b$ and $\bar{b}$ equal $\frac{5}{6}$ for the states with conformal weights $(2,2)$.

Using similar  analysis, we can in principle obtain $\VirN$ subquotient structure for the scaling limit of all other JTL simples. For example, the structure for $\IrrJTL{3}{\q^2}$ is
\begin{equation}\label{VirN-3q}
  \xymatrix@R=17pt@C=2pt
   { &\\
   \IrrJTL{3}{\q^2}:& }
   \xymatrix@R=19pt@C=16pt
   {
     &&{(2,5)}\ar[dr]\ar[drrr]\ar[dl]\ar[dll]&&\ldots\ar[dl]\ar[dr]&&\ldots\\
     (0,5)\ar[drr]&(2,1)\ar[dr]
     &&(2,12)\ar[dl]\ar[dr]
     &&(7,5)\ar[dlll]\ar[dl]&\ldots\\
     &&(2,5)&&(7,12)&&\ldots
     }
\end{equation}
Similarly, the structure for $\IrrJTL{3}{\q^{-2}}$
is obtained just by replacing $(h,\bar{h})$ by $(\bar{h}, h)$.

Our conclusion in this section is that the scaling limit of the JTL simples involved in the doubly critical class (containing $\IrrJTL{2}{1}$, $\IrrJTL{3}{\q^{\pm2}}$, $\IrrJTL{4}{1}$, etc) gives $\VirN$-modules with non-trivial Jordan cells of rank $2$ for the Hamiltonian $L_0+\bar{L}_0$. It may have escaped the reader at the end of this long discussion that  \textit{these Jordan cells are not present on the lattice\footnote{The corresponding states on a finite lattice have different eigenvalues of $H$ and thus can not be in a non-trivial Jordan cell. The two eigenvalues tend to the same value only in the limit $N\to\infty$.} and they only arise in the scaling limit}. This is definitely unpleasant, although perfectly possible: the existence of Jordan cells for all system sizes indicates a Jordan cell in the continuum theory as well, but the converse does not have to be true. Nevertheless, such phenomenon did not occur for boundary theories, or for the $\gl(1|1)$ spin chain. It implies, in particular, that the lattice algebraic analysis by itself can only provide lower bounds to the size of the Jordan cells.

\subsection{The field content up to level $(2,2)$}\label{sec:fields}

\renewcommand{\arraystretch}{1.5}
\begin{table}
\begin{center}
\begin{tabular}{|c|c|c|c|c|c|c|c|c|c|}
  \cline{1-10}
  & \multicolumn{8}{|c|}{Multiplicities} &  \\
  \cline{2-9}
    $(h,\overline{h})$ & $\TilJTL{0}{\q^2}$ & $\TilJTL{1}{1}$ & $\TilJTL{2}{1}$  & $\TilJTL{2}{-1}$ & $\TilJTL{3}{1}$ & $\TilJTL{3}{\q^2}$ & $\TilJTL{3}{\q^{-2}}$ & Total & $s\ell(2|1)$ representation\\
  \hline
 $(0,0)$ & 1 & 0 & 0 & 0 & 0 & 0 & 0 & 1 & $\atyp{0,0}$\\
 $(1,0)$ & 0 & 1 & 0 & 0 & 0 & 0 & 0 & 8 & $\atyp{0,1}$\\
 $(0,1)$ & 0 & 1 & 0 & 0 & 0 & 0 & 0 & 8 & $\atyp{0,1}$\\
 $(1,1)$ & 0 & 2 & 0 & 0 & 0 & 0 & 0 & 16 & $2 \atyp{0,1}$\\
 $(\frac{1}{8},\frac{1}{8})$ & 0 & 1 & 0 & 0 &0 & 0 & 0 & 8 & $\atyp{0,1}$\\
 $(\frac{1}{8},\frac{9}{8})$ & 0 & 1 & 0 & 0 &0 & 0 & 0 & 8 & $\atyp{0,1}$\\
 $(\frac{9}{8},\frac{1}{8})$ & 0 & 1 & 0 & 0 &0 & 0 & 0 & 8 & $\atyp{0,1}$\\
 $(\frac{9}{8},\frac{9}{8})$ & 0 & 1 & 0 & 0 &0 & 0 & 0 & 8 & $\atyp{0,1}$\\
 $(\frac{5}{8},\frac{5}{8})$ & 2 & 0 & 1 & 0 &0 & 0 & 0 & 24 & $\atyp{0,2} \oplus \slPr(0)$\\
 $(\frac{5}{8},\frac{13}{8})$ & 2 & 0 & 1 & 0 &0 & 0 & 0 & 24 & $\atyp{0,2} \oplus \slPr(0)$\\
  $(\frac{13}{8},\frac{5}{8})$ & 2 & 0 & 1 & 0 &0 & 0 & 0 & 24 & $\atyp{0,2} \oplus \slPr(0)$\\
  $(\frac{13}{8},\frac{13}{8})$ & 2 & 0 & 1 & 0 &0 & 0 & 0 & 24 & $\atyp{0,2} \oplus \slPr(0)$\\
 $(\frac{7}{32},\frac{39}{32})$ & 0 & 0 & 0 & 1 &0 & 0 & 0 & 24 & $\atyp{-\frac{1}{2},\frac{3}{2}} \oplus \atyp{\frac{1}{2},\frac{3}{2}}$\\
 $(\frac{39}{32},\frac{7}{32})$ & 0 & 0 & 0 & 1 &0 & 0 & 0 & 24 & $\atyp{-\frac{1}{2},\frac{3}{2}} \oplus \atyp{\frac{1}{2},\frac{3}{2}}$\\
 $(\frac{39}{32},\frac{39}{32})$ & 0 & 0 & 0 & 2 &0 & 0 & 0 & 48 & $2\atyp{-\frac{1}{2},\frac{3}{2}} \oplus 2\atyp{\frac{1}{2},\frac{3}{2}}$\\
  $(\frac{35}{24},\frac{35}{24})$ & 0 & 0 & 0 & 0 &1 & 0 & 0 & 112 & eq.~\eqref{eqsl21_3leg} \\
  $(2,0)$ & 2 & 1 & 1 & 0 &0 & 0 & 0 & 32 & $\atyp{0,1} \oplus \atyp{0,2} \oplus \slPr(0)$\\
  $(0,2)$ & 2 & 1 & 1 & 0 &0 & 0 & 0 & 32 & $\atyp{0,1} \oplus \atyp{0,2} \oplus \slPr(0)$\\
  $(2,1)$ & 1 & 4 & 2 & 0 &0 & 1 & 0 & 152 & eq.~\eqref{eqsl21_21} \\
  $(1,2)$ & 1 & 4 & 2 & 0 &0 & 0 & 1 & 152 & eq.~\eqref{eqsl21_21}\\
  $(2,2)$ & 6 & 6 & 6 & 0 &0 & 1 & 1 & 336 & eq.~\eqref{eqsl21_22}\\

   \hline

\end{tabular}
\end{center}
\caption{Operator content up to $(h,\bar{h})=(2,2)$. We show the multiplicities of each fields in the various tilting modules, together with the total multiplicity with which they appear in the Hilbert space. The way they transform with respect to $s\ell(2|1)$ is also given in the last column.  }
  \label{tab_opcontent}
\end{table}

Before turning to the indecomposable structure of the full LCFT, we describe the operator content of our theory up to level $(h,\bar{h})=(2,2)$~\cite{ReadSaleur01}, and analyze the multiplicities with respect to the $s\ell(2|1)$ supersymmetry, using the analysis of section~\ref{secsl21resuts}. The results are gathered in Tab.~\ref{tab_opcontent}.

First of all, the groundstate $(h,\bar{h})=(0,0)$ is non-degenerate, and transforms trivially under $s\ell(2|1)$. We also find $8$ Noether currents $(h,\bar{h})=(1,0)$ living in the adjoint representation $\atyp{0,1}$ of $s\ell(2|1)$, as expected. An important point is that these currents {\em do not} generate an affine Lie superalgebra~\cite{ReadSaleur01}, as in that case, we would get $64$ weight $(1,1)$ states. Instead, the multiplicity of these $(1,1)$ fields turn out to be $16$, and they  form two adjoint representations.

The one-hull operators ${\cal O}_1$, with conformal weights $(h,\bar{h})=(\frac{1}{8},\frac{1}{8})$, form an adjoint representation as well, so they appear with multiplicity $8$. More interesting are the two-hulls operators ${\cal O}_2$ with $(h,\bar{h})=(\frac{5}{8},\frac{5}{8})$ as one of these fields is the logarithmic partner of the energy operator -- the relevant thermal perturbation of our critical theory. These fields appear twice in $\TilJTL{0}{\q^2}$ and once in $\TilJTL{2}{1}$, they therefore transform as $24=\atyp{0,2} \oplus \slPr(0)$. The energy field lives at the bottom of the indecomposable module $\slPr(0)$, it is invariant under $s\ell(2|1)$.

The three-hull operators ${\cal O}_3$ with conformal weights $(h,\bar{h})=(\frac{35}{24},\frac{35}{24})$ appear once in the standard module $\StJTL{3}{1}$ -- or in the tilting module $\TilJTL{3}{1}$. They thus come with a multiplicity $112$ in the full Hilbert space our theory, and transform as
\begin{equation}\label{eqsl21_3leg}
112 = \atyp{0,3} \oplus \atyp{0,2} \oplus \atyp{0,1} \oplus \atyp{\pm\half,\ffrac{3}{2}} \oplus \atyp{\pm1,2} \oplus P(0),
\end{equation}
under $s\ell(2|1)$.

Of particular interest are the fields with conformal weights $(h,\bar{h})=(2,0)$, as they include for example the stress energy tensor $T(z)$, and its logarithmic partner $t(z,\bar{z})$. There are $32$ fields with such conformal weights in the spectrum, and they transform according to $\atyp{0,1} \oplus \atyp{0,2} \oplus \slPr(0)$ under $s\ell(2|1)$. The piece $\atyp{0,1}$ corresponds to descendants of the currents so we will discard them in the following. Meanwhile, the fields $T$ and $t$ live at the bottom and at the top, respectively, of the projective cover $\slPr(0)$
\begin{equation}
   \xymatrix@R=19pt@C=16pt
   {
     &{ \{0\}}\ar@{}|{\substack{}\kern-7pt}[]+<-30pt,20pt>\ar[dr]\ar[dl]&\\
     {\{{1\over 2}\}_-}\ar@{}|{\substack{}\kern-7pt}[]+<-60pt,15pt>\ar[dr]
     &&{\{{1\over 2}\}_+}\ar@{}|{\substack{}\kern-7pt}[]+<50pt,15pt>\ar[dl]\\
     &{ \{0\}}\ar@{}|{\substack{}\kern-7pt}[]+<-60pt,-15pt> &\\
     }
\end{equation}

We also find $152$ fields with conformal weights $(h,\bar{h})=(2,1)$ (resp. $(h,\bar{h})=(1,2)$), living in the modules $\TilJTL{0}{\q^2}$, $\TilJTL{1}{1}$, $\TilJTL{2}{1}$ and $\TilJTL{3}{\q^2}$ (resp. $\TilJTL{3}{\q^{-2}}$). With respect to the $s\ell(2|1)$ supersymmetry, they transform as
\begin{equation}\label{eqsl21_21}
152=4 \atyp{0,1} \oplus 2 \atyp{0,2} \oplus \atyp{\pm \ffrac{1}{2},\ffrac{3}{2}}\oplus \atyp{\pm \ffrac{1}{2},\ffrac{5}{2}} \oplus  \slPr(\ffrac{1}{2})_{\pm}.
\end{equation}
and one can check that the  dimensions match since  $4\times 8 + 2 \times 16 + 2 \times 12+ 2 \times 20 +  2 \times 12=152$. Finally, we find that the fields with conformal weights $(h,\bar{h})=(2,2)$ -- including for example the field $T\bar{T}$ -- transform as
\begin{equation}\label{eqsl21_22}
336=6 \atyp{0,1} \oplus 6 \atyp{0,2} \oplus 2 \atyp{\pm \ffrac{1}{2},\ffrac{3}{2}}\oplus 2 \atyp{\pm \ffrac{1}{2},\ffrac{5}{2}} \oplus 2 \slPr(0) \oplus 2 \slPr(\ffrac{1}{2})_{\pm},
\end{equation}
where one can check similarly that $6\times 8 + 6 \times 16 + 2 \times 2 \times 12+ 2 \times 2 \times 20 + 2 \times 8 + 2 \times 2 \times 12=336$ indeed.

\section{Content of indecomposable tilting modules: the full LCFT}
In the previous section, we discussed  left-right Virasoro content in the scaling limit of simple JTL modules and we learned that the simples correspond to (a direct sum of) indecomposable Virasoro modules. We also learned that (reducible) quotients $\bStJTL{j}{P}$ of JTL standard modules are filtered in the limit by products (of quotients or submodules) of Feigin--Fuchs modules and they have quite simple Virasoro character expressions~\eqref{Tr-Kac-gen}. To proceed, it  is thus useful to express the structure of tilting modules in the scaling limit in terms of the quotient modules $\bStJTL{j}{P}$.
 We begin with the structure of the  tilting module $\TilJTL{2}{1}$ presented in Fig.~\ref{Tilt21-limit}.
\begin{figure}
\begin{equation*}
   \xymatrix@R=28pt@C=24pt
   {
     {\dots}\ar[d]\ar[dr]\ar[drrrrrrr]
     &&&&&&&&{\dots}\ar[d]\ar[dl]\ar[dlllllll]\\
     {\dots}\ar[d]\ar[dr]&{\bullet}\ar@{}|{\substack{\bStJTL{6}{\q^2}}\kern-7pt}[]+<-35pt,0pt>\ar[d]\ar[dr]\ar[drrrrr]
     &&&&&&{\bullet}\ar@{}|{\substack{\bStJTL{6}{\q^{-2}}}\kern-7pt}[]+<35pt,0pt>\ar[d]\ar[dl]\ar[dlllll]
     &{\dots}\ar[d]\ar[dl]\\
     {\dots}\ar[d]\ar[dr]
     &{\bullet}\ar@{}|{\substack{\bStJTL{7}{1}}\kern-7pt}[]+<-35pt,0pt>\ar[dr]\ar[d]
     &{\square}\ar@{}|{\substack{\bStJTL{4}{1}}\kern-7pt}[]+<-35pt,0pt>\ar[d]\ar[dr]\ar[drrr]&
     &&&{\square}\ar@{}|{\substack{\bStJTL{5}{1}}\kern-7pt}[]+<30pt,0pt>\ar[d]\ar[dl]\ar[dlll]
     &{\bullet}\ar@{}|{\substack{\bStJTL{8}{1}}\kern-7pt}[]+<30pt,0pt>\ar[dl]\ar[d]
     &{\dots}\ar[d]\ar[dl]\\
     {\dots}
     &{\bullet}\ar@{}|{\substack{\bStJTL{9}{\q^2}}\kern-7pt}[]+<-35pt,0pt>\ar[dr]\ar[d]
     &{\square}\ar@{}|{\substack{\bStJTL{6}{\q^2}}\kern-7pt}[]+<-40pt,0pt>\ar[dr]\ar[d]
     &{\circ}\ar@{}|{\substack{\bStJTL{3}{\q^2}}\kern-7pt}[]+<-35pt,0pt>\ar[dr]\ar[d]
     &&{\circ}\ar@{}|{\substack{\bStJTL{3}{\q^{-2}}}\kern-7pt}[]+<30pt,-5pt>\ar[dl]\ar[d]
     &{\square}\ar@{}|{\substack{\bStJTL{6}{\q^{-2}}}\kern-7pt}[]+<35pt,0pt>\ar[dl]\ar[d]
     &{\square}\ar@{}|{\substack{\bStJTL{9}{\q^{-2}}}\kern-7pt}[]+<35pt,0pt>\ar[dl]\ar[d]
     &{\dots}\\
     &{\dots}
     &{\square}\ar@{}|{\substack{\bStJTL{7}{1}}\kern-7pt}[]+<-35pt,0pt>\ar[dr]\ar[d]
     &{\circ}\ar@{}|{\substack{\bStJTL{4}{1}}\kern-7pt}[]+<-35pt,0pt>\ar[dr]\ar[d]
     &{\bullet}\ar@{}|{\substack{\bStJTL{2}{1}}\kern-7pt}[]+<-35pt,5pt>\ar[d]
     &{\circ}\ar@{}|{\substack{\bStJTL{5}{1}}\kern-7pt}[]+<25pt,0pt>\ar[dl]\ar[d]
     &{\square}\ar@{}|{\substack{\bStJTL{8}{1}}\kern-7pt}[]+<35pt,0pt>\ar[dl]\ar[d]
     &{\dots}&\\
     &&{\dots}
     &{\circ}\ar@{}|{\substack{\bStJTL{6}{\q^2}}\kern-7pt}[]+<-35pt,0pt>\ar[dr]\ar[d]
     &{\bullet}\ar@{}|{\substack{\bStJTL{3}{\q^{-2}}}\kern-7pt}[]+<-40pt,5pt>\ar[d]
     &{\circ}\ar@{}|{\substack{\bStJTL{6}{\q^{-2}}}\kern-7pt}[]+<30pt,0pt>\ar[dl]\ar[d]
     &{\dots}&&\\
     &&&{\dots}
     &{\bullet}\ar@{}|{\substack{\bStJTL{5}{1}}\kern-7pt}[]+<-35pt,0pt>\ar[d]
     &{\dots}&&&\\
     &&&&{\dots} &&&&
     }
\end{equation*}
\caption{The structure of the  tilting module $\TilJTL{2}{1}$ in the scaling limit, where the character of each $\bStJTL{j}{P}$ is given in~\eqref{Tr-Kac-gen}.}\label{Tilt21-limit}
\end{figure}
For simplicity and readability of the diagram, we do not show all the arrows --  only the minimum number needed for consistency (for example, there might also be  arrows like $\stackrel{\bStJTL{6}{\q^2}}{\square}\longrightarrow\stackrel{\bStJTL{5}{1}}{\circ}$ and $\stackrel{\bStJTL{8}{1}}{\bullet}\longrightarrow\stackrel{\bStJTL{6}{\q^{2}}}{\square}$, and so on). In other words, one can translate all arrows on the top surface of the ``cube'' in Fig.~\ref{Tilt21-limit} down along the lattice to recover all the arrows for the limit of $\TilJTL{2}{1}$. The structure for $\TilJTL{1}{1}$ is obtained by replacing the central line in the diagram by $\bStJTL{1}{1}\to\bStJTL{3}{\q^{2}}\to\bStJTL{4}{1}\to\dots$.
All the other tilting modules $\TilJTL{j}{P}$ with the relations $(j,P)\pomore(2,1)$ can be similarly presented by the ``cubic'' diagrams.

Finally, the structure of the vacuum tilting module $\TilJTL{0}{\q^2}$ in the limit in terms of $\bStJTL{j}{P}$ looks slightly different and it can be obtained using the cell filtration in Fig.~\ref{fig:tilt-JTL-vac} and using description of each of $\StJTL{j}{P}$ in terms of quotients $\bStJTL{j'}{P'}$. As this vacuum tilting module is very important for applications, we give below a more detailed analysis of its structure in terms of irreducible Virasoro algebra subquotients.

\subsection{The indecomposable tilting module for the stress tensor}\label{sec:indextiltmod}
We now give a more detailed analysis for the vacuum tilting module $\TilJTL{0}{\q^2}$ in the scaling limit $N\to\infty$. The subquotients structure in terms of JTL simples can be obtained in this limit by  continuing the ``Eiffel tower'' diagram from Fig.~\ref{fig:tilt-N18} where the general pattern is quite clear: one should simply continue ``ladders'' on the left and right parts of the diagram without restrictions on $N$.  Of course, each JTL simple in the limit is a direct sum of complicated indecomposables over the left-right Virasoro algebra $\VirN$ at $c=0$ (as we just discussed previously), and it is actually very hard to give a full picture. Nevertheless, to describe the field content for the first few excited levels it is enough to consider only the ``kernel'' part of the vacuum tilting module from Fig.~\ref{fig:tilt-N18}. We depict this part in terms of JTL simples as
\begin{equation}\label{kernel-tilt}
        \xymatrix@R=24pt@C=20pt
   {&{\bullet}\ar@{}|{\substack{\IrJTL{2}{1}}\kern-7pt}[]+<-20pt,20pt>\ar[dl]\ar[d]\ar[dr]&\\
     {\bullet}\ar@{}|{\substack{\IrJTL{3}{\q^2}}\kern-7pt}[]+<-40pt,15pt>\ar[dr]
     &{\circ}\ar@{}|{\substack{\IrJTL{0}{\q^2}}\kern-7pt}[]+<-35pt,0pt>\ar[d]
     &{\bullet}\ar@{}|{\substack{\IrJTL{3}{\q^{-2}}}\kern-7pt}[]+<30pt,15pt>\ar[dl]\\
      &{\circ}\ar@{}|{\substack{\IrJTL{2}{1}}\kern-7pt}[]+<-35pt,0pt>&
           }
\end{equation}
It is important to note that there are no arrows in $\TilJTL{0}{\q^2}$ coming to $\stackrel{\IrJTL{2}{1}}{\bullet}$ and no arrows going from the bottom $\stackrel{\IrJTL{2}{1}}{\circ}$, as we learned from the lattice analysis, for any $N$, and thus it should be true in the limit. We also note that the full field content up to the level $(2,2)$ is contained in this kernel part.

Now, the idea is to use $\VirN$-module diagrams for irreducible JTL subquotients and compose these pieces into a crucial part of the vacuum tilting module -- the part that would contain energy momentum tensors and all their logarithmic partners and all the necessary $\VirN$ suquotients that admit the indecomposability parameters $b=\bar{b}=-5$ measured in~\cite{VGJSletter}. We recall that the $\VirN$ structure of these pieces (scaling limit of JTL simples) was discussed in Sec.~\ref{sec:lrVir-simple}.
Combining the diagrams for subquotients in~\eqref{kernel-tilt} (represented in ~\eqref{VirN-21} and~\eqref{VirN-3q}) as $\VirN$-modules, joining nodes from different pieces\footnote{Doing this one should keep the rule that action of $\mathfrak{vir}$ can connect only a node $(h,\bar{h})$ with $(h',\bar{h})$, that is it should commute with $\overline{\mathfrak{vir}}$, and vice versa for
$\overline{\mathfrak{vir}}$. There are also obvious restrictions on possible values of these $h$ and $h'$ coming from the structure of Verma Virasoro modules.}   we obtain in Fig.~\ref{Vir-tilt} the subquotient structure of a part (of the whole vacuum tilting module) that we call \textit{the physically crucial part}. We empasize that a $\VirN$-module corresponding to this diagram is not a submodule or a quotient of the scaling limit of $\TilJTL{0}{\q^2}$ but it is an (non-irreducible) \textit{self-dual} $\VirN$-subquotient in the full vacuum tilting module. Therefore, in the full picture, there should be additional arrows coming from above to and going out of this diagram. We do not draw them for simplicity.

We also note that Ridout's ``Ockham's razor'' for a non-chiral staggered  $\VirN$-module~\cite{Ridout} with indecomposability parameters $b=\bar{b}=-5$  is only a submodule in Fig.~\ref{Vir-tilt}. This submodule is generated by states from the two nodes $(2,0)$ and $(0,2)$ at the fourth layer if counted from below (second layer counted from the top).

The structure for the ``kernel'' part of the vacuum tilting module involving fields of dimensions up to $(2,2)$ deserves a more detailed analysis. This structure is described diagrammatically on Fig.~\ref{fig:vac-tilt-22}.
We note that at our level of analysis it is hard to state whether this part of the vacuum module is indecomposable or not. Most probably the  rightmost subquotient $(2,2)$ and its descendants like $(2,7)$, {\it etc.}, are decoupling from the vacuum module. After glueing  two indecomposables, the final module can  in principle be a direct sum of non-trivial indecomposables (even infinitely many of them if we talk about Virasoro algebra modules).  Arrows coming in and out of  this node $(2,2)$ are thus depicted in a dotted style. Recall that $T$ and $\bar{T}$ denote states corresponding to chiral and anti-chiral energy momentum tensors while $t$ and $\bar{t}$  are their logarithmic partners, respectively. The state $\psi$ is for the descendant $A \bar{t}=\bar{A}t$, with the operator $A=L_{-2} -\frac{3}{2}L_{-1}^2$.
Following this diagram we thus expect a Jordan cell of rank $2$ for the fields in $(0,2)$ and $(2,0)$, and of rank $3$ for those from $(2,2)$. In particular, the field $T\bar{T}$ should be involved into a Jordan cell of rank~$3$. This will be discussed in more detail in a subsequent paper.

\subsection{Higher rank Jordan cells for $L_0 + \bar{L}_0$}\label{sec:conjrank}
Here, we discuss Jordan cells of ranks higher than $3$. These cells involve fields of higher conformal dimensions from the vacuum tilting module.
Recall first the structure of $\TilJTL{0}{\q^2}$ with irreducible subquotients in Fig.~\ref{fig:tilt-N18} and its continuation to $N\to\infty$. It  has infinite `ladders' on the left and right sides of the corresponding ``Eiffel tower''. Using this diagram and the finite lattice analysis in Sec.~\ref{sec:Jcells} on the structure of Hamiltonian's Jordan cells (see in particular the discussion at the end of Sec.~\ref{sec:N8}), we can conjecture lower bounds (of the maximally\footnote{there might be states of non-integer conformal dimensions involved in Jordan cells of a rank less than the bound.} possible rank in a given subquotient)
for Jordan cells of the Hamiltonian $L_0 + \bar{L}_0$ in the corresponding bulk LCFT. These  bounds are given by the number of  appearance of the subquotients $\IrJTL{j}{P}$ on the left or right part. And note that this number is finite for a fixed $j$ and was already computed from the lattice   (since  it  stabilizes at large enough $N$) at   the end of Sec.~\ref{sec:vac-tilt-gen}.

 We also mentioned above that due to additional degenerations in the Hamiltonian's eigenvalues at $N\to\infty$ we have additional Jordan cells of rank $2$ in the scaling limit of JTL simples, at least for those from the doubly critical class. Actually, using rather natural assumptions  we were able to see these additional Jordan cells directly in the scaling limit for states having integer conformal weights. We also saw in Fig.~\ref{fig:vac-tilt-22} that the two Jordan cells are combined into a Jordan cell of maximum rank $3$ and not $4$. This happens because the  bottom part of the upper cell is at the same (Loewy) layer as the top of the lower Jordan cell and they thus can not be connected by $\VirN$ action.  So, we obtain that  the maximum rank of Jordan cells for states from $\IrrJTL{2}{1}$ is $3$. Similarly, we expect that states from $\IrrJTL{3}{\q^\pm2}$ in $\TilJTL{0}{\q^2}$ are in rank $2$ maximum (and not $1$, as one could expect using only the lattice analysis), states from $\IrrJTL{4}{1}$ are in rank $3$ (and not $2$), states from $\IrrJTL{6}{\q^\pm2}$  are in rank $4$ maximum (and not $3$), etc. In other words, we expect the ranks of the lattice Hamiltonian Jordan cells
 discussed in Sec.~\ref{sec:vac-tilt-gen} (see the results~\eqref{eq:min rank-1} and  below~\eqref{eq:min rank-1}) to be increased by one in the scaling limit, giving rise to the following conjecture on the rank of Jordan cells for $L_0 + \bar{L}_0$ in the vacuum tilting module.

\begin{conj}In the scaling limit of JTL tilting modules $\TilJTL{0}{\q^2}$, states from (limits of) irreducible  $\IrrJTL{j}{P}$ subquotients
should be involved in Jordan cells of the Hamiltonian $L_0 + \bar{L}_0$ of maximum rank given
by the following number:
\begin{equation}\label{conj-conf-rank}
\text{max rank of}\, L_0 + \bar{L}_0\; \text{in}\; \TilJTL{0}{\q^2} =
\begin{cases}
3,& \qquad j=2, P=1,\\
2 \bigl\lceil \frac{j-2}{3}  \bigr\rceil,& \qquad j\bmod 3 =0 \; \text{and}\; P=\q^{\pm2}, j>0\\
2 \bigl\lceil \frac{j-2}{3}\bigr\rceil +1,& \qquad j\bmod 3 = 1 \; \text{or}\; 2 \; \text{and}\; j>2 \; \text{and}\; P=1.\\
\end{cases}
\end{equation}
\end{conj} \label{conjrank}
Note that the statement here is about {\sl maximum} rank because there are states of non-integer conformal dimensions in $\IrrJTL{j}{P}$ subquotients which might be in cells of lower rank. The rank for those states in the vacuum tilting module that have integer conformal dimensions is expected to be given precisely by~\eqref{conj-conf-rank}.

Since  the $\VirN$ characters of the scaling limit of the JTL simples $\IrrJTL{j}{P}$ are known (Sec.~\ref{sec:char-JTLsimples} and  App~\ref{app:simples}), this conjecture gives
possible values of  ranks of Jordan cells involving corresponding fields of conformal weights $(h,\bar{h})$.

Note finally, that similar results on ranks of Jordan cells can be formulated for other tilting modules. For example, states from (limits of) $\IrrJTL{j}{P}$ subquotients in $\TilJTL{2}{1}$, for $j>2$, should be involved in Jordan cells of maximum rank greater given by those  in~\eqref{conj-conf-rank} plus one.

\section{Conclusion}

The first obvious conclusion is that the situation seems more complicated than one may have expected. If the lattice analysis makes the understanding of titling modules possible, the fact that JTL simples do not correspond to direct sums of $\VirN$ simples forces a very delicate discussion, and a proliferation of arrows of   doubtful use. We believe this simply means simply that $\VirN$ is not the proper object to analyze the continuum limit of our spin chain - and probably LCFTs in general. The proper object should be the full scaling limit of the JTL algebra, which contains $\VirN$, but extends it, giving rise to what we called in~\cite{GRS3} the interchiral algebra.  We did discuss this algebra in the case of $\gl(1|1)$, showing that it was generated by the additional inclusion of the field $\Phi_{2,1}\bar{\Phi}_{2,1}$ of weights $h=\bar{h}=1$.
We expect that an analogous interchiral algebra appears in the scaling limit of $\rJTL{}(1)$ represented in the $s\ell(2|1)$ spin chain, and is probably generated by the field  $\Phi_{2,1}\bar{\Phi}_{2,1}$ again. Note however that now this field has non integer dimensions $h=\bar{h}={5\over 8}$. It is likely that   each simple JTL module goes over, in the scaling limit, to a simple module over this interchiral algebra, and that the analysis in terms of these modules simplifies considerably.  We leave the corresponding discussion for a subsequent paper however. Our algebraic results also have interesting physical consequences: for instance, the field $\bar{T}T$ with conformal weights $(2,2)$ is found to lie at the bottom of a Jordan cell of rank $3$. It would be very interesting to understand how this Jordan cell arises in the $c \to 0$ limit of operator product expansions.

We also would like to briefly comment that the LCFT obtained as the scaling limit of the $s\ell(2|1)$ spin chain differs fundamentally from previous proposals at $c=0$: in both \cite{Flohr-Lohmann} and \cite{GabRunW}, the vacuum indeed appears with multiplicity greater than one. Moreover, a detailed analysis of the conformal weights and their multiplicities shows that the operator contents of these proposed theories are not at all compatible with the one we obtained here. It is not clear to us whether there might be other lattice models whose scaling limit would correspond to the theories in these references, of whether the corresponding LCFTs are really consistent\footnote{In fact, the  bulk theory  in \cite{GabRunW} has three ground states and not one as we have. Further, the one ground state which one would like to call the ``vacuum'' because the energy momentum tensor is its descendant is not really an $\mathrm{SL}(2,{\mathbb C})$-invariant vacuum state.}.

It is probably useful to reiterate here that the only modular invariant our theory is associated with is the trivial invariant $Z=1$. See the discussion in Sec. 2 of this paper.

There are certainly many aspects we did not discuss much. Among these is the centralizer and the corresponding bimodule structure. While this played a crucial role in our analysis of the $\gl(1|1)$ spin chain, it turned out to be not so important here since the $s\ell(2|1)$ provides a faithful representation of the JTL algebra, and other tools could then be used. But the nature of the centralizer (the `symmetry' of the theory) in the $s\ell(2|1)$ spin chain and what it becomes in the continuum limit remains to be understood.

We also note that, according to our analysis,  indecomposable tilting module $\TilJTL{j}{P}$ can be  considered in the scaling limit, under the algebra  $\VirN$, as a (complicated) glueing of  an infinite number of Feigin--Fuchs modules.
This is reminiscent of the construction in   \cite{FFHST}, where full LCFTs are obtained as   glueings of many copies of free-boson theories via the introduction of extra zero modes. In general, understanding the ``naturalness" of the (scaling limit of) the tilting modules and how to obtain them in terms of some free field representation would be, we believe, a great progress. This would probably require a thorough numerical examination.

\medskip

\noindent {\bf Acknowledgments:} We thank J.L. Jacobsen, D. Ridout, V. Schomerus, I.Yu. Tipunin and O.~Mathieu for stimulating discussions. This work was supported by the Quantum Materials program of LBNL (R.V), a Marie-Curie IIF fellowship, a Humboldt fellowship and RFBR-grant 13-01-00386 (A.M.G) and the NSF grant no. DMR-1005895 (N.R.). A.M.G. wishes to thank the Simons Center for Geometry and Physics for hospitality in 2013 and the Physics Department of Yale University (2012) where a part of this work was done. R.V. wishes to thank the University of Southern California and H.~Saleur for hospitality and support through the US Department of Energy (grant number DE-FG03-01ER45908). The authors are also grateful to the organizers of the ACFTA program at the Institut Henri Poincar\'e in Paris, where part of this work was undertaken.

\appendix

\section{The Lie superalgebra $s\ell(2|1)$ and some of its representations}
\label{appSl21}

In this appendix, we gather some well-known results about the Lie superalgebra $s\ell(2|1)$.
In particular, we give an explicit Fock space formulation of the fundamental and dual representations
used to construct our Temperley-Lieb spin chain. We also recall some properties of several finite-dimensional
representations that shall be used throughout this paper. We follow here~\cite{sl21rep,superDic}.

\subsection{The Lie superalgebra $s\ell(2|1)$}

We define the Lie superalgebra $\mathfrak{g} = s\ell(2|1)$ by the commutation relations of its 8 generators. Its bosonic part is  $\mathfrak{g} ^{0} = \mathfrak{u}(1) \oplus \mathfrak{sl}(2)$, that is

\begin{equation}
\displaystyle{ [B, Q^{\pm}] = [B, Q^{z}] = 0, }
\end{equation}
\begin{equation}
\displaystyle{ [Q^{+}, Q^{-}] = 2 Q^{z}, \ \ \ \ [Q^{z}, Q^{\pm}] = \pm Q^{\pm}  }.
\end{equation}
The fermionic generators obey the simple relations
\begin{equation}
\displaystyle{ \{ F^{\pm} , F^{\mp} \} =  \{ \bar{F}^{\pm} , \bar{F}^{\mp} \}  = 0, }
\end{equation}
\begin{equation}
\displaystyle{ \{ F^{\pm} , \bar{F}^{\pm} \} = Q^{\pm},  \ \ \ \ \{ F^{\pm} , \bar{F}^{\mp} \}  = B \mp Q^{z} . }
\end{equation}
Finally, we have
\begin{equation}
\displaystyle{ [Q^z, F^{\pm}] = \pm \frac{1}{2} F^{\pm}, \ \ \ \  [Q^z, \bar{F}^{\pm}] = \pm \frac{1}{2} \bar{F}^{\pm},  }
\end{equation}
\begin{equation}
\displaystyle{ [B, F^{\pm}] = \frac{1}{2} F^{\pm}, \ \ \ \  [B, \bar{F}^{\pm}] = - \frac{1}{2} \bar{F}^{\pm},  }
\end{equation}
\begin{equation}
\displaystyle{ [Q^{\pm}, F^{\pm}] = [Q^{\pm}, \bar{F}^{\pm}] =0, [Q^{\pm}, F^{\mp}] = - F^{\pm}, [Q^{\pm}, \bar{F}^{\mp}] = \bar{F}^{\pm}. }
\end{equation}
Note that there is a subalgebra $\gl(1|1)$ spanned by the generators $\Psi^+ = F^+, \Psi^- = F^-, E = B - Q_z$ and $N = B + Q_z$.

\subsection{Fundamental and Dual representations in Fock space}

Three-dimensional representations of this superalgebra are readily obtained using creation and annihilation operators. To construct what we will refer to as fundamental representation $\Box$, we introduce two boson operators $[ b_{\sigma}, b^\dag_{\sigma'} ] = \delta_{\sigma,\sigma'}$, where $\sigma \in \{ \uparrow, \downarrow \}$, and one fermion $\{ f, f^\dag \} = 1$. The generators read\footnote{We use the same notation for the generators and their representation in the Fock space.}
\begin{equation}
\displaystyle{ B = f^\dag f + \frac{1}{2} (b^\dag_\uparrow b_\uparrow + b^\dag_\downarrow b_\downarrow),
\quad Q_z = \frac{1}{2} (b^\dag_\uparrow b_\uparrow - b^\dag_\downarrow b_\downarrow),
\quad  Q^+ = b^\dag_\uparrow b_\downarrow, \quad Q^- = b^\dag_\downarrow b_\uparrow, }
\end{equation}
\begin{equation}
\displaystyle{ F^+ = f^\dag b_\uparrow, \quad  F^- = f^\dag b_\downarrow, \quad \bar{F}^+ = b_\downarrow^\dag f, \quad \bar{F}^- = b^\dag_\uparrow f. }
\end{equation}
These generators furnish a representation of $ s\ell(2|1)$ in the space $\mathbb C^{2|1} \simeq \square \equiv \mathrm{Span} \{b_\uparrow^\dag \Ket{0}, b_\downarrow^\dag \Ket{0} , f^\dag \Ket{0}\}$. One can also construct the so-called dual representation $\bar{\square} \equiv \mathrm{Span} \{\bar{b}_\uparrow^\dag \Ket{0}, \bar{b}_\downarrow^\dag \Ket{0}, \bar{f}^\dag \Ket{0} \}$, where $[ \bar{b}_{\sigma}, \bar{b}^\dag_{\sigma'} ] = \delta_{\sigma,\sigma'}$ and $\{ \bar{f}, \bar{f}^\dag \} = -1$. The generators act as
\begin{equation}
\displaystyle{ B = \bar{f}^\dag \bar{f} - \frac{1}{2} (\bar{b}^\dag_\uparrow \bar{b}_\uparrow + \bar{b}^\dag_\downarrow \bar{b}_\downarrow), \quad Q_z = \frac{1}{2} (\bar{b}^\dag_\uparrow \bar{b}_\uparrow - \bar{b}^\dag_\downarrow \bar{b}_\downarrow), \quad Q^+ = - \bar{b}^\dag_\uparrow \bar{b}_\downarrow,  \quad Q^- = - \bar{b}^\dag_\downarrow \bar{b}_\uparrow, }
\end{equation}
\begin{equation}
\displaystyle{ F^+ = \bar{b}_\uparrow^\dag \bar{f} , \quad F^- = \bar{b}_\downarrow^\dag \bar{f}, \quad \bar{F}^+ = \bar{f}^\dag \bar{b}_\downarrow , \quad \bar{F}^- = \bar{f}^\dag \bar{b}_\uparrow. }
\end{equation}
Note also that the operator
\begin{equation}
\displaystyle{ e = (\bar{b}_\uparrow^\dag b_\downarrow^\dag + \bar{b}_\downarrow^\dag b_\uparrow^\dag + \bar{f}^\dag f^\dag) (b_\uparrow \bar{b}_\downarrow + b_\downarrow \bar{b}_\uparrow + f \bar{f}) }
\end{equation}
is the projector onto the singlet in the tensor product
representation $ \square \otimes \bar{\square} $. This expression will be used as a definition of the Temperley-Lieb generator.

\subsection{Finite dimensional representations}

We also recall some usual notations for the finite dimensional representations of $s\ell(2|1)$.
We begin with the irreducible representations.
Except for the trivial representation $\lbrace 0 \rbrace$ of dimension $1$,
the irreducible atypical representations $\lbrace j \rbrace_\pm$ are labeled
by the half-integer $j=\frac{1}{2},1, \dots$;
they have dimension $4j+1$. There are also typical representations $\lbrace b, j \rbrace$
(with dimension $8j$ and $b \neq \pm j$)
where $b$ is a $U(1)$ charge, they are also projective. When $b=\pm j$, the modules
$\lbrace \pm j, j \rbrace$ become indecomposable.
Using these notations, the fundamental and dual representations are
$\square=\lbrace \frac{1}{2} \rbrace_{+}$ and $\bar{\square}=\lbrace \frac{1}{2}\rbrace_{-}$;
and the adjoint representation is $\lbrace 0,1 \rbrace$.

We will also be interested in atypical projective covers $\slPr_{\pm}(j)$ (with dimension
$16 j +4$ for $j \neq 0$ and dimension $8$ if $j=0$).
The projective covers $\slPr_{\pm}(j)$ have the following subquotient structure (left diagram for $j \neq 0$,
whereas $\slPr(0)$ is given by the right diagram):
 \begin{equation*}
    \xymatrix@R=16pt@C=12pt
   {
     &{ \{j\}_{\pm}}\ar@{}|{\substack{}\kern-7pt}[]+<-30pt,20pt>\ar[dr]\ar[dl]&\\
     {\{j-\frac{1}{2}\}_{\pm}}\ar@{}|{\substack{}\kern-7pt}[]+<-60pt,15pt>\ar[dr]
     &&{\{j+\frac{1}{2}\}_{\pm}}\ar@{}|{\substack{}\kern-7pt}[]+<50pt,15pt>\ar[dl]\\
     &{ \{j\}_{\pm}}\ar@{}|{\substack{}\kern-7pt}[]+<-60pt,-15pt> &\\
     }
\qquad
   \xymatrix@R=16pt@C=15pt
   {
     &{ \{0\}}\ar@{}|{\substack{}\kern-7pt}[]+<-30pt,20pt>\ar[dr]\ar[dl]&\\
     {\{{1\over 2}\}_-}\ar@{}|{\substack{}\kern-7pt}[]+<-60pt,15pt>\ar[dr]
     &&{\{{1\over 2}\}_+}\ar@{}|{\substack{}\kern-7pt}[]+<50pt,15pt>\ar[dl]\\
     &{ \{0\}}\ar@{}|{\substack{}\kern-7pt}[]+<-60pt,-15pt> &\\
     }
\end{equation*}
The arrows represent here the action of the generators of $s\ell(2|1)$.
We shall not describe the tensor product of all these representations here, and
refer the interested reader to Ref.~\cite{sl21rep}. Using those results, one
can decompose the Hilbert space $\mathcal{H} = (\square \otimes \bar{\square})^{\otimes N}$
onto projective representations only (except for the fundamental).
In particular, we have $\square \otimes \bar{\square}
= \lbrace 0 \rbrace \oplus \lbrace 0,1 \rbrace$.

\section{Faithfulness of $\rJTL{N}(m)$ representations on super-spin chains}\label{sec:TL-faith}

We discuss in this appendix the more general case of periodic spin-chains with $s\ell(m+n|n)$ symmetry.
The  $\sll(m+n|n)$ super-spin chain~\cite{ReadSaleur07-1} is the tensor product $\chVv=\tensor_{j=1}^{N}\chV_j$, with $\chV_j\cong\oC^{m+n|n}$, which consists of  $N=2L$ sites
labelled $j=1,\ldots,2L$ with the fundamental representation of $\sll(m+n|n)$ on even sites and its dual on odd sites.
We choose a basis $\langle v_i, 1\leq i\leq m+2n\rangle$ in each
$\chV_j$ such that the range $1\leq i\leq m+n$ corresponds to
the grade-$0$ (bosonic) subspace and the $v_i$ with  $m+n < i\leq
m+2n$ span the fermionic subsapce.

We consider the following representation of $\JTL{2L}(m)$, which we denote by
$\repgl:\JTL{N}(m)\to\Endo_{\oC}(\chVv)$, with the matrix elements
for $e_j$ (for even $j$, at least), with $1\leq c_j,d_j\leq m+2n$,
\begin{equation}\label{ej-rep-def}
\repgl(e_j)_{c_1\dots c_j c_{j+1}\dots c_N}^{d_1\dots d_j
  d_{j+1}\dots d_N} =
  (-1)^{|c_j|}\delta(c_j,c_{j+1})\delta(d_j,d_{j+1})\prod_{k\ne
  j,j+1}\delta(c_k,d_k), \qquad 1\leq j\leq N,
\end{equation}
and with matrix elements for the translation operator $u^2$
\begin{equation}\label{transl-rep}
\repgl(u^2)_{c_1\dots c_j c_{j+1}\dots c_N}^{d_1\dots d_j
  d_{j+1}\dots d_N} =
  (-1)^{(|c_{N-1}|+|c_N|)\sum_{k=1}^{N-2}|c_k|}\delta(c_{N-1},d_{1})\delta(c_{N},d_{2})\prod_{k=1}^{N-2}\delta(c_k,d_{k+2}).
\end{equation}
Here, we use the basis in $\chVv$ spanned by the monoms $v_{c_1}\tensor
v_{c_2}\tensor\dots\tensor v_{c_N}$, with $v_{c_j}$ being the basis in each
$\chV_{j}\cong\oC^{m+n|n}$. The representation $\repgl$ is equivalent
to the one in~\cite{ReadSaleur07-1} for any $m,n\geq 0$. In particular, for $m=1$, the equation \eqref{ej-rep-def} gives the representation defined in \eqref{ej-1} and \eqref{ej-2}.

Obviously, the representation $\repgl$ defines a homomorphism of $\JTL{N}(m)$ to (a representation
of) the Brauer algebra. The sector of affine diagrams with $N$ through lines
is spanned by $u^{2j}$ (with the relation $u^N=1$) and the
corresponding image in the Brauer algebra is spanned by `planar'
diagrams, where $2j$ through-lines on the right (of the fundamental rectangle) are going from the right to
the left-top crossing the other lines and thus picking up a sign
factor $(-1)^{\sum_{l=0}^{2j-1}|c_{N-l}|\sum_{k=1}^{N-2j}|c_k|}$ as in~\eqref{transl-rep}.
Modulo multiplication by appropriate even powers of $u$, linearly independent affine diagrams in the ideal of $\JTL{N}(m)$ with number of through lines less than
$N$ contain only arcs and through lines connecting the inner boundary
with the outer one and without `winding' around the annulus (this
guarantees that the corresponding diagram in the Brauer algebra have no intersecting
through lines).  Therefore, any such
diagram, sandwiched between appropriate even powers of $u$, is represented by a matrix with elements obtained by the following rules:
\begin{enumerate}
\item
each through line connecting the $j$th site on the inner boundary with
the $k$th site on the outer is replaced by the multiplier
$\delta(c_j,d_k)$;
\item each arc
connecting the $j$th and $k$th sites on the inner boundary is replaced by the multiplier
$(-1)^{|c_j|}\delta(c_j,c_k)$;
\item each arc
connecting the $j$th and $k$th sites on the outer boundary  is replaced by
$\delta(d_j,d_k)$.
\end{enumerate}

The representation $\repgl$ is therefore non-faithful by trivial reasons --
the kernel of $\repgl$ contains affine diagrams in the ideal $J_0$
(without through lines) that are not planar, i.e. the image $\repgl(J_0)$ has the dimension
$\bigr(\binom{2L}{L}-\binom{2L}{L-1}\bigl)^2$ and is generated by (the
representation of) the diagrams that are in bijection with
the usual Temperley--Lieb diagrams without through lines (which
represented faithfully~\cite{ReadSaleur07-1} for $n\geq1$ or $m>1$ and $n=0$).
We next give a proof that non-planar diagrams in $J_0\subset \JTL{N}(m)$ exhaust the kernel of
$\repgl$, for $m,n>0$ or $n>1$. The case $m=0$ and $n=1$ is known to be
highly non-faithful.

\begin{Thm}
For $m,n>0$, the representation $\repgl:\JTL{N}(m)\to\Endo_{\oC}(\chVv)$ defined
in~\eqref{ej-rep-def}-\eqref{transl-rep} has  kernel spanned by
non-planar affine diagrams with zero through-lines.
\end{Thm}
\begin{proof}
Our proof is essentially an adaptation  of
the original  proof of V. Jones ~\cite{Jones} in the bosonic case ($n=0$) to the super-symmetric case.

We have already shown above that the kernel of $\repgl$ contains
the non-planar diagrams\footnote{actually, it contains a linear
  combination of non-planar and planar diagrams.} from the ideal $J_0\subset\JTL{N}(m)$ and
$\repgl\bigl(J_0\cap\TL{N}(m)\bigr)$ has no additional linear relations among
the planar diagrams due to the faithfullness of $\repgl$ restricted to
the subalgebra $\TL{N}(m)$. The non-planar diagrams in the kernel of $\repgl$ span an ideal in
$J_0$ which we denote as $J_0^n$.
In the following, we will thus consider only diagrams $\mu$ from the
quotient $\rJTL{}=\JTL{N}(m)/J_0^n$ (we use the simplified notation for
the algebra $\rJTL{N}(m)$) by this ideal.

Suppose the $\repgl$ has a non-trivial kernel in $\rJTL{}$. It means
there are linear relations in $\Endo_{\oC}(\chVv)$ among the basis
diagrams $\mu\in \rJTL{}$, i.e. there exists a linear combination
\begin{equation}\label{A-def}
A=\sum_{\mu\in \rJTL{}}a_{\mu}\repgl(\mu)\equiv0,
\end{equation}
 where some
$a_{\mu}\in\oC$ are non-zero for a diagram with
through-lines. Consider a diagram $\mu$ (which we denote as $\mu_{j,k}$) with non-zero
$a_{\mu}$ and containing a through-line
connecting, say, the $j$th site on the inner boundary of the annulus with
the $k$th site on the outer.  Matrix
elements of this diagram are given either by the
ones for $\repgl(u^{2r}$), with an integer $r=|j-k|/2$, in~\eqref{transl-rep}  or, following the rules $1.$-$3.$ given above,
\begin{equation}\label{mat-mu-jk}
\repgl(\mu_{j,k})_{\dots c_j \dots c_N}^{\dots d_k
 \dots d_N} =
 \delta(c_j,d_k)\prod_{\{(j',k')\}}\delta(c_{j'},d_{k'})\prod_{\{(i,l)\}}
 \prod_{\{(i',l')\}}(-1)^{|c_i|}\delta(c_i,c_l) \delta(d_{i'},d_{l'}),
\end{equation}
for appropriate pairs $\{(j',k')\}$, with $j'\ne j$ and $k'\ne k$,
corresponding to end points of all the other through lines, and pairs
$\{(i,l)\}$, $\{(i',l')\}$ corresponding to end points of arcs at the
inner and outer boundaries, respectively.

Consider a projector $P_j$ from
$\Endo_{\oC}(\chVv)$ with the kernel of $P_j$ being the bosonic
(grade-$0$) subspace in the $j$th tensorand $V_j$ and the fermionic
(grade-$1$) subspace in each $V_{j'}$ with $j'\ne j$ (the image of
$P_j$ is thus spanned by $v_{c_1}\tensor v_{c_2}\tensor\dots\tensor
v_{c_j}\tensor\dots\tensor v_{c_N}$ with $m+n+1\leq c_j\leq m+2n$ and
$1\leq c_{j'}\leq m+n$ for $j'\ne j$). Then, we have the equalities
(we recall  the definition~\eqref{A-def})
\begin{equation}\label{eq:PAP}
0=P_k A P_j = \sum_{\mu\in \rJTL{}}a_{\mu}P_k \repgl(\mu)P_j
=\sum_{\mu_{j,k}} a_{\mu_{j,k}} P_k \repgl(\mu_{j,k})P_j
\end{equation}
while
\begin{equation}\label{eq:PAP-ne}
P_k \repgl(\mu_{j,k})P_j \ne 0,
\end{equation}
 where the last sum in~\eqref{eq:PAP} is taken over all possible
affine diagrams having a through-line connecting the $j$th site with
the $k$th site. The last equality in~\eqref{eq:PAP}
and inequality~\eqref{eq:PAP-ne} follow from a direct simple calculation
using~\eqref{mat-mu-jk}.

We note that some of the $a_{\mu_{j,k}}$ in the last sum
in~\eqref{eq:PAP} are non-zero by our construction and the set of
diagrams $\{\mu_{j,k}\}$ is in a bijection with planar diagrams from
$\TL{N-1}$; the bijection is obtained by cutting the annulus along the
through line connecting the $j$th site with the $k$th one. We thus
obtain that $P_k \repgl(\mu_{j,k})P_j$ (after multiplying it with translation operators by $j$ and
$k$ sites on the right and left, respectively) represents a planar
diagram
from $\TL{N-1}(m+n)$ on the ($n$ copies of) $s\ell(m+n)$-invariant
spin-chain of length $N-1$.  This representation is well-known to be
faithful~\cite{Jones}, in contradiction with ~\eqref{eq:PAP}. This finishes
the proof.
\end{proof}

\begin{rem}
As a corrollary to the previous theorem we conclude that the
representation $\repgl$ can be restricted to the quotient-algebra
$\rJTL{N}(m)$ introduced before~\eqref{dim-rJTL} and this
representation is faithful on the alternating super-spin chains with
$s\ell(m+n|n)$ symmetry for $m,n>0$. Faithfullness in the case $m=0$
and $n>1$ can be proved in a similar way. The orthogonal projectors $P_j$ can be
chosen such that they project onto a
one-dimensional subspace in the tensorand $V_j$ and onto the complementary subspace in each
tensorand $V_{j'}$, with $j'\ne j$.
\end{rem}

\section{Virasoro content of JTL simples}\label{app:simples}
\label{AppCharacters}

In this appendix, we gather some character formulas that give the Virasoro content of the JTL simple modules $\IrrJTL{j}{P}$.
We focus first on the case that corresponds to $(j,e^{2iK})=(0\hbox{ mod }3,1)$ that is, a spin multiple of three and no twist. In this case, the structure of submodules is much simpler, and implies a single ladder given on the right of Fig.~\ref{fig:cell-thirdroot}.
For instance, consider
\begin{equation}
F_{3,1}^{(0)}=F_{3,1}-F_{6,1}=\sum_{r=1}^\infty K_{r,3}\bar{K}_{r,9}
\end{equation}
We now have the following identities
\begin{eqnarray}
K_{2k,3}&=&\chi_{2,3k}\nonumber\\
K_{2k-1,3}&=&\chi_{1,3k-3}+\chi_{1,3k}\nonumber\\
K_{2k,9}&=&\chi_{2,3k-6}+\chi_{2,3k}+\chi_{2,3k+6}\nonumber\\
K_{2k-1,9}&=&\chi_{1,3k-9}+\chi_{1,3k-6}+\chi_{1,3k-3}+\chi_{1,3k}+\chi_{1,3k+3}+\chi_{1,3k+6}
\end{eqnarray}
and thus
\begin{eqnarray}
F_{3,1}^{(0)}=\chi_{13}\left(\bar{\chi}_{16}+2\bar{\chi}_{19}+\bar{\chi}_{1,12}\right)+
\chi_{1,6}\left(\bar{\chi}_{13}+2\bar{\chi}_{16}+2\bar{\chi}_{1,9}+2\bar{\chi}_{1,12}+\bar{\chi}_{1,15}\right)
\nonumber\\
+\sum_{k=3}^\infty \chi_{1,3k}\left(\bar{\chi}_{1,3k-9}+2\bar{\chi}_{1,3k-6}+2\bar{\chi}_{1,3k-3}+\ldots+
2\bar{\chi}_{1,3k+6}+\bar{\chi}_{1,3k+9}\right)\nonumber\\
+\chi_{2,3}\bar{\chi}_{2,9}+\chi_{2,6}\left(\bar{\chi}_{26}+\bar{\chi}_{2,12}\right)+
\sum_{k=3}^\infty \chi_{2,3k}\left(\bar{\chi}_{2,3k-6}+\bar{\chi}_{2,3k}+\bar{\chi}_{2,3k+6}\right).
\end{eqnarray}
Note that we have the symmetries
\begin{equation*}
h_{1,3k}=h_{2k+1,3},\qquad\quad
h_{2,3k}=h_{2k,3},
\end{equation*}
so the conformal weights here fill up the third row of the Kac table in Fig.~\ref{Kac-t}.
Similarly, we find
\begin{eqnarray}
F_{3,e^{2i\pi/3}}^{(0)}&=&
\chi_{11}\bar{\chi}_{52}+\chi_{31}\left(\bar{\chi}_{32}+2\bar{\chi}_{52}+\bar{\chi}_{72}\right)+\sum_{k=1}^\infty \chi_{2k,1}\left(\bar{\chi}_{2k-2,2}+\bar{\chi}_{2k+2,2}\right)\nonumber\\
&+&\sum_{k=0}^\infty \chi_{2k+5,1}\left(\bar{\chi}_{2k+1,2}+2\bar{\chi}_{2k+3,2}+2\bar{\chi}_{2k+5,2}+2\bar{\chi}_{2k+7,2}+\bar{\chi}_{2k+9,2}\right)\nonumber\\
&=&\chi_{51}\bar{\chi}_{12}+\left(\chi_{31}+2\chi_{51}+\chi_{71}\right)\bar{\chi}_{32}+\sum_{k=1}^\infty \left(\chi_{2k-2,1}+\chi_{2k+2,1}\right)\bar{\chi}_{2k,2}
\nonumber\\
&+&\sum_{k=0}^\infty \left(\chi_{2k+1,1}+2\chi_{2k+3,1}+2\chi_{2k+5,1}+2\chi_{2k+7,1}+\chi_{2k+9,1}\right)\bar{\chi}_{2k+5,2}
\label{F3q-0}
\end{eqnarray}
and
\begin{eqnarray}
F_{3,e^{4i\pi/3}}^{(0)}&=&
\chi_{12}\bar{\chi}_{51}+\chi_{32}\left(\bar{\chi}_{31}+2\bar{\chi}_{51}+\bar{\chi}_{71}\right)
+\sum_{k=1}^\infty \chi_{2k,2}\left(\bar{\chi}_{2k-2,1}+\bar{\chi}_{2k+2,1}\right)\nonumber\\
&+&\sum_{k=0}^\infty \chi_{2k+5,2}\left(\bar{\chi}_{2k+1,1}+2\bar{\chi}_{2k+3,1}+2\bar{\chi}_{2k+5,1}+2\bar{\chi}_{2k+7,1}+\bar{\chi}_{2k+9,1}\right)
\nonumber\\
&=&\chi_{52}\bar{\chi}_{11}+\left(\chi_{32}+2\chi_{52}+\chi_{72}\right)\bar{\chi}_{31}
+ \sum_{k=1}^\infty \left(\chi_{2k-2,2}+\chi_{2k+2,2}\right)\bar{\chi}_{2k,1}
\nonumber\\
&+& \sum_{k=0}^\infty \left(\chi_{2k+1,2}+2\chi_{2k+3,2}+2\chi_{2k+5,2}+2\chi_{2k+7,2}+\chi_{2k+9,2}\right)\bar{\chi}_{2k+5,1}
\end{eqnarray}
\begin{eqnarray}
F_{4,1}^{(0)}&=&\chi_{12}\bar{\chi}_{7,2}+\chi_{3,2}\left(\bar{\chi}_{5,2}+2\bar{\chi}_{7,2}+\bar{\chi}_{9,2}\right)+\chi_{5,2}\left(\bar{\chi}_{3,2}+2\bar{\chi}_{5,2}+2\bar{\chi}_{7,2}+
2\bar{\chi}_{9,2}+\bar{\chi}_{11,2}\right)\nonumber\\
&+&\sum_{k=0}^\infty \chi_{2k+5,2}\left(\bar{\chi}_{2k+1,2}+2\bar{\chi}_{2k+3,2}+2\bar{\chi}_{2k+5,2}+\ldots
+2\bar{\chi}_{2k+11,2}+\bar{\chi}_{2k+13,2}\right)\nonumber\\
&+&\chi_{22}\bar{\chi}_{62}+\sum_{k=2}^\infty \chi_{2k,2}\left(\bar{\chi}_{2k-4,2}+\bar{\chi}_{2k,2}+\bar{\chi}_{2k+4,2}\right)\nonumber\\
&=&\chi_{7,2}\bar{\chi}_{12}+\left(\chi_{5,2}+2\chi_{7,2}+\chi_{9,2}\right)\bar{\chi}_{3,2}+
\left(\chi_{3,2}+2\chi_{5,2}+2\chi_{7,2}+2\chi_{9,2}+\chi_{11,2}\right)\bar{\chi}_{5,2}\nonumber\\
&+&\sum_{k=0}^\infty\left(\chi_{2k+1,2}+2\chi_{2k+3,2}+2\chi_{2k+5,2}+\ldots
+2\chi_{2k+11,2}+\chi_{2k+13,2}\right) \bar{\chi}_{2k+5,2}\nonumber\\
&+& \chi_{62}\bar{\chi}_{22}+\sum_{k=2}^\infty \left(\chi_{2k-4,2}+\chi_{2k,2}+\chi_{2k+4,2}\right)\bar{\chi}_{2k,2}
\end{eqnarray}
\begin{eqnarray}
F_{5,1}^{(0)}&=&\chi_{11}\bar{\chi}_{7,1}+\chi_{3,1}\left(\bar{\chi}_{5,1}+2\bar{\chi}_{7,1}+\bar{\chi}_{9,1}\right)+\chi_{5,1}\left(\bar{\chi}_{3,1}+2\bar{\chi}_{5,1}+2\bar{\chi}_{7,1}+
2\bar{\chi}_{9,1}+\bar{\chi}_{11,1}\right)\nonumber\\
&+&\sum_{k=0}^\infty \chi_{2k+5,1}\left(\bar{\chi}_{2k+1,1}+2\bar{\chi}_{2k+3,1}+2\bar{\chi}_{2k+5,1}+\ldots
+2\bar{\chi}_{2k+11,1}+\bar{\chi}_{2k+13,1}\right)\nonumber\\
&+&\chi_{21}\bar{\chi}_{61}+\sum_{k=2}^\infty \chi_{2k,1}\left(\bar{\chi}_{2k-4,1}+\bar{\chi}_{2k,1}+\bar{\chi}_{2k+4,1}\right)\nonumber\\
&=&\chi_{7,1}\bar{\chi}_{11}+\left(\chi_{5,1}+2\chi_{7,1}+\chi_{9,1}\right)\bar{\chi}_{3,1}+
\left(\chi_{3,1}+2\chi_{5,1}+2\chi_{7,1}+2\chi_{9,1}+\chi_{11,1}\right)\bar{\chi}_{5,1}\nonumber\\
&+&\sum_{k=0}^\infty\left(\chi_{2k+1,1}+2\chi_{2k+3,1}+2\chi_{2k+5,1}+\ldots
+2\chi_{2k+11,1}+\chi_{2k+13,1}\right) \bar{\chi}_{2k+5,1}\nonumber\\
&+& \chi_{61}\bar{\chi}_{21}+\sum_{k=2}^\infty \left(\chi_{2k-4,1}+\chi_{2k,1}+\chi_{2k+4,1}\right)\bar{\chi}_{2k,1}
\end{eqnarray}
\begin{eqnarray}
F_{6,e^{2i\pi/3}}^{(0)}&=&\chi_{11}\bar{\chi}_{9,2}+\chi_{3,1}\left(\bar{\chi}_{7,2}+2\bar{\chi}_{9,2}+\bar{\chi}_{11,2}\right)+\chi_{5,1}\left(\bar{\chi}_{5,2}+2\bar{\chi}_{7,2}+2\bar{\chi}_{9,2}+
2\bar{\chi}_{11,2}+\bar{\chi}_{13,2}\right)\nonumber\\
&+&\sum_{k=0}^\infty \chi_{2k+7,1}\left(\bar{\chi}_{2k+3,2}+2\bar{\chi}_{2k+5,2}+2\bar{\chi}_{2k+7,2}+\ldots
+2\bar{\chi}_{2k+13,2}+\bar{\chi}_{2k+15,2}\right)\nonumber\\
&+&\chi_{21}\bar{\chi}_{82}+\chi_{41}\left(\bar{\chi}_{62}+\bar{\chi}_{10,2}\right)+\sum_{k=3}^\infty
\chi_{2k,1}\left(\bar{\chi}_{2k-6,2}+\bar{\chi}_{2k-2,2}+\ldots+\bar{\chi}_{2k+6,2}\right)\nonumber\\
&=&\chi_{9,1}\bar{\chi}_{12}+\left(\chi_{7,1}+2\chi_{9,1}+\chi_{11,1}\right)\bar{\chi}_{3,2}+
\left(\chi_{5,1}+2\chi_{7,1}+2\chi_{9,1}+2\chi_{11,1}+\chi_{13,1}\right)\bar{\chi}_{5,2}\nonumber\\
&+&\sum_{k=0}^\infty\left(\chi_{2k+3,1}+2\chi_{2k+5,1}+2\chi_{2k+7,1}+\ldots
+2\chi_{2k+13,1}+\chi_{2k+15,1}\right) \bar{\chi}_{2k+7,2}\nonumber\\
&+&\chi_{81}\bar{\chi}_{22}+\left(\chi_{61}+\chi_{10,1}\right)\bar{\chi}_{42}+\sum_{k=3}^\infty
\left(\chi_{2k-6,1}+\chi_{2k-2,1}+\ldots+\chi_{2k+6,1}\right)\bar{\chi}_{2k,2}
\end{eqnarray}
while $F_{6,e^{4i\pi/3}}^{(0)}$ would be obtained by switching the $1$ and $2$ for the right  character labels.

\section{Free field and vertex operators}\label{appFF}

We recall here basic facts about scalar free fields and Feigin--Fuchs modules. Let $\varphi$ denote a free scalar field with the OPE
\begin{equation}
  \dd\varphi(z)\dd\varphi(w)=\ffrac{1}{(z-w)^2}
\end{equation}
and the mode expansion
\begin{equation}\label{scalar-modes}
  \dd\varphi(z)=\sum_{n\in\oZ}\varphi_n z^{-n-1}.
\end{equation}
The energy--momentum tensor is given by
\begin{equation}\label{eq:the-Virasoro}
  T(z)=\half\,\dd\varphi(z)\dd\varphi(z)
  +\ffrac{\alpha_0}{2}\,\dd^2\varphi(z)
\end{equation}
where we fix two coprime positive integers $p_+$ and $p_-$ and set
\begin{equation}\label{eq:numbers}
  \alpha_-=-\sqrt{\ffrac{2p_+}{p_-}},
  \quad
  \alpha_+=\sqrt{\ffrac{2p_-}{p_+}},\quad
  \alpha_0=\alpha_++\alpha_-.
\end{equation}.

The modes of $\dd\varphi(z)$ span the
Heisenberg algebra and the modes of $T(z)$ span the Virasoro algebra
$\Vir$ with the central charge
\begin{equation}\label{eq:centr-charge}
  c=1-6\ffrac{\bigl(p_+-p_-\bigr)^2}{p_+p_-}.
\end{equation}

The vertex operators are given by $e^{j(r,s)\varphi(z)}$ with $j(r,s)=
\frac{1-s}{2}\,\alpha_- + \frac{1-r}{2}\,\alpha_+$, \ $r,s\in\oZ$.
Equivalently, these vertex operators can be parameterized as
\begin{equation}\label{V-param}
  V_{r,s;n}(z)= e^{\frac{p_-(1-r)-p_+(1-s)+p_+p_-n}{\sqrt{2p_+p_-}}
    \varphi(z)},
  \qquad 1\leq r\leq p_+,\quad 1\leq s\leq p_-,\quad n\,{\in}\,\oZ.
\end{equation}
The conformal dimension of $V_{r,s;n}(z)$ assigned by the
energy--momentum tensor is
\begin{equation}\label{Delta-rs}
  \Delta_{r,s;n}=\ffrac{(p_+ s\!-\!p_- r\!+\!p_+ p_- n)^2
    -(p_+\!-\!p_-)^2}{4p_+p_-}.
\end{equation}
We note that
\begin{equation}\label{id-dim}
  \Delta_{-r,-s;-n}=\Delta_{r,s;n},\qquad
  \Delta_{r+kp_+,s+kp_-;n}=\Delta_{r,s;n},
  \qquad  \Delta_{r,s+kp_-;n}=\Delta_{r,s;n+k}.
\end{equation}

\subsection{Feigin--Fuchs modules}
For $1\leq r\leq p_+$, $1\leq s\leq p_-$, and $n\,{\in}\,\oZ$, let
$\repF_{r,s;n}$ denote the Fock module of the Heisenberg algebra
generated from (the state corresponding to) the vertex operator
$V_{r,s;n}(z)$.  The zero mode~$\varphi_0=\frac{1}{2i\pi}\oint
dz\dd\varphi(z)$ acts in~$\repF_{r,s;n}$ by multiplication with the
number
\begin{equation*}
  \varphi_0\,v=
  \ffrac{p_-(1-r)-p_+(1-s)+p_+p_-n}{
    \sqrt{\mathstrut 2p_+p_-}}\,v,
  \quad
  v\in\repF_{r,s;n}.
\end{equation*}
We write $\repF_{r,s}\equiv\repF_{r,s;0}$.  For convenience of
notation, we identify $\repF_{0,s;n}\equiv\repF_{p_+,s;n+1}$ and
$\repF_{r,0;n}\equiv\repF_{r,p_-;n-1}$.

Let $\rep{Y}_{r,s;n}$ with $1\leq r\leq p_+$, $1\leq s\leq p_-$, and
$n\,{\in}\,\oZ$ denote the Virasoro module that coincides with
$\repF_{r,s;n}$ as a linear space, with the Virasoro algebra action
given by~\eqref{eq:the-Virasoro} (see~\cite{[FF2]}).  As with the
$\repF_{r,s;n}$, we also write $\rep{Y}_{r,s}\equiv\rep{Y}_{r,s;0}$.

 The well-known structure~\cite{[FF]} of $\rep{Y}_{r,s}$ for $1\leq r\leq
p_+\,{-}\,1$ and $1\leq s\leq p_-\,{-}\,1$ is recalled in
Fig.~\ref{fig:embedding}.  We
let $\repJ_{r,s;n}$ denote the irreducible Virasoro module with the
highest weight~$\Delta_{r,s;n}$ (as before, $1\leq r\leq p_+$, $1\leq
s\leq p_-$, and $n\,{\in}\,\oZ$).  Evidently,
$\repJ_{r,s;n}\simeq\repJ_{p_+-r,p_--s;-n}$.  The
$\half(p_+\,{-}\,1)(p_-\,{-}\,1)$ nonisomorphic modules among the
$\repJ_{r,s;0}$ with $1\leq r\leq p_+\,{-}\,1$ and $1\leq s\leq
p_-\,{-}\,1$ are the irreducible modules from the Virasoro $(p_+,p_-)$
minimal model.  We also write $\repJ_{r,s}\equiv\repJ_{r,s;0}$. For
convenience of notation, we identify
$\repJ_{0,s;n}\equiv\repJ_{p_+,s;n+1}$ and
$\repJ_{r,0;n}\equiv\repJ_{r,p_-;n-1}$.

\begin{figure}[tbp]
  \mbox{}\hfill\mbox{}
  \xymatrix
  {{}&\times\ar[dr]\ar@{}|{[r,s;0]}[]+<20pt,20pt>&&\\
    \blacktriangle\ar[ur]\ar[d]\ar[drr]\ar@{}|{\kern-25pt
      \substack{[p_+-r,s;1]}}[]+<-30pt,10pt>
    &&\bullet\ar@{}|{\substack{[r,p_--s;1]\kern-10pt}}[]+<45pt,10pt>\\
    \circ\ar[urr]\ar[drr]\ar@{}|{[r,s;2]}[]+<-45pt,10pt>
    &&\scrBox\ar[u]\ar[d]\ar@{}|{[p_+-r,p_--s;2]}[]+<70pt,10pt>\\
    \blacktriangle\ar[u]\ar[d]\ar[urr]\ar[drr]
    \ar@{}|{\kern-7pt\substack{[p_+-r,s;3]}}[]+<-45pt,10pt>
    &&\bullet\ar@{}|{\substack{[r,p_--s;3]}\kern-7pt}[]+<45pt,10pt>\\
    \circ\ar[urr]\ar@{}|{[r,s;4]}[]+<-45pt,10pt> \ar@{.}[]-<0pt,18pt>
    &&\scrBox\ar[u]\ar@{}|{[p_+-r,p_--s;4]}[]+<70pt,10pt>
    \ar@{.}[]-<0pt,18pt>}
  \mbox{}\hfill\mbox{}
  \caption{Subquotient structure of the
      Feigin--Fuchs module $\rep{Y}_{r,s}$.   The
    notation is as follows.  The cross {\small$\times$} corresponds to
    the subquotient $\repJ_{r,s}$, the filled dots {\small$\bullet$}
    to $\repJ_{r,p_--s;2n+1}$ with $n\in\oN_0$, the triangles
    {\small$\blacktriangle$} to $\repJ_{p_+-r,s;2n+1}$
    with~$n\in\oN_0$, the open dots {\small$\circ$} to
    $\repJ_{r,s;2n}$ with~$n\in\oN$, and the squares {\small$\scrBox$}
    to $\repJ_{p_+-r,p_--s;2n}$ with $n\in\oN$.  The notation $[a,b;n]$ in square brackets is for
    subquotients isomorphic to~$\repJ_{a,b;n}$.  The filled dots
    constitute the socle of $\rep{Y}_{r,s}$.}
  \label{fig:embedding}
\end{figure}
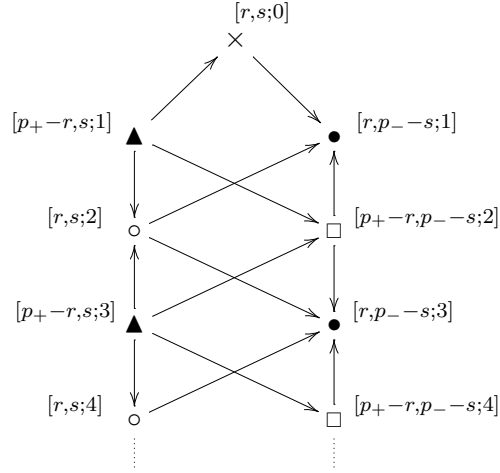

The Fock spaces introduced above constitute a free-field module
\begin{equation}\label{lattice-F}
  \repF
  =\bigoplus_{n\in\oZ}
  \bigoplus_{r=1}^{p_+}\bigoplus_{s=1}^{p_-}\repF_{r,s;n}.
\end{equation}
It can be regarded as (the chiral sector of) the space of states of
the Gaussian Coulomb gas model compactified on the circle of
radius~$\sqrt{2p_+p_-}$. Note that the central charges  we consider correspond to $p_+=x$ and $p_-=x+1$.

\end{document}